\documentclass[12pt,reqno]{article}
\usepackage{amsmath,amsthm,amssymb,caption,subcaption,tikz,intcalc,etoolbox}
\usepackage{graphicx}
\usepackage{bm}
\usepackage{url}
\usepackage{enumerate, enumitem, afterpage, geometry}

\definecolor{edge-color}{rgb}{0.85,0.85,0.85}
\definecolor{edge-on-color}{rgb}{0.8,0.2,0.2}
\definecolor{edge-bnd-color}{rgb}{0.4,0.4,0.9}
\definecolor{domain-path-color}{rgb}{0.4,0.4,0.9}
\tikzstyle edge=[color=edge-color]
\tikzstyle edge-on=[color=edge-on-color, thick]
\tikzstyle edge-off=[color=white,thick,dotted]
\tikzstyle edge-bnd=[color=edge-bnd-color, ultra thick]
\tikzstyle domain-path=[color=domain-path-color, thick, dotted]
\tikzstyle node0=[circle, fill=gray!50]
\tikzstyle node1=[circle, fill]
\tikzstyle node2=[circle, draw]

\def\hexagonEdge[#1][#2][#3][#4]{ 
 \def \h {0.866cm}
 \begin{scope}[xshift=(#1) * 1.5cm, yshift=(#1+2*(#2)) *\h]
	\def \dd { \intcalcMod{\intcalcAdd{#3}{6}}{6} }
	\ifnumequal{\dd}{1}{ \draw [#4] {(0:1) -- (60:1)}; } {}
	\ifnumequal{\dd}{0}{ \draw [#4] {(60:1) -- (120:1)}; } {}
	\ifnumequal{\dd}{5}{ \draw [#4] {(120:1) -- (180:1)}; } {}
	\ifnumequal{\dd}{4}{ \draw [#4] {(180:1) -- (240:1)}; } {}
	\ifnumequal{\dd}{3}{ \draw [#4] {(240:1) -- (300:1)}; } {}
	\ifnumequal{\dd}{2}{ \draw [#4] {(300:1) -- (0:1)}; } {}
 \end{scope}
}

\def\hexagonEdges[#1][#2][#3][#4]{
 	\foreach \x/\y/\d in {#4} {
 		\hexagonEdge[(#2)+\x][(#3)+\y][\d][#1];
}}

\def\hexagonEdgesFullRow[#1][#2][#3][#4]{
 	\foreach \x/\d in {#4} {
 		\hexagonEdge[(#2)+\x][(#3)-\intcalcDiv{\x}{2}][\d][#1];
}}

\def\hexagonFullRow[#1][#2]{
	\foreach \x in {#2} {
 		\hexagon[\x][(#1)-\intcalcDiv{\x}{2}][1];
}}

\def\hexagonFullColumn[#1][#2]{
	\foreach \y in {#2} {
 		\hexagon[#1][\y-\intcalcDiv{#1}{2}][1];
}}

\def\hexagon[#1][#2]{
 \def \h {0.866cm}
 \begin{scope}[xshift=(#1) * 1.5cm, yshift=(#1+2*(#2)) *\h]
    \draw [edge] {(0:1) -- (60:1) -- (120:1) -- (180:1) -- (240:1) -- (300:1) -- cycle};
 \end{scope}
}

\def\hexagonNode[#1][#2][#3][#4]{
 \def \h {0.866cm}
 \begin{scope}[xshift=(#1) * 1.5cm, yshift=(#1+2*(#2)) *\h]
    \def\c{\intcalcMod{(#1)-(#2)}{3}}
	\def\d{#3}
	\ifnumequal{\d}{-1}{}{\ifnumequal{\d}{\c}{}{\def\c{9}}}
    \ifnumequal{\c}{0}{ \node [node0] at (0,0) { }; } {}
    \ifnumequal{\c}{1}{ \node [node1] at (0,0) { }; } {}
    \ifnumequal{\c}{2}{ \node [node2] at (0,0) { };} {}
 \end{scope}
}

\def\hexagonGrid[#1][#2][#3]{
  \foreach \i in {0,...,#1} {
   \foreach \j in {0,...,#2} {
 	\def\y{\j - \intcalcDiv{\i}{2}}
  	\hexagon[\i][\y];
	\ifnumequal{#3}{1}{\hexagonNode[\i][\y][-1][0]}{}
}}}

\def\hexagonGridOneClass[#1][#2][#3][#4]{
\begin{scope}
	\tikzstyle node0=[circle, fill=gray!50, minimum size=5pt,inner sep=0]
	\tikzstyle node1=[circle, fill=gray!50, minimum size=5pt,inner sep=0]
	\tikzstyle node2=[circle, fill=gray!50, minimum size=5pt,inner sep=0]

  \foreach \i in {0,...,#1} {
   \foreach \j in {0,...,#2} {
 	\def\y{\j - \intcalcDiv{\i}{2}}
  	\hexagon[\i][\y];
	\hexagonNode[\i][\y][#3][#4];
}}
\end{scope}}

\def\groundState[#1][#2][#3]{
  \hexagonGrid[#1][#2][#3];
  \foreach \i in {0,...,#1} {
   \foreach \j in {0,...,#2} {
	\ifnumequal{ \intcalcMod{\i- \j + \intcalcDiv{\i}{2}}{3} }{1}{\trivialLoop[\i][\j - \intcalcDiv{\i}{2}]}{}
}}}

\def\trivialLoop[#1][#2]{ \hexagonEdges[edge-on][#1][#2][0/0/0,0/0/1,0/0/2,0/0/3,0/0/4,0/0/5];}
\def\doubleLoopUp[#1][#2]{ \hexagonEdges[edge-on][#1][#2][0/1/-2,0/1/-1,0/1/0,0/1/1,0/1/2,0/0/1,0/0/2,0/0/3,0/0/4,0/0/5]; }
\def\doubleLoopRight[#1][#2]{ \hexagonEdges[edge-on][#1][#2][0/0/0,0/0/5,0/0/4,0/0/3,0/0/2,1/0/3,1/0/2,1/0/1,1/0/0,1/0/-1]; }
\def\doubleLoopLeft[#1][#2]{ \hexagonEdges[edge-on][#1][#2][0/1/0,0/1/1,1/0/0,1/0/1,1/0/2,1/0/3,1/0/4,0/1/3,0/1/4,0/1/5]; }

\def\R{\mathbb{R}}
\def\C{\mathbb{C}}
\def\P{\mathbb{P}}
\def\E{\mathbb{E}}

\def\Z{\mathbb{Z}}
\def\T{\mathbb{T}}

\def\cS{\mathcal{S}}
\def\one{\mathbf{1}}

\renewcommand\c[1]{\emph{#1}}
\renewcommand{\S}{\mathbb{S}}
\DeclareMathOperator{\var}{Var}

\DeclareMathOperator{\Cov}{Cov}

\DeclareMathOperator{\dist}{dist}

\newtheorem{definition}{Definition}[section]
\newtheorem{theorem}{Theorem}[section]
\newtheorem{proposition}[theorem]{Proposition}
\newtheorem{lemma}[theorem]{Lemma}
\newtheorem{remark}[theorem]{Remark}

\newtheorem{fact}[theorem]{Fact}

\newtheorem{cor}[theorem]{Corollary}

\def\eps{\varepsilon}

\renewcommand{\Pr}{\mathbb{P}}

\newcommand{\HH}{\mathbb{H}}
\newcommand{\VH}{V(\HH)}
\newcommand{\EH}{E(\HH)}
\newcommand{\LC}{\mathsf{LoopConf}}
\newcommand{\ground}{\omega_{\operatorname{gnd}}}
\newcommand{\sfT}{\mathsf{T}}
\newcommand{\reffig}[1] {\textsc{\ref{#1}}}

\newcommand{\Int}[1]{\mathrm{Int}(#1)}
\newcommand{\Ext}[1]{\mathrm{Ext}(#1)}
\newcommand{\IntEdge}[1]{\mathrm{Int}^\mathrm{E}(#1)}
\newcommand{\IntVert}[1]{\mathrm{Int}^\mathrm{V}(#1)}
\newcommand{\ExtEdge}[1]{\mathrm{Ext}^\mathrm{E}(#1)}
\newcommand{\ExtVert}[1]{\mathrm{Ext}^\mathrm{V}(#1)}
\newcommand{\IntHex}[1]{\mathrm{Int}^{\mathrm{hex}}(#1)}
\newcommand{\BadEdges}{E^{\operatorname{bad}}}
\newcommand{\BadEdgesBefore}{\overline E}

\newcommand{\clr}{{{\mathsf{c}}}}
\newcommand{\breakup}{{{\cC}}}
\newcommand{\shiftFunc}{R}
\newcommand{\shift}[1]{{\shiftFunc(#1)}}
\newcommand{\DIRup}{\,\uparrow\,}
\newcommand{\DIRdown}{\,\downarrow\,}

\newcommand{\cC}{\mathcal{C}}
\newcommand{\cF}{\mathcal{F}}

\begin{document}

\title{Lectures on the Spin and Loop $O(n)$ Models}
\author{Ron Peled\thanks{School of Mathematical Sciences, Tel Aviv University, Tel Aviv, Israel. Research supported by Israeli Science Foundation grant 861/15 and the European Research Council starting grant 678520 (LocalOrder). The research of Y.S. was also supported by the Adams Fellowship Program of the Israel Academy of Sciences and Humanities. E-mails: {\tt peledron@post.tau.ac.il},\ \ {\tt yinonspi@post.tau.ac.il}.}\and Yinon Spinka\footnotemark[1]}

\date{\small{\today}}

\maketitle

\begin{center}
	\vspace{-20pt}
	\footnotesize{\textit{Dedicated to Chuck Newman on the occasion of his 70th birthday}}
	\vspace{20pt}\\
	\includegraphics[width=8cm, height=7.5cm]{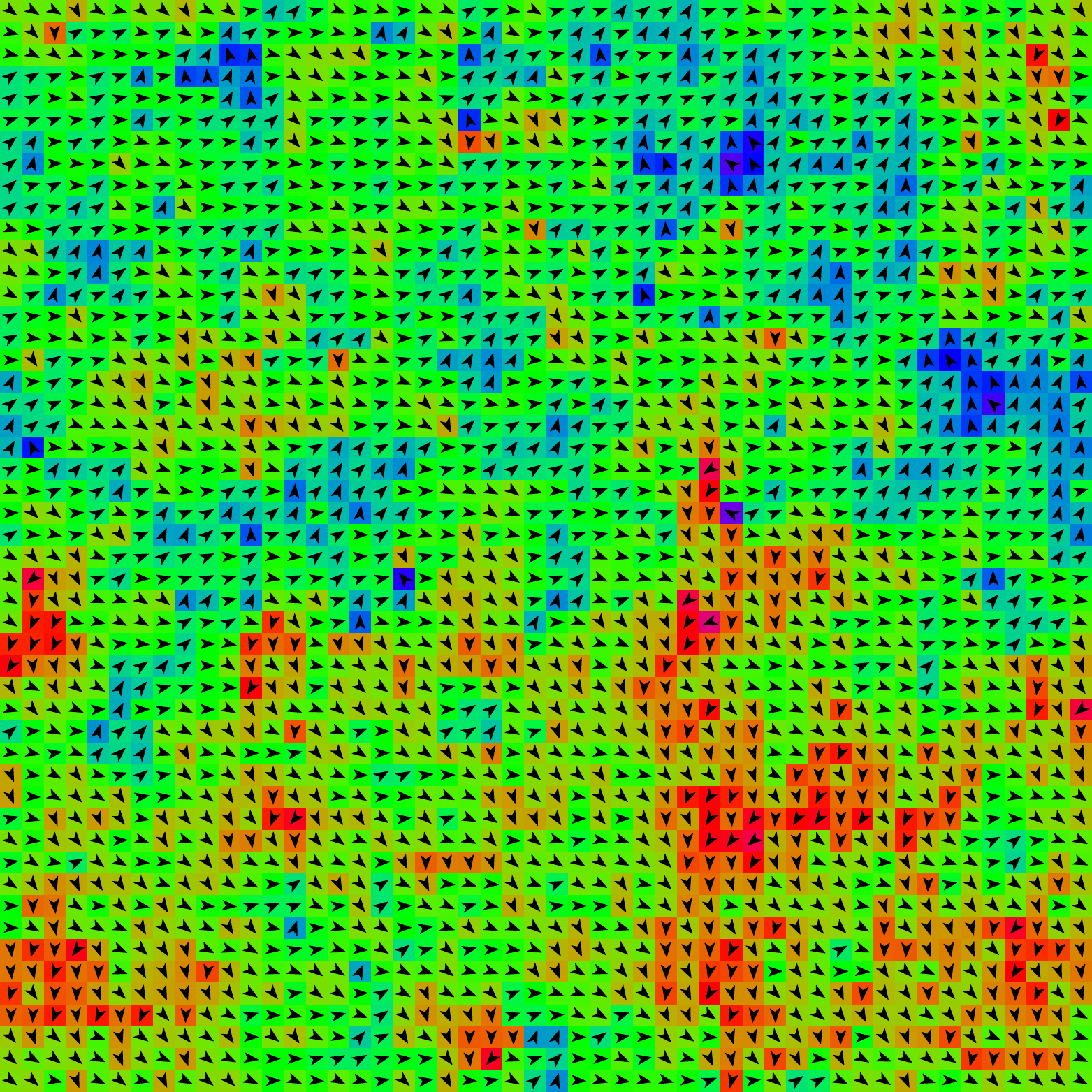}\quad\includegraphics[width=8cm, height=7.5cm]{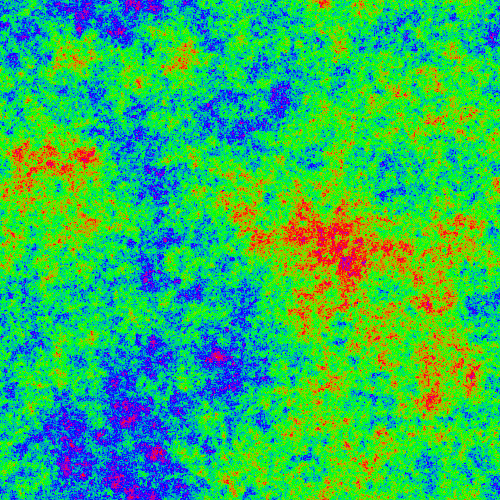}\\\bigbreak
	\includegraphics[width=8cm, height=7.2cm]{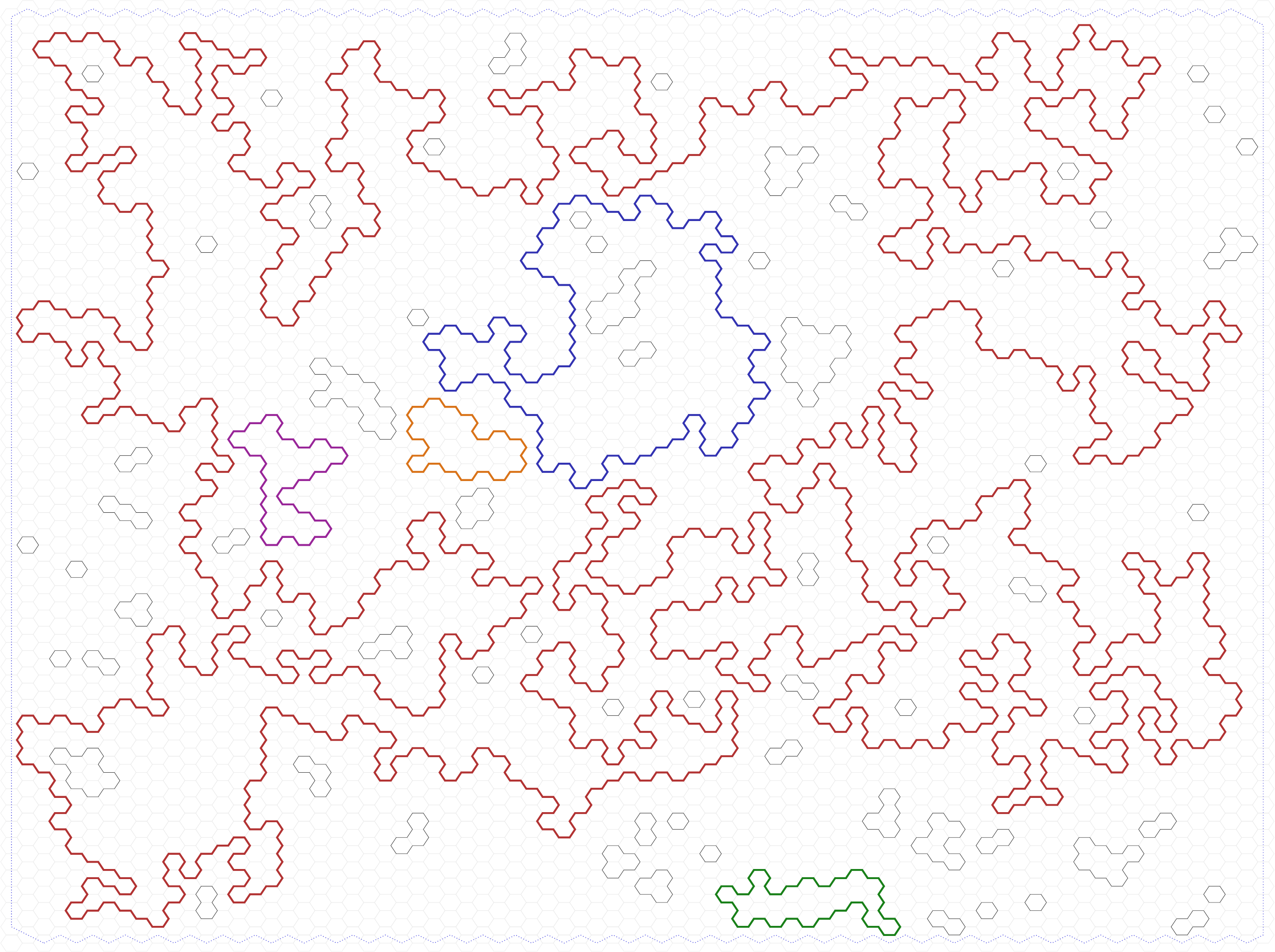}\quad\includegraphics[width=8cm, height=7.2cm]{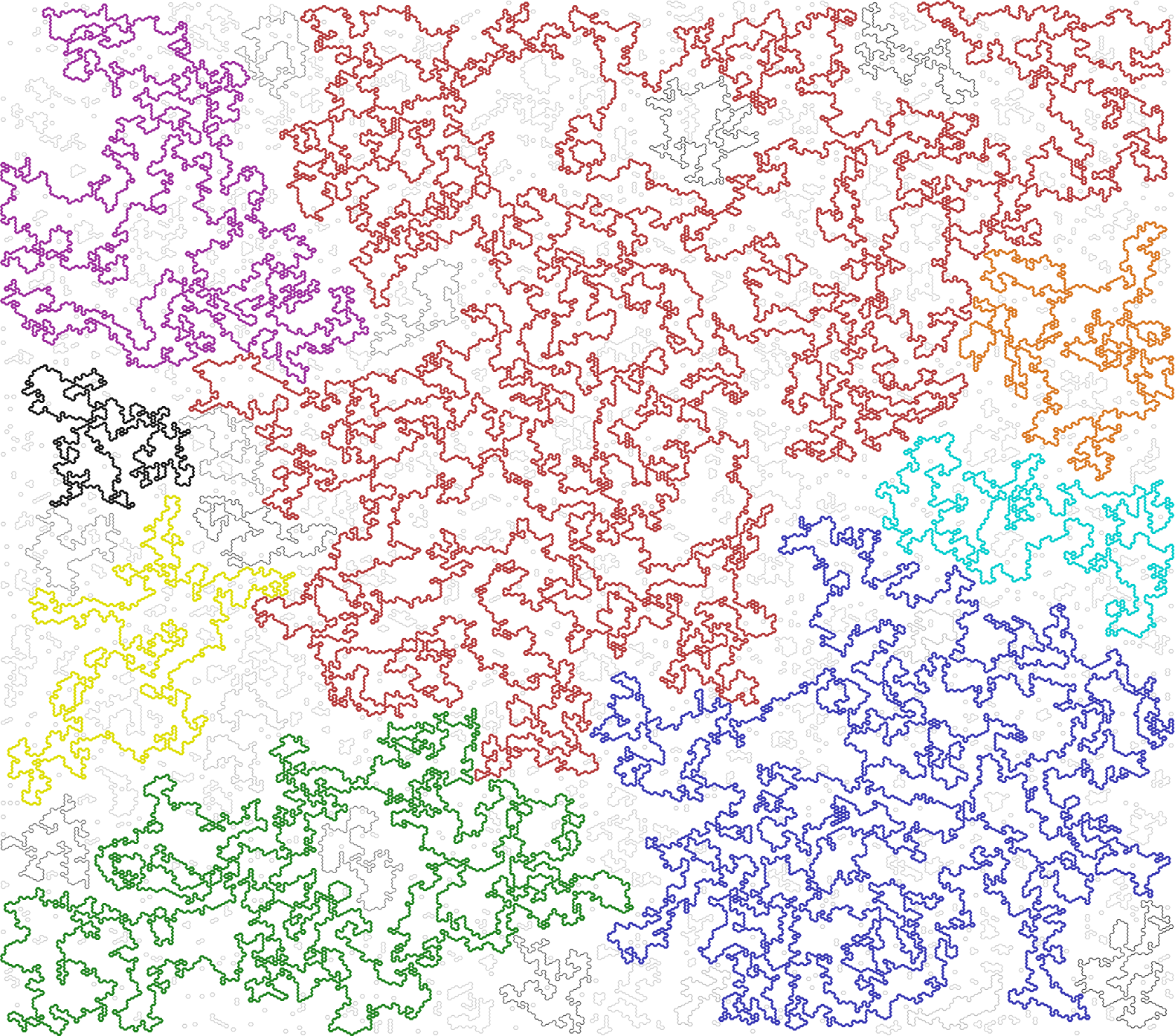}
\end{center}

\newpage

\tableofcontents

\section{Introduction}
The classical spin $O(n)$ model is a model on a $d$-dimensional
lattice in which a vector on the $(n-1)$-dimensional sphere is
assigned to every lattice site and the vectors at adjacent sites
interact ferromagnetically via their inner product. Special cases
include the Ising model ($n=1$), the XY model ($n=2$) and the
Heisenberg model ($n=3$). We discuss questions of long-range order
(spontaneous magnetization) and decay of correlations in the spin
$O(n)$ model for different combinations of the lattice dimension $d$
and the number of spin components~$n$. Among the topics presented
are the Mermin--Wagner theorem, the Berezinskii--Kosterlitz--Thouless
transition, the infra-red bound and Polyakov's conjecture on the
two-dimensional Heisenberg model.

The loop $O(n)$ model is a model for a random configuration of
disjoint loops. In these notes we discuss its properties on the
hexagonal lattice. The model is parameterized by a loop weight $n\ge
0$ and an edge weight $x\ge 0$.
Special cases include self-avoiding walk ($n=0$), the Ising model
($n=1$), critical percolation ($n=x=1$), dimer model ($n=1,
x=\infty$), proper $4$-coloring ($n=2, x=\infty)$, integer-valued
($n=2$) and tree-valued (integer $n>=3$) Lipschitz functions and the
hard hexagon model ($n=\infty$). The object of study in the model is
the typical structure of loops. We will review the connection of the
model with the spin $O(n)$ model and discuss its conjectured phase
diagram, emphasizing the many open problems remaining. We then
elaborate on recent results for the self-avoiding walk case and for
large values of $n$.

The first version of these notes was written for a series of lectures given at the School and
Workshop on Random Interacting Systems at Bath, England in June
2016. The authors are grateful to Vladas Sidoravicius and Alexandre
Stauffer for the organization of the school and for the opportunity
to present this material there. It is a pleasure to thank also the
participants of the meeting for various comments which greatly
enhanced the quality of the notes.

Our discussion is aimed at giving a relatively short and accessible introduction to the topics of the spin $O(n)$ and loop $O(n)$ models. The selection of topics naturally reflects the authors' specific research interests and this is perhaps most noticeable in the sections on the Mermin--Wagner theorem (Section \ref{sec:Mermin-Wagner}), the infra-red bound (Section \ref{sec:infra-red_bound}) and the chapter on the loop $O(n)$ model (Section \ref{sec:loop-model}).
The interested reader may find additional information in the recent books of Friedli and Velenik \cite{friedli2016statistical} and Duminil-Copin \cite{duminil2013parafermionic} and in the lecture notes of Bauerschmidt \cite{Bauerschmidt2016}, Biskup \cite{biskup2009reflection} and Ueltschi \cite{Ueltschi2013}.

\section{The Spin $O(n)$ model}
\label{sec:spin-O-n}

\begin{figure}[!ht]
	\centering
	\begin{subfigure}[t]{.49\textwidth}
		\includegraphics[width=\textwidth]{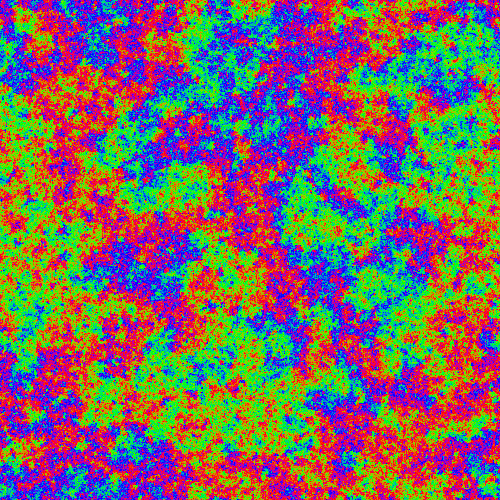}
		\caption{$\beta=1$}
		\label{fig:xy-sample-500x500-beta1}
	\end{subfigure}%
	\begin{subfigure}{10pt}
		\quad
	\end{subfigure}%
	\begin{subfigure}[t]{.49\textwidth}
		\includegraphics[width=\textwidth]{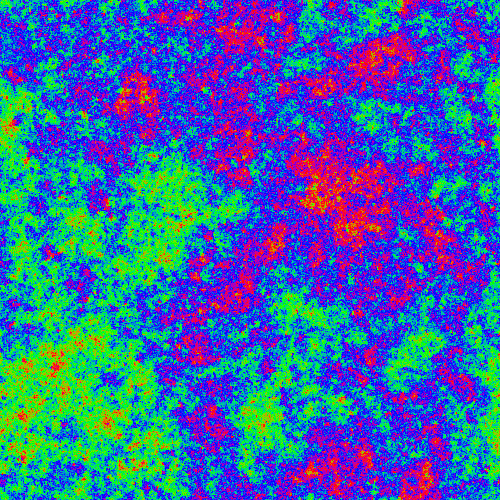}
		\caption{$\beta=1.12$}
		\label{fig:xy-sample-500x500-beta1_12}
	\end{subfigure}
	\medbreak
	\begin{subfigure}[t]{.49\textwidth}
		\includegraphics[width=\textwidth]{xy-sample-500x500-beta1_5.png}
		\caption{$\beta=1.5$}
		\label{fig:xy-sample-500x500-beta1_5}
	\end{subfigure}%
	\begin{subfigure}{10pt}
		\quad
	\end{subfigure}%
	\begin{subfigure}[t]{.49\textwidth}
		\includegraphics[width=\textwidth]{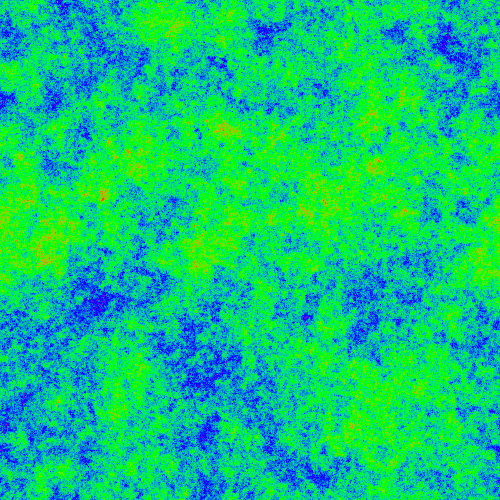}
		\caption{$\beta=3$}
		\label{fig:xy-sample-500x500-beta3}
	\end{subfigure}
	\caption{Samples of random spin configurations in the two-dimensional XY model ($n=2$) at and near the conjectured critical inverse temperature $\beta_c \approx 1.1199$ \cite{hasenbusch2005two, komura2012large}. Configurations are on a $500\times500$ torus. The angles of the spins are encoded by colors, with 0, 120 and 240 degrees having colors green, blue and red, and interpolating in between. The samples are generated using Wolff's cluster algorithm \cite{wolff1989collective}.}
	\label{fig:xy-samples}
\end{figure}

\subsection{Definitions}\label{sec:definitions}
Let $n\ge 1$ be an integer and let $G = (V(G), E(G))$ be a finite
graph. A \emph{configuration} of the \emph{spin $O(n)$ model}, sometimes
called the \emph{$n$-vector model}, on $G$ is an assignment
$\sigma:V(G)\to \mathbb S^{n-1}$ of spins to each vertex of $G$,
where $\mathbb S^{n-1} \subseteq \R^n$ is the $(n-1)$-dimensional
unit sphere (simply $\{-1,1\}$ if $n=1$). We write
\begin{equation*}
\Omega := (\S^{n-1})^{V(G)}
\end{equation*}
for the space of configurations. At inverse temperature $\beta\in
[0,\infty)$, configurations are randomly chosen from the probability
measure $\mu_{G,n,\beta}$ given by
\begin{equation}\label{eq:spin_O_n_def}
  d\mu_{G,n,\beta}(\sigma) := \frac{1}{Z^{\text{spin}}_{G,n,\beta}} \exp \left[\beta \sum_{\{u,v\}\in
E(G)}\left\langle\sigma_u,\sigma_v\right\rangle\right]
  d\sigma,
\end{equation}
where $\langle \cdot, \cdot\rangle$ denotes the standard inner product in
$\R^n$, the \emph{partition function} $Z^{\text{spin}}_{G,n,\beta}$
is given by
\begin{equation}\label{eq:Z_def}
  Z^{\text{spin}}_{G,n,\beta}:=\int_{\Omega} \exp \left[\beta \sum_{\{u,v\}\in
E(G)}\left\langle\sigma_u,\sigma_v\right\rangle\right]
  d\sigma
\end{equation}
and $d\sigma$ is the uniform probability measure on $\Omega$ (i.e.,
the product measure of the uniform distributions on $\mathbb
S^{n-1}$ for each vertex in $G$).

Special cases of the model have names of their own:
\begin{itemize}[noitemsep,topsep=0.5em]
\setlength\itemsep{0.25em}
  \item When $n=1$, spins take values in $\{-1,1\}$ and the model becomes
the famous \emph{Ising model}. See Figure~\reffig{fig:ising-samples} for samples from this model.
  \item When $n=2$, spins take values in the
unit circle and the model is called the \emph{XY model} or the
\emph{plane rotator model}. See Figure~\reffig{fig:xy-samples} for samples from this model. See also the two top figures on the cover page which show samples of the XY model with $\beta=1.5$.
  \item When $n=3$, spins take values in the
two-dimensional sphere $\mathbb S^2$ and the model is called the
\emph{Heisenberg model}. See Figure~\reffig{fig:heisenberg-samples} for samples from this model.

  \item In a sense, as $n$ tends to infinity the model approaches the \emph{Berlin--Kac spherical model} (which will not be discussed in these notes), see \cite{BerKac52, ThompsonKac71, Sta68-2} and \cite[Chapter 5]{Bax89}.
\end{itemize}

We will sometimes discuss a more general model, in which we replace
the inner product in \eqref{eq:spin_O_n_def} by a function of that
inner product. In other words, when the energy of a configuration is
measured using a more general pair interaction term. Precisely,
given a measurable function $U:[-1,1]\to\R\cup\{\infty\}$, termed the
\emph{potential function}, we define the \emph{spin $O(n)$ model
with potential $U$} to be the probability measure $\mu_{G,n,U}$ over
configurations $\sigma:V(G)\to \mathbb S^{n-1}$ given by
\begin{equation}\label{eq:spin_O_n_potential_U}
  d\mu_{G,n,U}(\sigma) := \frac{1}{Z^{\text{spin}}_{G,n,U}} \exp \left[-\sum_{\{u,v\}\in
E(G)}U(\left\langle\sigma_u,\sigma_v\right\rangle)\right]
  d\sigma,
\end{equation}
where the partition function $Z^{\text{spin}}_{G,n,U}$ is defined analogously to \eqref{eq:Z_def} and where we set $\exp(-\infty):=0$. Of course, for this to be well defined (i.e., to have finite $Z^{\text{spin}}_{G,n,U}$) some restrictions need to be placed on $U$ but this will always be the case in the models discussed in these notes.

The spin $O(n)$ model defined in \eqref{eq:spin_O_n_def} with $\beta\in[0,\infty)$ is called \emph{ferromagnetic}. If $\beta$ is taken negative in \eqref{eq:spin_O_n_def}, equivalently $U(r) = \beta r$ for $\beta>0$ in \eqref{eq:spin_O_n_potential_U}, the model is called \emph{anti-ferromagnetic}. On bipartite graphs, the ferromagnetic and anti-ferromagnetic versions are isomorphic through the map which sends $\sigma_v$ to $-\sigma_v$ for all $v$ in one of the partite classes. The two versions are genuinely different on non-bipartite graphs; see Section~\ref{sec:loop-model-def} and Section~\ref{sec:loop_O_n_phase_diagram} for a discussion of the Ising model on the triangular lattice.

The model admits many extensions and generalizations. One may impose \emph{boundary conditions} in which
the values of certain spins are pre-specified. An \emph{external magnetic
field} can be applied by taking a vector $s\in\mathbb \R^n$ and adding a term
of the form $\sum_{v\in V(G)} \left\langle \sigma_v,s\right\rangle$ to the exponent in the definition of the densities
\eqref{eq:spin_O_n_def} and \eqref{eq:spin_O_n_potential_U}. The model can be made \emph{anisotropic} by replacing the standard inner product $\langle\cdot,\cdot\rangle$ in \eqref{eq:spin_O_n_def} and \eqref{eq:spin_O_n_potential_U} with a different inner product. A different \emph{single-site distribution} may be imposed, replacing the measure $d\sigma$ in \eqref{eq:spin_O_n_def} and \eqref{eq:spin_O_n_potential_U} with another product measure on the vertices of $G$, thus allowing spins to take values in all of $\R^{n}$ (e.g., taking the single-site density $\exp(-|\sigma_v|^4)$). We will, however, focus on the versions of
the model described above.

The graph $G$ is typically taken to be a portion of a
$d$-dimensional lattice, possibly with periodic boundary
conditions. When discussing the spin $O(n)$ model in these notes we
mostly take
\begin{equation*}
  G = \T_L^d,
\end{equation*}
where $\T_{L}^d$ denotes the \emph{$d$-dimensional discrete torus of
side length $2L$} defined as follows: The vertex set of $\T_L^d$ is
\begin{equation}\label{eq:vertex_set_of_T_L_d}
V(\T_L^d):=\{-L+1,-L+2,\ldots,L-1,L\}^{d}
\end{equation}
and a pair $u,v\in V(\T_L^d)$ is adjacent, written $\{u,v\}\in
E(\T_L^d)$, if $u$ and $v$ are equal in all but one coordinate and
differ by exactly $1$ modulo $2L$ in that coordinate. We write
$\|x-y\|_1$ for the \emph{graph distance in $\T_L^d$} of two
vertices $x,y\in V(\T_L^d)$ (for brevity, we suppress the dependence
on $L$ in this notation).

The results presented below should admit analogues if the graph $G$
is changed to a different $d$-dimensional lattice graph with
appropriate boundary conditions. However, the presented proofs
sometimes require the presence of symmetries in the graph $G$.

\begin{figure}[!ht]
	\centering
	\begin{subfigure}[t]{.49\textwidth}
		\includegraphics[width=\textwidth]{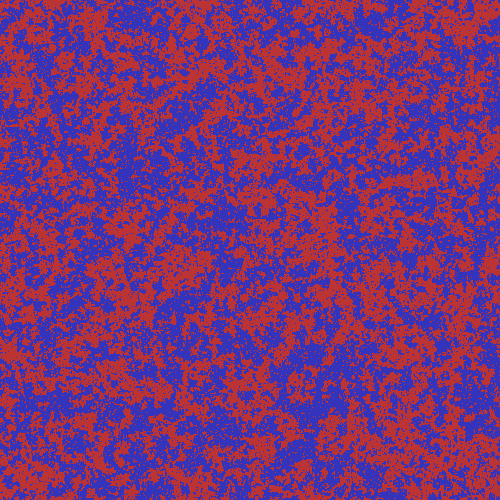}
		\caption{$\beta=0.4<\beta_c$}
		\label{fig:ising-sample-500x500-beta1}
	\end{subfigure}%
	\begin{subfigure}{10pt}
		\quad
	\end{subfigure}%
	\begin{subfigure}[t]{.49\textwidth}
		\includegraphics[width=\textwidth]{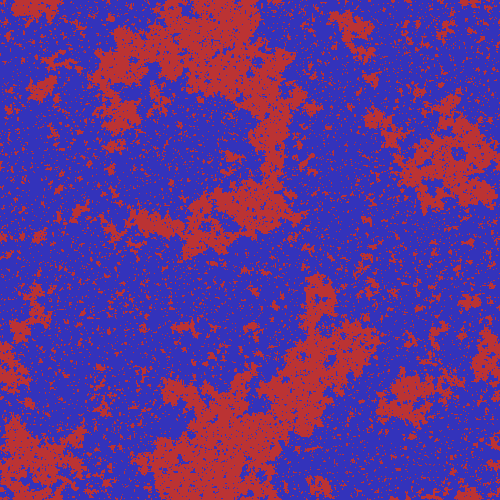}
		\caption{$\beta=\beta_c \approx 0.4407$}
		\label{fig:ising-sample-500x500-beta1_12}
	\end{subfigure}
	\medbreak
	\begin{subfigure}[t]{.49\textwidth}
		\includegraphics[width=\textwidth]{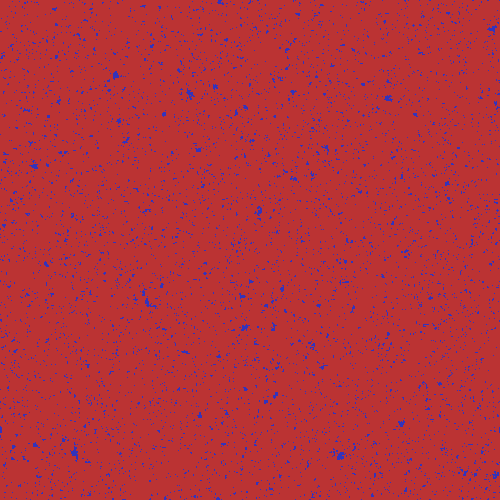}
		\caption{$\beta=0.5>\beta_c$}
		\label{fig:ising-sample-500x500-beta1_5}
	\end{subfigure}%
	\begin{subfigure}{10pt}
		\quad
	\end{subfigure}%
	\begin{subfigure}[t]{.49\textwidth}
		\includegraphics[width=\textwidth]{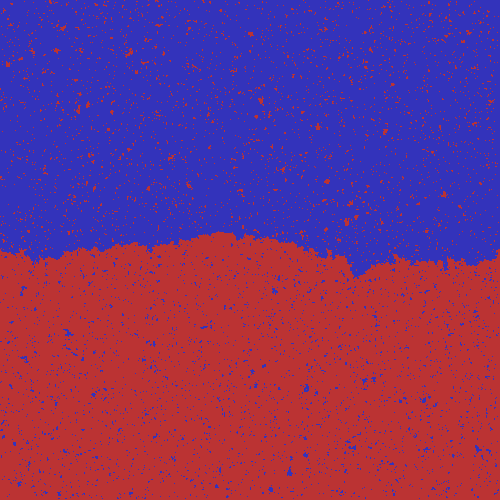}
		\caption{$\beta=0.5$ with Dobrushin boundary conditions}
		\label{fig:ising-sample-500x500-beta3}
	\end{subfigure}
	\caption{Samples of random configurations in the two-dimensional Ising model ($n=1$) at and near the critical inverse temperature $\beta_c = \tfrac12 \log(1+\sqrt{2})$. Configurations are on a $500\times500$ torus and are generated using Wolff's cluster algorithm \cite{wolff1989collective}. Dobrushin boundary conditions corresponds to fixing the top and bottom halves of the boundary to have different spins.}
	\label{fig:ising-samples}
\end{figure}

\begin{figure}[!ht]
	\centering
	\begin{subfigure}[t]{.49\textwidth}
		\includegraphics[width=\textwidth]{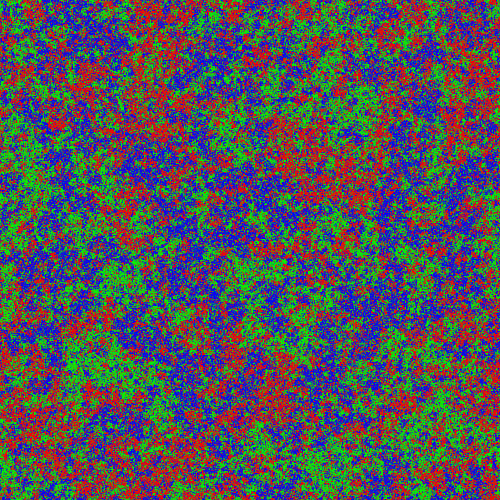}
		\caption{$\beta=2$}
		\label{fig:heisenberg-sample-500x500-beta2}
	\end{subfigure}%
	\begin{subfigure}{10pt}
		\quad
	\end{subfigure}%
	\begin{subfigure}[t]{.49\textwidth}
		\includegraphics[width=\textwidth]{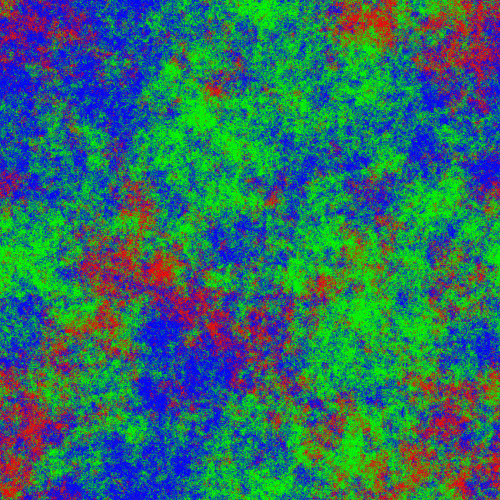}
		\caption{$\beta=10$}
		\label{fig:heisenberg-sample-500x500-beta10}
	\end{subfigure}
	\caption{Samples of random configurations in the two-dimensional Heisenberg model ($n=3$). Configurations are on a $500\times500$ torus and are generated using Wolff's cluster algorithm \cite{wolff1989collective}. It is predicted \cite{polyakov1975interaction} that there is no phase transition for $d=2$ and $n \ge 3$ so that correlations decay exponentially at any inverse temperature.}
	\label{fig:heisenberg-samples}
\end{figure}

\subsection{Main results and conjectures}\label{sec:spin_O_n_results_and_conjectures}
We will focus on the questions of existence of long-range order and decay of correlations in the spin $O(n)$ model. To this end we shall study the correlation
\begin{equation*}
  \rho_{x,y}:=\E(\left\langle\sigma_x,\sigma_y\right\rangle)
\end{equation*}
for a configuration $\sigma$ randomly chosen from
$\mu_{\T_L^d,n,\beta}$, the (ferromagnetic) spin $O(n)$ model at
inverse temperature $\beta\in[0,\infty)$, and two vertices $x,y\in
V(\T_L^d)$ with large graph distance $\|x-y\|_1$. The magnitude of
this correlation behaves very differently for different combinations
of the spatial dimension $d$, number of spin components $n$ and
inverse temperature $\beta$. The following list summarizes the main
results and conjectures regarding $\rho_{x,y}$. Most of the claims
in the list are elaborated upon and proved in the subsequent
sections. We use the notation $c_{\beta}, C_{\beta}, c_{n,\beta},
\ldots$ to denote positive constants whose value depends only on the
parameters given in the subscript (and is always independent of the
lattice size $L$) and may change from line to line.

\medbreak \noindent{\bf Non-negativity and monotonicity.} The
correlation is always non-negative, that is,
\begin{equation*}
  d,n\ge 1,\; \beta\in [0,\infty)\colon\quad \rho_{x,y} \ge 0\quad\text{for all $x,y\in V(\T_L^d)$}.
\end{equation*}
As we shall discuss, this result is a special case of an inequality
of Griffiths \cite{Gri67}. It is also natural to expect
the correlation to be monotonic non-decreasing in $\beta$. A second
inequality of Griffiths \cite{Gri67} implies this for the Ising model and was
later extended by Ginibre \cite{Gin70} to include the XY model and more general
settings. Precisely,
\begin{equation*}
  d\ge 1,\; n\in\{1,2\}\colon\quad \text{for all $x,y\in V(\T_L^d)$, $\rho_{x,y}$ is
  non-decreasing as $\beta$ increases in $[0,\infty)$}.
\end{equation*}
It appears to be unknown whether this monotonicity holds also for
$n\ge 3$. Counterexamples exist for related inequalities in
certain quantum \cite{hurst1969griffiths} and classical \cite{Syl80} spin systems.

\medbreak
\noindent{\bf High temperatures and spatial dimension $d=1$.} All the models
exhibit \emph{exponential decay of correlations} at high
temperature. Precisely, there exists a $\beta_0(d,n)>0$ such that
\begin{equation}\label{eq:exponential_decay_high_temperatures}
  d,n\ge 1,\; \beta<\beta_0(d,n)\colon\quad \rho_{x,y} \le C_{d,n,\beta}\exp(-c_{d,n,\beta}\|x-y\|_1)\quad\text{for all $x,y\in V(\T_L^d)$}.
\end{equation}
This is a relatively simple fact and the main interest is in
understanding the behavior at low temperatures. In one spatial
dimension ($d=1$) the exponential decay persists at all positive
temperatures. That is,
  \begin{equation}\label{eq:exponential_decay_1d}
    d=1,\; n\ge 1,\; \beta\in[0,\infty)\colon\quad \rho_{x,y} \le C_{n,\beta}\exp(-c_{n,\beta}\|x-y\|_1)\quad\text{for all $x,y\in V(\T_L^1)$}.
  \end{equation}

\medbreak
\noindent{\bf The Ising model $n=1$.}
The Ising model exhibits a \emph{phase
transition} in all dimensions $d\ge 2$ at a critical inverse temperature
$\beta_c(d)$. The transition is from a regime with exponential decay
of correlations \cite{aizenman1985absence,Aiz85Rig,AizBarFer87, duminil2016new,duminil2017sharp}\footnote{Exponential decay is stated in these references in the infinite-volume limit, but is derived as a consequence of a finite-volume criterion and is thus implied, as the infinite-volume measure is unique, also in finite volume.},
\begin{equation*}
  d\ge 2,\; n=1,\; \beta<\beta_c(d)\colon\quad \rho_{x,y} \le C_{d,\beta}\exp(-c_{d,\beta}\|x-y\|_1)\quad\text{for all $x,y\in V(\T_L^d)$}
\end{equation*}
to a regime with \emph{long-range order}, or
\emph{spontaneous magnetization}, which is characterized by
\begin{equation*}
  d\ge 2,\; n=1,\; \beta>\beta_c(d)\colon\quad\rho_{x,y} \ge c_{d,\beta}\quad\text{for all $x,y\in V(\T_L^d)$}.
\end{equation*}

The behavior of the model at the critical temperature, when $\beta =
\beta_c(d)$, is a rich source of study with many mathematical
features. For instance, the two-dimensional model is exactly solvable, as discovered by Onsager \cite{Ons44}, and has a conformally-invariant scaling limit, features of which were first established by Smirnov \cite{Smi06, Smi10}; see \cite{CheSmi12, CheDumHon14, CHI15, benoist2016scaling} and references within for recent progress.
We mention that it is proved (see Aizenman, Duminil-Copin,
Sidoravicius \cite{AizDumSid15} and references within) that the
model does not exhibit long-range order at its critical point in all
dimensions $d\ge 2$. Moreover, in dimension $d=2$ it is known \cite{McCWu73, pinson2012rotational} (see also \cite{CHI15}) that
correlations \emph{decay as a power-law} with exponent $1/4$ at the critical point, whose exact value is $\beta_c(2)=\tfrac12 \log(1+\sqrt{2})$ as first determined by Kramers--Wannier \cite{KraWan41} and Onsager \cite{Ons44},
\begin{equation*}
  d=2,\; n=1,\; \beta=\beta_c(2)\colon\quad\E^{\Z^2}(\sigma_x \sigma_y) \sim C\|x-y\|_2^{-\frac{1}{4}},\quad\text{$x,y\in \Z^2$, $\|x-y\|_2\to\infty$},
\end{equation*}
where we write $\E^{\Z^2}$ for the expectation in the (unique) infinite-volume measure of the two-dimensional critical Ising model, and $\|\cdot\|_2$ denotes the standard Euclidean norm. Lastly, in dimensions higher than some threshold $d_0$, Sakai \cite{sakai2007lace} proved that
\begin{equation*}
  d\ge d_0,\; n=1,\; \beta=\beta_c(d)\colon\quad\E^{\Z^d}(\sigma_x \sigma_y) \sim C_d\|x-y\|_2^{-(d-2)},\quad\text{$x,y\in \Z^d$, $\|x-y\|_2\to\infty$},
\end{equation*}
where, as before, $\E^{\Z^d}$ is the expectation in the (unique) infinite-volume measure of the $d$-dimensional critical Ising model.

The study of the model at or near its critical temperature is beyond the scope of these notes.

\medbreak
\noindent{\bf The Mermin--Wagner theorem: No continuous symmetry breaking in
$2d$.} Perhaps surprisingly, the behavior of the two-dimensional
model when $n\ge2$, so that the spin space $\mathbb S^{n-1}$ has a
continuous symmetry, is quite different from that of the Ising
model. The \emph{Mermin--Wagner theorem} \cite{mermin1966absence, mermin1967absence} asserts that
in this case there is no phase with long-range order at any inverse
temperature $\beta$. Quantifying the rate at which correlations
decay has been the focus of much research along the years \cite{hohenberg1967existence,jasnow1969broken,polyakov1975interaction,dobrushin1975absence,shlosman1977absence, shlosman1978decrease,pfister1981symmetry,simon1981rigorous,FroSpe81,ito1982clustering,bonato1982mermin,messager1984upper,naddaf1997decay,ioffe20022d,gagnebin2014upper}
and is still not completely understood. Improving on earlier bounds, McBryan and Spencer~\cite{McBSpe77} showed in 1977 that the decay occurs at
least at a power-law rate,
\begin{equation}\label{eq:algebraic_decay_Mermin_Wagner}
  d=2,\; n\ge 2,\; \beta\in[0,\infty)\colon\quad \rho_{x,y}\le C_{n,\beta}\|x-y\|_1^{-c_{n,\beta}}\quad\text{for all $x,y\in V(\T_L^2)$}.
\end{equation}
The sharpness of this bound is discussed in the next paragraphs.

\medbreak
\noindent{\bf The Berezinskii--Kosterlitz--Thouless transition for the $2d$ XY
Model.} It was predicted by Berezinskii \cite{Ber72} and by
Kosterlitz and Thouless \cite{KosTho72, KosTho73} that the XY model
($n=2$) in two spatial dimensions should indeed exhibit power-law
decay of correlations at low temperatures. Thus the model undergoes
a phase transition (of a different nature than that of the Ising
model) from a phase with exponential decay of correlations to a
phase with power-law decay of correlations. This transition is
called the Berezinskii--Kosterlitz--Thouless transition. The existence
of the transition has been proved mathematically in the celebrated
work of Fr\"ohlich and Spencer \cite{FroSpe81}, who show that there
exists a $\beta_1$ for which
\begin{equation}\label{eq:polynomial_decay_lower_bound}
  d=2,\; n=2,\; \beta>\beta_1\colon\quad \E^{\Z^2}(\left\langle\sigma_x,\sigma_y\right\rangle)\ge c_\beta\|x-y\|_1^{-C_\beta}\quad\text{for all distinct $x,y\in \Z^2$},
\end{equation}
where $\E^{\Z^2}$ denotes expectation in the unique \cite{bricmont1977uniqueness} translation-invariant infinite-volume Gibbs measure of the two-dimensional XY model at inverse temperature $\beta$.

A rigorous proof of the bound
\eqref{eq:polynomial_decay_lower_bound} is beyond the scope of these
notes (see~\cite{kharash2017fr} for a recent presentation of the proof). In Section~\ref{sec:Berezinskii-Kosterlitz-Thouless} we
present a heuristic discussion of the transition highlighting the
role of \emph{vortices} - cycles of length $4$ in $\T_L^2$ on which
the configuration completes a full rotation. We then proceed to
present a beautiful result of Aizenman \cite{Aiz94}, following
Patrascioiu and Seiler \cite{PatSei92}, who showed that correlations
decay at most as fast as a power-law in the spin $O(2)$ model with
potential $U$, for certain potentials $U$ for which vortices are
deterministically excluded.

\medbreak
\noindent{\bf Polyakov's conjecture for the $2d$ Heisenberg model.} Polyakov
\cite{polyakov1975interaction} predicted in 1975 that the spin $O(n)$
model with $n\ge 3$ should exhibit \emph{exponential} decay of correlations
in two dimensions at \emph{any} positive temperature. That is, that there is no
phase transition of the Berezinskii--Kosterlitz--Thouless type in the
Heisenberg model and in the spin $O(n)$ models with larger $n$. On the torus, this
prediction may be stated precisely as
\begin{equation*}
  d=2,\; n\ge 3,\; \beta\in[0,\infty)\colon\quad \rho_{x,y}\le C_{n,\beta}\exp(-c_{n,\beta}\|x-y\|_1)\quad\text{for all $x,y\in V(\T_L^2)$}.
\end{equation*}
Giving a mathematical proof of this statement (or its analog in infinite volume) remains one of the
major challenges of the subject. The best known results in this
direction are by Kupiainen~\cite{Kup80} who performed a $1/n$-expansion as $n$ tends to infinity.

\medbreak
\noindent{\bf The infra-red bound: Long-range order in dimensions $d\ge3$.}
In three and higher spatial dimensions, the spin $O(n)$ model
exhibits long-range order at sufficiently low temperatures for all
$n$. This was established by Fr\"ohlich, Simon and Spencer
\cite{FroSimSpe76} in 1976 who introduced the powerful method of the
\emph{infra-red bound}, and applied it to the analysis of the spin
$O(n)$ and other models. They
prove that correlations do not decay at temperatures below a
threshold $\beta_1(d,n)^{-1}$, at least in the following averaged sense,
\begin{equation*}
  d\ge 3,\; n\ge 1,\; \beta>\beta_1(d,n)\colon\quad\frac{1}{|V(\T_L^d)|^2}\sum_{x,y\in V(\T_L^d)}\rho_{x,y} \ge c_{d,n,\beta}.
\end{equation*}
The proof uses the reflection symmetries of the underlying lattice,
relying on the tool of \emph{reflection positivity}.

\subsection{Non-negativity and monotonicity of correlations}
\label{spin_nonnegativity_of_correlations}
In this section we discuss the non-negativity and monotonicity in temperature of the correlations $\rho_{x,y}=\E(\left\langle\sigma_x,\sigma_y\right\rangle)$. To remain with a unified presentation, our discussion is restricted to the simplest setup with nearest-neighbor interactions. Many extensions are available in the literature. Recent accounts can be found in the book of Friedli and Velenik \cite[Sections~3.6, 3.8 and 3.9]{friedli2016statistical} and in the review of Benassi--Lees--Ueltschi \cite{benassi2016correlation}.

We start our discussion by introducing the spin $O(n)$ model with general
non-negative coupling constants. Let $N\ge 1$ be an integer and let $J = (J_{\{i,j\}})_{1\le i<j\le N}$
be non-negative real numbers. The spin $O(n)$ model with coupling
constants $J$ is the probability measure on $(\S^{n-1})^{N}$
defined by
\begin{equation}\label{eq:spin_O_n_general_couplings_def}
  d\mu_{n,J}(\sigma) := \frac{1}{Z^{\text{spin}}_{n,J}} \exp \left[\sum_{1\le i<j\le N} J_{\{i,j\}} \left\langle\sigma_i,\sigma_j\right\rangle\right] d\sigma,
\end{equation}
where, as before, $d\sigma$ is the uniform probability measure on
$(\S^{n-1})^{N}$, $Z^{\text{spin}}_{n,J}$ is chosen to
normalize $\mu_{n,J}$ to be a probability measure and we refer to the case $n=1$ as the Ising model.
When we speak about the spin $O(n)$ model on a finite graph $G = (V(G), E(G))$ with coupling constants $J=(J_{\{u,v\}})_{\{u,v\} \in E(G)}$, it should be understood that $N = |V(G)|$, that the vertex-set $V(G)$ is identified with $\{1,\ldots, N\}$ and that $J_{\{i,j\}}=0$ for $\{i,j\} \notin E(G)$. Thus, the standard spin $O(n)$ model \eqref{eq:spin_O_n_def} on $G$ at inverse temperature $\beta$ is obtained as the special case in which $J_{\{u,v\}} = \beta$ for $\{u,v\}\in E(G)$.

The following non-negativity result is a special case of Griffiths' first
inequality \cite{Gri67}.
\begin{theorem}\label{thm:Griffiths_first_inequality}
  Let $N\ge 1$ be an integer and let $J = (J_{\{i,j\}})_{1\le i<j\le N}$ be non-negative.
  If $\sigma$ is sampled from the Ising model with coupling constants $J$ then
  \begin{equation*}
     \E\left(\prod_{x \in A}\sigma_x \right)\ge 0\quad\text{for all $A \subset \{1,\dots,N\}$}.
  \end{equation*}
\end{theorem}
\begin{proof}
  By definition,
  \begin{equation*}
    \E\left(\prod_{x \in A}\sigma_x\right) =
    \frac{1}{2^{N}Z^{\text{spin}}_{1,J}}\sum_{\sigma\in
    \{-1,1\}^{N}} \left(\prod_{x \in A}\sigma_x\right) \exp \left[\sum_{1\le i<j\le N} J_{\{i,j\}} \sigma_i \sigma_j\right].
  \end{equation*}
  Using the Taylor expansion $e^t = \sum_{m=0}^{\infty}
  \frac{t^m}{m!}$, we conclude that $\E(\sigma_x \sigma_y)$ is an absolutely convergent series
  with \emph{non-negative} coefficients of products of
  the values of $\sigma$ on various vertices. That is,
  \begin{equation*}
    \E\left(\prod_{x \in A}\sigma_x\right) = \sum_{\sigma\in
    \{-1,1\}^{N}}\, \sum_{{m\in \{0,1,2,\ldots\}}^{N}} C_m \prod_{1\le i\le N} \sigma_i^{m_i},
  \end{equation*}
  where each $C_m=C_m(A)\ge 0$ and the series is absolutely convergent (in addition, one may, in fact, restrict to $m\in\{0,1\}^{N}$ as when $\eps\in\{-1,1\}$
  we have $\eps^k = \eps$ or $\eps^k = 1$ according to the parity of
  $k$). The non-negativity of $\E(\prod_{x \in A}\sigma_x)$ now follows as, for each
  $m\in \{0,1,2,\ldots\}^{N}$,
  \begin{equation*}
    \sum_{\sigma\in\{-1,1\}^{N}} \prod_{1\le i\le N} \sigma_i^{m_i} = \prod_{1\le i \le N} \left((-1)^{m_i} + 1^{m_i}\right) = \begin{cases}2^{N} & \text{$m_i$ is
    even for all $i$}\\0&\text{otherwise}\end{cases}.\qedhere
  \end{equation*}
\end{proof}
\medbreak
\noindent {\bf Exercise.}
Give an alternative proof of Theorem~\ref{thm:Griffiths_first_inequality} by extending the derivation of the Edwards--Sokal coupling in Section~\ref{sec:high-temperature_expansion} below to the Ising model with general non-negative coupling constants and arguing similarly to Remark~\ref{rem:edwards_sokal}.
\medbreak
We now deduce non-negativity of correlations for the spin $O(n)$
models with $n\ge 2$ by showing that conditioning on $n-1$ spin components induces an Ising model with non-negative coupling constants on the sign of the remaining spin component. The argument applies to spin $O(n)$ models
with potential $U:[-1,1]\to\mathbb{R}\cup\{\infty\}$ (see
\eqref{eq:spin_O_n_potential_U}) as long as the potential is
\emph{non-increasing} in the sense that
\begin{equation}\label{eq:ferromagnetic_def}
  \text{$U(r_1)\ge U(r_2)$ when $r_1\le r_2$.}
\end{equation}
This property implies that configurations in which adjacent spins
are more aligned (i.e., have larger inner product) have higher
density, a characteristic of ferromagnets.

To state the above precisely, we embed $\mathbb S^{n-1}$ into $\R^n$ so as to allow writing the
coordinates of a configuration $\sigma:V(G)\to \mathbb S^{n-1}$ explicitly as
\begin{equation*}
	\sigma_v = (\sigma_v^1, \sigma_v^2, \ldots, \sigma_v^n)\quad\text{at each vertex $v\in V(G)$}.
\end{equation*}
For $1\le j\le n$, we write $\sigma^j$ for the function
$(\sigma^j_v)$, $v\in V(G)$.
We also introduce a function $\eps: V(G) \to \{-1,1\}$ defined uniquely by $\sigma_v^1 = |\sigma_v^1| \eps_v$ (when $\sigma_v^1=0$, we arbitrarily set $\eps_v := 0$).
We note that $\sigma$ is determined by $(\eps, \sigma^2, \dots, \sigma^n)$ since $\sigma^1_v = \eps_v |\sigma^1_v|$ and $|\sigma_v^1|$ is determined from $(\sigma_v^j)_{2\le j\le n}$ as $\sigma_v\in \mathbb S^{n-1}$.

\begin{theorem}\label{thm:conditioning_gives_Ising}
	Let $n\ge 2$, $G = (V(G), E(G))$ be a finite graph and let $U:[-1,1]\to\R\cup\{\infty\}$
	be \emph{non-increasing}. If $\sigma$ is sampled from the spin $O(n)$ model on $G$ with potential $U$, then, conditioned on $(\sigma^2, \sigma^3,\ldots, \sigma^n)$, the random signs $\eps$ are distributed as an Ising model on $G$
	with coupling constants $J$ given by
	\[ J_{\{u,v\}} := - \frac12 U\left(|\sigma_u^1|\cdot |\sigma_v^1| + \sum_{j=2}^n\sigma_u^j \sigma_v^j\right) + \frac12 U\left(-|\sigma_u^1|\cdot |\sigma_v^1| + \sum_{j=2}^n\sigma_u^j \sigma_v^j\right) .\]
	In particular, the coupling constants are non-negative so that for all $x,y\in V(G)$,
	\begin{equation*}
	\E(\left\langle\sigma_x,\sigma_y\right\rangle)\ge 0 \qquad\text{and}\qquad \E\left(\eps_x \eps_y\; |\; (\sigma^j)_{2\le j\le n}\right)\ge
	0\quad\text{almost surely}.
	\end{equation*}
\end{theorem}
\begin{proof}
	Observe that the density of $\eps$ conditioned on
	$(\sigma^j)_{2\le j\le n}$ (with respect to the uniform measure on
	$\{-1,1\}^{V(G)}$) is proportional to
	\begin{align*}
	\exp \Bigg[-\sum_{\{u,v\}\in E(G)}U(\left\langle\sigma_u,\sigma_v\right\rangle)\Bigg]
	 &= \exp\Bigg[-\sum_{\{u,v\}\in E(G)}U\Bigg(|\sigma_u^1|\cdot |\sigma_v^1|
	\eps_u \eps_v + \sum_{j=2}^n\sigma_u^j \sigma_v^j\Bigg)\Bigg] \\
	 &= \exp\Bigg[\sum_{\{u,v\}\in E(G)} \big (J_{\{u,v\}} \eps_u \eps_v + I_{\{u,v\}} \big) \Bigg] \\
	 &= I \cdot \exp\Bigg[\sum_{\{u,v\}\in E(G)} J_{\{u,v\}} \eps_u \eps_v \Bigg] ,
	\end{align*}
	where $I_{\{u,v\}}$ and $I$ are measurable with respect to $(\sigma^j)_{2\le j\le n}$.
	We conclude that, conditioned on $(\sigma^j)_{2\le
		j\le n}$, the signs $\eps$ are distributed as an Ising model on $G$
	with coupling constants $J=(J_{\{u,v\}})_{\{u,v\} \in E(G)}$.

	By the assumption that $U$ is non-increasing, the coupling constants are almost surely non-negative. Thus, Theorem~\ref{thm:Griffiths_first_inequality} implies that
	\begin{equation*}
	\E\left(\eps_x \eps_y\; |\; (\sigma^j)_{2\le j\le n}\right)\ge
	0\quad\text{almost surely for every }x,y \in V(G).
	\end{equation*}
	Finally, to see that $\E(\left\langle\sigma_x,\sigma_y\right\rangle) \ge 0$, note that
	\begin{equation*}
	\E(\left\langle\sigma_x,\sigma_y\right\rangle) = \E\left(\sum_{j=1}^n\sigma_x^j \sigma_y^j\right) = n\,\E\left(\sigma_x^1 \sigma_y^1\right),
	\end{equation*}
	as the distribution of $\sigma$ is invariant to global rotations
	(that is, for any $n\times n$ orthogonal matrix $O$, $\sigma$ has
	the same distribution as $(O\sigma_v)$, $v\in V(G)$, by the choice
	of density \eqref{eq:spin_O_n_potential_U}). In particular,
	\[ \E(\left\langle\sigma_x,\sigma_y\right\rangle) =
	n\,\E\left(\E\left(\sigma_x^1 \sigma_y^1\; |\; (\sigma^j)_{2\le j\le n}\right)\right) =
	n\,\E\left(|\sigma_x^1|\cdot |\sigma_y^1|\cdot\E\left(\eps_x \eps_y\; |\; (\sigma^j)_{2\le j\le n}\right)\right) \ge 0 . \qedhere \]
\end{proof}

We remark that Theorem~\ref{thm:conditioning_gives_Ising}
and its proof may be extended in a straightforward manner to the
case that different non-increasing potentials are placed on
different edges of the graph.

As another remark, we note that the non-negativity of $\E(\left\langle\sigma_x,\sigma_y\right\rangle)$ asserted by
Theorem~\ref{thm:conditioning_gives_Ising} may fail for
potentials which are not non-increasing. For instance, the
discussion of the anti-ferromagnetic spin $O(n)$ model in
Section~\ref{sec:definitions} shows that, on bipartite
graphs $G$ and with $x$ and $y$ on different bipartition classes,
the sign of $\E(\left\langle\sigma_x,\sigma_y\right\rangle)$ in the
spin $O(n)$ model is reversed when replacing $\beta$ by $-\beta$ in
\eqref{eq:spin_O_n_def}. A similar remark applies to the assertion
of Theorem~\ref{thm:Griffiths_first_inequality} when some of the
coupling constants are negative.

Lastly, we mention that the assumptions of
Theorem~\ref{thm:conditioning_gives_Ising} imply a stronger conclusion than the non-negativity of $\E(\left\langle\sigma_x,\sigma_y\right\rangle)$. In \cite{CohenAlloro2018} it is shown
that conditioned on $\sigma_x$, there is a version of the density of
$\sigma_y$ (with respect to the uniform measure on $\S^{n-1}$) which is a non-decreasing function of $\left\langle\sigma_x,\sigma_y\right\rangle$.

We move now to discuss the monotonicity of correlations with the
inverse temperature $\beta$ in the spin $O(n)$ model. This was first established by Griffiths
for the Ising case \cite{Gri67} and is sometimes referred to as Griffiths' second inequality. It was established by Ginibre \cite{Gin70} for the XY case (the case $n=2$) and in more general settings. Establishing or refuting such monotonicity when $n\ge 3$ is an open problem of significant interest.

We again work in the generality of the spin $O(n)$ model with non-negative coupling constants.
\begin{theorem}\label{thm:second_Griffiths_inequality}
  Let $n\in \{1,2\}$, let $N\ge 1$ be an integer and let $J = (J_{\{i,j\}})_{1\le i<j\le N}$ be non-negative.
  If $\sigma$ is sampled from the spin $O(n)$ model with coupling constants $J$ then
  \begin{equation}\label{eq:non-negative_correlation_of_correlations}
     \E\left(\left\langle\sigma_x,\sigma_y\right\rangle\cdot\left\langle\sigma_z,\sigma_w\right\rangle\right)\ge \E\left(\left\langle\sigma_x,\sigma_y\right\rangle\right)\cdot \E\left(\left\langle\sigma_z,\sigma_w\right\rangle\right)\quad\text{for all $1\le x,y,z,w\le N$}.
  \end{equation}
  In other words, the random variables $\left\langle\sigma_x,\sigma_y\right\rangle$ and $\left\langle\sigma_z,\sigma_w\right\rangle$ are non-negatively correlated.
\end{theorem}
The theorem implies that each correlation $\E\left(\left\langle\sigma_x,\sigma_y\right\rangle\right)$ is a monotone non-decreasing function of each coupling constant $J_{\{z,w\}}$. Indeed, in the setting of the theorem, one checks in a straightforward manner that, for all $1\le x,y\le N$ and $1\le z<w\le N$,
\begin{equation*}
  \frac{\partial}{\partial J_{\{z,w\}}} \E\left(\left\langle\sigma_x,\sigma_y\right\rangle\right) = \E\left(\left\langle\sigma_x,\sigma_y\right\rangle\cdot\left\langle\sigma_z,\sigma_w\right\rangle\right)-\E\left(\left\langle\sigma_x,\sigma_y\right\rangle\right)\cdot \E\left(\left\langle\sigma_z,\sigma_w\right\rangle\right)\stackrel{\eqref{eq:non-negative_correlation_of_correlations}}{\ge} 0.
\end{equation*}
This monotonicity property is exceedingly useful as it allows to compare the correlations of the spin $O(n)$ model on different graphs by taking limits as various coupling constants tend to zero or infinity (corresponding to deletion or contraction of edges of the graph). For instance, one may use it to establish the existence of the infinite-volume (thermodynamic) limit of correlations in the spin $O(n)$ model ($n\in \{1,2\}$) on $\Z^d$, or to compare the behavior of the model in different spatial dimensions $d$.

The following lemma, introduced by Ginibre~\cite{Gin70}, is a key step in the proof of Theorem~\ref{thm:second_Griffiths_inequality}. Sylvester \cite{Syl80} has found counterexamples to the lemma when $n\ge 3$.
\begin{lemma}\label{lem:Ginibre_inequality}
Let $n\in\{1,2\}$ and let $N\ge 1$ be an integer. Then for every choice of non-negative integers $(k_{\{i,j\}}), (\ell_{\{i,j\}})$, $1\le i<j\le N$, we have
\begin{equation}\label{eq:Ginibre_inequality}
   \int\int \prod_{1\le i<j\le N} (\left\langle\sigma_i,\sigma_j\right\rangle - \left\langle\sigma'_i,\sigma'_j\right\rangle)^{k_{\{i,j\}}}\cdot (\left\langle\sigma_i,\sigma_j\right\rangle + \left\langle\sigma'_i,\sigma'_j\right\rangle)^{\ell_{\{i,j\}}} d\sigma d\sigma' \ge 0,
\end{equation}
where, as before, $d\sigma$ and $d\sigma'$ denote the uniform probability measure on $(\S^{n-1})^{N}$.
\end{lemma}
\begin{proof}
The change of variables $(\sigma, \sigma')\mapsto(\sigma', \sigma)$ preserves the measure $d\sigma d\sigma'$ and reverses the sign of each term of the form $\left\langle\sigma_i,\sigma_j\right\rangle - \left\langle\sigma'_i,\sigma'_j\right\rangle$ while keeping terms of the form $\left\langle\sigma_i,\sigma_j\right\rangle + \left\langle\sigma'_i,\sigma'_j\right\rangle$ fixed. The lemma thus follows in the case that $\sum_{1\le i<j\le N} k_{\{i,j\}}$ is odd as the integral in \eqref{eq:Ginibre_inequality} evaluates to zero. Let us then assume that
\begin{equation}\label{eq:parity_condition}
  \sum_{1\le i<j\le N}  k_{\{i,j\}}\;\;\;\text{is even}.
\end{equation}
  Identifying $\S^1$ with the unit circle in the complex plane and using that $n\in\{1,2\}$, we may express the spins as $\sigma_j = e^{i\theta_j}$ and $\sigma'_j = e^{i\theta'_j}$. With this notation, we have
  \begin{equation}\label{eq:first_trig_equality}
  \begin{split}
    \left\langle\sigma_i,\sigma_j\right\rangle - \left\langle\sigma'_i,\sigma'_j\right\rangle &= \cos(\theta_i - \theta_j) - \cos(\theta'_i - \theta'_j)\\
     &= -2\sin\left(\frac{\theta_i + \theta'_i}{2} - \frac{\theta_j + \theta'_j}{2}\right)\sin\left(\frac{\theta_i - \theta'_i}{2} - \frac{\theta_j - \theta'_j}{2}\right) ,
  \end{split}
  \end{equation}
  and similarly,
  \begin{equation*}
    \left\langle\sigma_i,\sigma_j\right\rangle + \left\langle\sigma'_i,\sigma'_j\right\rangle = 2\cos\left(\frac{\theta_i + \theta'_i}{2} - \frac{\theta_j + \theta'_j}{2}\right)\cos\left(\frac{\theta_i - \theta'_i}{2} - \frac{\theta_j - \theta'_j}{2}\right).
  \end{equation*}
  Thus, using \eqref{eq:parity_condition} to cancel the minus sign in the right-hand side of \eqref{eq:first_trig_equality}, we may write
  \begin{multline*}
    \int\int \prod_{1\le i<j\le N} (\left\langle\sigma_i,\sigma_j\right\rangle - \left\langle\sigma'_i,\sigma'_j\right\rangle)^{k_{\{i,j\}}}\cdot (\left\langle\sigma_i,\sigma_j\right\rangle + \left\langle\sigma'_i,\sigma'_j\right\rangle)^{\ell_{\{i,j\}}} d\sigma d\sigma'\\
    = \int \int F(\theta + \theta')F(\theta-\theta')d\sigma d\sigma' =: I
  \end{multline*}
  for a real-valued function $F$, satisfying the condition that $F(\theta + \theta')F(\theta-\theta')$ remains invariant when adding integer multiplies of $2\pi$ to any of the coordinates of $\theta$ or to any of the coordinates of $\theta'$. We now consider the cases $n=1$ and $n=2$ separately.

  Suppose first that $n=2$. Writing $d\theta, d\theta'$ for Lebesgue measure on $\R^N$, and using the above invariance property of $F$, we have
  \[ I = \frac{1}{(8\pi^2)^N}\int_{[-2\pi,2\pi]^N} \int_{[-\pi,\pi]^N} F(\theta + \theta')F(\theta-\theta')d\theta d\theta' .\]
  One may regard the domain of integration above as $([-2\pi,2\pi] \times [-\pi,\pi])^N$. Consider $E_0 := [-2\pi,2\pi] \times [-\pi,\pi]$, the projection of this domain onto one coordinate of $(\theta,\theta')$.
We shall split this domain into pieces and then rearrange them so as to obtain a square domain with side-length $2\sqrt{2}\pi$ rotated by 45 degrees and symmetric about the origin, i.e., the domain defined by $E_1 := \{ (\theta,\theta') \in \R^2 : |\theta \pm \theta'|\le 2\pi \}$. Indeed, each of the differences $E_0 \setminus E_1$ and $E_1 \setminus E_0$ consists of four triangular pieces, each being an isosceles right triangle with side-length $\pi$ and sides parallel to the axis, so that these pieces can be rearranged to obtain $E_1$ from $E_0$. In fact, the only operations involved in this procedure are translations by multiples of $2\pi$ in the direction of the axes. Thus, using the above invariance property of $F$, we conclude that $I$ can be written as
\[ I = \frac{1}{(8\pi^2)^{N}}\int \int_{(E_1)^N} F(\theta + \theta')F(\theta-\theta') d\theta d\theta'. \]
  The change of variables $(\theta,\theta')\mapsto(\theta + \theta', \theta - \theta')$ now shows that $I$ is the square of an integral of a real-valued function and hence is non-negative.

  The case $n=1$ is treated similarly, though one must take extra care in handling boundaries between domains of integration, as these no longer need to have measure zero. Writing $d\theta,d\theta'$ for the counting measure on $(\pi\Z)^N$, we have
\[ I = \frac{1}{8^N} \int_{\{-\pi,0,\pi,2\pi\}^N} \int_{\{0,\pi\}^N} F(\theta + \theta')F(\theta-\theta')d\theta d\theta' .\]
As before, we consider a single coordinate of $(\theta,\theta')$.
Observe that there is quite some freedom in changing the domain of integration $E_0 := \{-\pi,0,\pi,2\pi\} \times \{0,\pi\}$ without effecting the integral. Consider for instance the domain $E'_0$ obtained from $E_0$ by removing the points $\{(-\pi,\pi),(2\pi,\pi)\}$ and adding $\{(0,-\pi),(\pi,-\pi)\}$ instead. By the invariance property of $F$, the integral on $E'_0$ is the same as on $E_0$. To conclude as before that $I$ is non-negative, it suffices to find a domain of integration $E_1$, which coincides with $E'_0$ on $(\pi\Z)^2$, and which is a 45-degree rotated square (i.e., the product of an interval with itself in the $(\theta+\theta',\theta-\theta')$ coordinates). Indeed, one may easily verify that $E_1 := \{ (\theta,\theta') : -3/2 \le \theta \pm \theta' \le 5/2  \}$ is such a domain.
\end{proof}

\begin{proof}[Proof of Theorem~\ref{thm:second_Griffiths_inequality}]
	Let $\sigma$ and $\sigma'$ be two independent samples from the spin $O(n)$ model with coupling constants $J$. Then
\[ 2 \Cov\left(\left\langle\sigma_x,\sigma_y\right\rangle, \left\langle\sigma_z,\sigma_w\right\rangle\right) = \E\big[ \left(\left\langle\sigma_x,\sigma_y\right\rangle - \left\langle\sigma'_x,\sigma'_y\right\rangle\right)\cdot\left(\left\langle\sigma_z,\sigma_w\right\rangle-\left\langle\sigma'_z,\sigma'_w\right\rangle\right)\big] .\]
Thus, it suffices to show that the expectation on the right-hand side is non-negative.
Indeed, denoting $S^{\pm}_{\{i,j\}} := \left\langle\sigma_i,\sigma_j\right\rangle \pm \left\langle\sigma'_i,\sigma'_j\right\rangle$, this expectation is equal to
\[ \frac{1}{\big(Z^{\text{spin}}_{n,J}\big)^2} \int \int S^-_{\{x,y\}} \cdot S^-_{\{z,w\}} \cdot \exp \left[\sum_{1\le i<j\le N} J_{\{i,j\}} S^+_{\{i,j\}}\right] d\sigma d\sigma' ,\]
which, by expanding the exponent into a Taylor's series, is equal to
\[ \frac{1}{\big(Z^{\text{spin}}_{n,J}\big)^2} \sum_{m \in \{0,1,2,\dots\}^{\{\{i,j\}:1 \le i<j \le N\}}} C_m \int \int S^-_{\{x,y\}} \cdot S^-_{\{z,w\}} \cdot \prod_{1\le i<j\le N} \left(S^+_{\{i,j\}}\right)^{m_{\{i,j\}}} \, d\sigma d\sigma' ,\]
where each $C_m$ is non-negative and the series is absolutely convergent. The desired non-negativity now follows from Lemma~\ref{lem:Ginibre_inequality}.
\end{proof}

We are not aware of other proofs for Griffiths' second inequality, Theorem~\ref{thm:second_Griffiths_inequality}, for the XY model ($n=2$). The above proof may also be adapted to treat clock models, models of the XY type in which the spin is restricted to roots of unity of a given order (the ticks of the clock), see \cite{Gin70}. Alternative approaches are available in the Ising case ($n=1$): One proof relies on positive association (FKG) for the corresponding random-cluster model (see also Remark~\ref{rem:edwards_sokal}). A different argument of Ginibre~\cite{ginibre1969simple} deduces Theorem~\ref{thm:second_Griffiths_inequality} directly from Theorem~\ref{thm:Griffiths_first_inequality}.

\subsection{High-temperature expansion}\label{sec:high-temperature_expansion}

At infinite temperature ($\beta=0$) the models are completely
disordered, having all spins independent and uniformly distributed on $\S^{n-1}$. In this section we show that the disordered phase extends to high, but finite, temperatures (small positive $\beta$). Specifically, we show that the models exhibit
exponential decay of correlations in this regime, as stated in \eqref{eq:exponential_decay_high_temperatures} and \eqref{eq:exponential_decay_1d}.

We begin by expanding the partition function of the model on an arbitrary finite graph $G = (V(G), E(G))$ in the following manner. Denoting $f_\beta(s,t) := \exp \big[\beta \big(\langle s,t
\rangle + 1\big)\big] - 1$ for $s,t \in \S^{n-1}$, we have
\begin{equation}\label{eq:high-temperature_expansion}
\begin{split}
Z^{\text{spin}}_{G,n,\beta}
 &= \int_{\Omega} \prod_{\{u,v\}\in
    E(G)} \exp \left[\beta \left\langle\sigma_u,\sigma_v\right\rangle\right] d\sigma
 = e^{-\beta |E(G)|} \int_{\Omega} \prod_{\{u,v\}\in
    E(G)} \exp \left[\beta \left(\left\langle\sigma_u,\sigma_v\right\rangle + 1\right)\right] d\sigma \\
 &= e^{-\beta |E(G)|} \int_{\Omega} \prod_{\{u,v\}\in
    E(G)} \big( 1 + f_\beta(\sigma_u,\sigma_v) \big) d\sigma
 = e^{-\beta |E(G)|} \sum_{E \subset E(G)} \int_{\Omega} \prod_{\{u,v\}\in
    E} f_\beta(\sigma_u,\sigma_v) d\sigma.
\end{split}
\end{equation}

\medbreak
\noindent {\bf Exercise.}
Verify the last equality in the above expansion by showing that for any $(x_e)_{e \in \mathcal{E}}$,
\[ \prod_{e \in \mathcal{E}} (1+x_e) = \sum_{E \subset \mathcal{E}} \prod_{e \in E} x_e .\]
\medbreak
Thus, we have
\begin{equation}\label{eq:Z_high_temp_expansion}
Z^{\text{spin}}_{G,n,\beta} = e^{-\beta |E(G)|} \sum_{E \subset
E(G)} Z(E) ,
\end{equation}
where
\begin{equation}\label{eq:Z_of_E_def}
Z(E) := \int_{\Omega} \prod_{\{u,v\}\in E}
f_\beta(\sigma_u,\sigma_v) d\sigma .
\end{equation}
Since $f_\beta$ is non-negative, we may interpret~\eqref{eq:Z_high_temp_expansion} as prescribing a
probability measure on (spanning) subgraphs of $G$, where the
subgraph $(V(G),E)$ has probability proportional to $Z(E)$.
Furthermore, given such a subgraph, we may
interpret~\eqref{eq:Z_of_E_def} as prescribing a probability measure
on spin configurations $\sigma$, whose density with respect to
$d\sigma$ is proportional to
\[ Z(E,\sigma) := \prod_{\{u,v\}\in E} f_\beta(\sigma_u,\sigma_v) .\]

\begin{remark}\label{rem:edwards_sokal}
  For the Ising model ($n=1$), the above joint distribution on the graph $(V(G), E)$ and spin configuration $\sigma$ is called the \emph{Edwards--Sokal coupling} \cite{EdwSok88}. Here, the marginal probability of $E$ is proportional to
  \begin{equation}\label{eq:FK_model}
    q^{N(E)} p^{|E|} \left(1-p\right)^{|E(\T_L^d)|\setminus |E|}\quad\text{with $q = 2$ and $p = 1 - \exp(-2\beta)$},
  \end{equation}
  where $N(E)$ stands for the number of connected components in $(V(G),E)$. Moreover, given~$E$, the spin configuration $\sigma$ is sampled by independently assigning to the vertices in each connected component of $(V(G),E)$ the same spin value, picked uniformly from $\{-1,1\}$. The marginal distribution~\eqref{eq:FK_model} of $E$ is the famous \emph{Fortuin--Kasteleyn (FK) random-cluster model}, which makes sense also for other values of $p$ and $q$ \cite{Gri06}. Both the Edwards--Sokal coupling and the FK model are available also for the more general \emph{Potts models}.

  The Edwards--Sokal coupling immediately implies that, for the Ising model, the correlation $\rho_{x,y}=\E(\sigma_x \sigma_y)$ equals the probability that $x$ is connected to $y$ in the graph $(V(G), E)$. In particular, $\rho_{x,y}$ is non-negative (as in Theorem~\ref{thm:Griffiths_first_inequality}) and, as connectivity probabilities in the FK model (with $q\ge 1$) are non-decreasing with $p$ \cite[Theorem 3.21]{Gri06}, it follows also that $\rho_{x,y}$ is non-decreasing with the inverse temperature $\beta$ (as in Theorem~\ref{thm:second_Griffiths_inequality}).
\end{remark}

\begin{remark}
Conditioned on $E$, the spin configuration $\sigma$ may be seen as a sample from the spin $O(n)$ model on the graph $(V(G),E)$ with potential $U(x) := -\log(\exp(\beta(1+x))-1)$. That is, conditioned on $E$, the distribution of $\sigma$ is given by $\mu_{(V(G),E),n,U}$.
\end{remark}

It follows from the last remark that, conditioned on $E$,
\begin{align*}
&\text{If $x \in V(G)$ then $\sigma_x$ is distributed uniformly on $\S^{n-1}$} .\\
&\text{If $x,y \in V(G)$ are not connected in $(V(G),E)$ then
$\sigma_x$ and $\sigma_y$ are independent}.
\end{align*}
Hence, we deduce that $\E(\left\langle\sigma_x,\sigma_y\right\rangle \mid E)=0$
when $x$ and $y$ are not connected in $(V(G),E)$. Since
$|\langle\sigma_x,\sigma_y\rangle| \le 1$, we obtain
\[ |\rho_{x,y}| \le \Pr(\text{$x$ and $y$ are connected in $(V(G),E)$}) ,\]
where $E$ is a random subset of $E(G)$ chosen according to the above
probability measure. Thus, to establish the decay of correlations,
it suffices to show that long connections in $E$ are very unlikely. We
first show that
\begin{equation}\label{eq:prob_of_edge}
\text{for any $e \in E(G)$ and $E_0 \subset E(G) \setminus \{e\}$,
}\quad\Pr(e \in E \mid E \setminus \{e\} = E_0) \le 1 - e^{-2\beta} .
\end{equation}
Indeed,
\[ \Pr(e \in E \mid E \setminus \{e\} = E_0)
 = \frac{Z(E_0 \cup \{e\})}{Z(E_0 \cup \{e\}) + Z(E_0)} = \frac{1}{1 + \frac{Z(E_0)}{Z(E_0 \cup \{e\})}} ,\]
and denoting $e=\{u,v\}$ and noting that $f_\beta(s,t) \le \exp(2\beta) - 1$,
\begin{align*}
\frac{Z(E_0 \cup \{e\})}{Z(E_0)}
 = \frac{\int_{\Omega} Z(E_0 \cup \{e\},\sigma) d\sigma}{\int_{\Omega} Z(E_0,\sigma) d\sigma}
 = \frac{\int_{\Omega} Z(E_0,\sigma) f_\beta(\sigma_u,\sigma_v) d\sigma}{\int_{\Omega} Z(E_0,\sigma) d\sigma}
 \le e^{2\beta} - 1 .
\end{align*}
Repeated application of~\eqref{eq:prob_of_edge} now yields that the
probability that $E$ contains any fixed $k$ edges is exponentially
small in $k$. Namely,
\[ \text{for any }e_1,\dots,e_k \in E(G),\quad \Pr(e_1,\dots,e_k \in E) \le \big(1 - e^{-2\beta} \big)^k .\]

We now specialize to the case $G = \T_L^d$ (in fact, the only property of $\T_L^d$ we use is that its maximum degree is $2d$).
Since the event that $x$ and $y$ are connected in $(V(G),E)$ implies
the existence of a simple path in $E$ of some length $k \ge \|x-y\|_1$ starting at
$x$, and since the number of such paths is at most
$2d(2d-1)^{k-1} \le 2(2d-1)^k$, we obtain
\begin{align*}
\Pr(\text{$x$ and $y$ are connected in $(V(G),E)$})
 &\le \sum_{k=\|x-y\|_1}^\infty 2(2d-1)^k (1 - e^{-2\beta})^k \\
 &\le C_{d,\beta} \Big((2d-1) (1 - e^{-2\beta})\Big)^{\|x-y\|_1} ,
\end{align*}
when $(2d-1) (1 - e^{-2\beta}) < 1$. Thus, we have established that
\[ |\rho_{x,y}| \le C_{d,\beta} \exp\left( -c_{d,\beta} \|x-y\|_1 \right) \quad\text{when }\beta < \frac{1}{2} \log \left( \frac{2d-1}{2d-2} \right) .\]

\begin{remark}\label{rem:Fisher_bound}
This gives exponential decay in dimension $d\ge 2$ whenever $\beta \le 1/4d$ and in one dimension for all finite $\beta$.

Fisher \cite{fisher1967critical} established an improved lower bound for the critical inverse temperature $\beta_c(d)$ for long-range order in the $d$-dimensional Ising model, showing that $\tanh(\beta_c(d))\ge \frac{1}{\mu(d)}$, where $\mu(d)$ is the connective constant of $\Z^d$ (the exponential growth rate of the number of self-avoiding walks of length $n$ on $\Z^d$ as $n\to\infty$). Since there are fewer self-avoiding walks than non-backtracking walks, we have the simple bound $\mu(d)\le 2d-1$, which implies that $\beta_c(d) \ge \frac{1+o(1)}{2d}$ as $d\to\infty$. A similar bound was proved by Griffiths \cite{griffiths1967correlations}. Simon \cite{simon1980mean} establishes a bound of the same type for spin $O(n)$ models with $n\ge 2$, proving the absence of spontaneous magnetization when $\beta<\frac{n}{2d}$. An upper bound with matching asymptotics as $d\to\infty$ is proved via the so-called infra-red bound in Section~\ref{sec:infra-red_bound} below.

Fisher's technique is based on a Kramers--Wannier \cite{KraWan41} expansion of the Ising model partition function. This expansion, different from \eqref{eq:high-temperature_expansion}, relates the model to a probability distribution over even subgraphs (subgraphs in which the degrees of all vertices are even). A special case of the expansion is described in Section~\ref{sec:loop-spin-relation} (see remark there).
\end{remark}

\subsection{Low-temperature Ising model - the Peierls argument}\label{sec:low_temperature_Ising}
One can approach the low-temperature Ising model using the Kramers--Wannier expansion mentioned in Remark~\ref{rem:Fisher_bound} and Section~\ref{sec:loop-spin-relation}.
Here, however, we follow a slightly different route, presenting the classical Peierls argument \cite{Pei36} which is useful in many similar contexts.

Let $G$ be a finite connected graph and let $x,y \in V(G)$ be two vertices.
We begin by noting that in the Ising model, since spins take values
in $\{-1,1\}$, we may write the correlations in the following form:
\[ \rho_{x,y} = \E(\sigma_x \sigma_y) = \Pr(\sigma_x = \sigma_y) - \Pr(\sigma_x \neq \sigma_y) = 1 - 2 \Pr(\sigma_x \neq \sigma_y) .\]
Thus, to establish a lower bound on the correlation, we must provide
an upper bound on the probability that the spins at $x$ and $y$ are
different. To this end, we require some definitions.
Given a set of vertices $A \subset V(G)$, we denote the \emph{edge-boundary}
of $A$, the set of edges in $E(G)$ with precisely one endpoint
in $A$, by $\partial A$. A \emph{contour} is a set of edges $\gamma
\subset E(G)$ such that $\gamma = \partial A$ for some $A
\subset V(G)$ satisfying that both $A$ and $A^c := V(G)
\setminus A$ are induced connected (non-empty) subgraphs of $G$. Thus, a
contour can be identified with a partition of the set of vertices of
$G$ into two connected sets. We say that $\gamma$ \emph{separates} two
vertices $x$ and $y$ if they belong to different sets of the
corresponding partition.
The \emph{length of a contour} is the number of
edges it contains.

\medbreak
\noindent{\bf Exercise.}
A set of edges $\gamma$ is a contour if and only if $\gamma$ is cutset (i.e., the removal of $\gamma$ disconnects the graph) which is minimal with respect to inclusion (i.e., no proper subset of $\gamma$ is also a cutset).
\medbreak

Let $\sigma$ be a spin configuration. We say that $\gamma$ is an
\emph{interface} (with respect to $\sigma$) if $\gamma$ is a contour
separating $x$ and $y$ such that
\[ \sigma_u \neq \sigma_v \quad\text{ for all }\{u,v\} \in \gamma .\]
The first step in the proof is the following observation:
\begin{equation}\label{eq:low_temp_Ising_interface_exists}
\text{if $\sigma_x \neq \sigma_y$ then there exists an interface.}
\end{equation}
Indeed, if $\sigma_x \neq \sigma_y$ then the connected component of
$\{ u \in V(G) : \sigma_u = \sigma_x \}$ containing $x$, which
we denote by $B$, does not contain $y$. Hence, if we denote the
connected component of $B^c$ containing $y$ by $A$, then $\gamma :=
\partial A$ is a contour separating $x$ and $y$. Moreover, it is
easy to check that $\sigma_u = \sigma_x$ and $\sigma_v = \sigma_y$
for all $\{u,v\} \in \gamma$ such that $u \in A^c$ and $v \in A$, so
that $\gamma$ is an interface.

Next, we show that for any fixed contour $\gamma$ of length $k$,
\begin{equation}\label{eq:low_temp_Ising_prob_of_interface}
\Pr(\gamma\text{ is an interface}) \le e^{-2\beta k} .
\end{equation}
To see this, let $\{A,A^c\}$ be the partition corresponding to
$\gamma$ and, given a spin configuration $\sigma$, consider the
modified spin configuration $\sigma'$ in which the spins in $A$ are
flipped, i.e.,
\[ \sigma'_u := \begin{cases}
 -\sigma_u &\text{if } u \in A\\
 \sigma_u &\text{if } u \in A^c
\end{cases} .\]
Observe that if $\gamma$ is an interface with respect to $\sigma$
then
\[ \sum_{\{u,v\}\in E(G)}\sigma'_u \sigma'_v - \sum_{\{u,v\}\in E(G)}\sigma_u \sigma_v = \sum_{\{u,v\}\in \gamma} (\sigma'_u \sigma'_v - \sigma_u \sigma_v) = 2 |\gamma| . \]
Thus, denoting $F := \{ \sigma \in \Omega : \gamma\text{ is an interface with
respect to }\sigma \}$ and noting that $\sigma \mapsto \sigma'$ is
injective on $\Omega$ (in fact, an involution of $\Omega$), we have
\begin{align*}
\Pr(\gamma\text{ is an interface})
 &= \frac{\sum_{\sigma \in F} \exp \left[\beta \sum_{\{u,v\}\in E(G)}\sigma_u \sigma_v\right]}{\sum_{\sigma \in \Omega} \exp \left[\beta \sum_{\{u,v\}\in E(G)}\sigma_u \sigma_v\right]} \\
 &\le \frac{\sum_{\sigma \in F} \exp \left[\beta \sum_{\{u,v\}\in E(G)}\sigma_u \sigma_v\right]}{\sum_{\sigma \in F} \exp \left[\beta \sum_{\{u,v\}\in E(G)}\sigma'_u \sigma'_v\right]} = e^{-2\beta |\gamma|} .
\end{align*}

The final ingredient in the proof is an upper bound on the number of contours of a given length.
For this, we henceforth restrict ourselves to the case $G = \T_L^d$, for which we use the following fact:
\begin{equation}\label{eq:low_temp_Ising_number_of_contours}
\text{The number of contours of length $k$ separating two given vertices is at most }e^{C_d k} .
\end{equation}
The proof of this fact consists of the following two lemmas.

\begin{lemma}
    Let $\gamma$ be a set of edges and consider the graph $\mathcal{G}_\gamma$ on $\gamma$ in which two edges $e,f \in \gamma$ are adjacent if the $(d-1)$-dimensional faces corresponding to $e$ and $f$ share a common $(d-2)$-dimensional face. If $\gamma$ is a contour then either $\mathcal{G}_\gamma$ is connected or every connected component of $\mathcal{G}_\gamma$ has size at least $L^{d-1}$.
\end{lemma}

Although intuitively clear, the proof of the above lemma is not
completely straightforward. Tim{\'a}r gave a proof~\cite{Tim13} of the analogous statement in $\Z^d$ (in which case the graph $\mathcal{G}_\gamma$ is always connected) via elementary graph-theoretical methods. In our case, there is an additional complication due to the topology of the torus (indeed, the graph $\mathcal{G}_\gamma$ need not be connected - although it can have only two connected components - a fact for which we do not have a simple proof). We refer the reader to~\cite{Pel10} for a proof.

\begin{lemma}\label{lem:count_connected_subgraphs}
    Let $\mathcal{G}$ be a graph with maximum degree $\Delta$. The number of connected subsets of $V(\mathcal{G})$ which have size $k$ and contain a given vertex is at most $a(\Delta)^k$, where $a(\Delta)$ is a positive constant depending only on $\Delta$.
\end{lemma}

This lemma has several simple proofs. One may for instance use a depth-first-search algorithm to provide a proof with the constant $a(\Delta)=\Delta^2$. We refer the reader to~\cite[Chapter~45]{Bol06} for a proof yielding the constant $a(\Delta)=e(\Delta-1)$ (which is optimal when $\Delta \ge 3$ as can be seen by considering the case when $\mathcal{G}$ is a regular tree).

\medbreak
\noindent{\bf Exercise.}
Deduce fact~\eqref{eq:low_temp_Ising_number_of_contours} from the two lemmas.
\medbreak

Finally, putting together~\eqref{eq:low_temp_Ising_interface_exists}, \eqref{eq:low_temp_Ising_prob_of_interface} and~\eqref{eq:low_temp_Ising_number_of_contours}, when $\beta \ge C_d$, we obtain
\begin{align*}
\Pr(\sigma_x \neq \sigma_y)
 \le \Pr(\text{there exists an interface})
 &\le \sum_{\substack{\gamma\text{ contour}\\\text{separating $x$ and $y$}}} \Pr(\gamma\text{ is an interface}) \\
 &\le \sum_{k=1}^\infty e^{C_d k} e^{-2\beta k} \le C_d e^{-2\beta} .
\end{align*}
Thus, in terms of correlations, we have established that
\[ \rho_{x,y} \ge 1 - C_d e^{-2\beta} \ge c_{d,\beta} \quad\text{when }\beta \ge C_d .\]

\begin{remark}
Specializing Lemma~\ref{lem:count_connected_subgraphs} to the
relevant graph in our situation, one may obtain an improved and explicit bound of $\exp(C k \log (d+1)/d)$
on the number of contours of length $k$ separating two given
vertices~\cite{LebMaz98, BalBol07}. This gives that $\beta_c(d) \le
C \log(d+1)/d$. In fact, the correct asymptotic value is $\beta_c(d)
\sim 1/2d$ as $d\to\infty$, as follows by combining Fisher's bound in Remark~\ref{rem:Fisher_bound} with Theorem~\ref{thm:long_range_order} below.

Aizenman, Bricmont and Lebowitz \cite{aizenman1987percolation} point out that a gap between the true value of $\beta_c(d)$ and the bound on $\beta_c(d)$ obtained from the Peierls argument is unavoidable in high dimensions. They point out that the Peierls argument, when it applies, excludes the possibility of  \emph{minority percolation}. That is, the possibility to have an infinite connected component of $-1$ spins in the infinite-volume limit obtained with $+1$ boundary conditions. However, as they show, such minority percolation does occur in high dimensions when $\beta\le c\frac{\log d}{d}$, yielding a lower bound on the minimal inverse temperature at which the Peierls argument applies.
\end{remark}

\subsection{No long-range order in two dimensional models with continuous symmetry - the Mermin--Wagner theorem}\label{sec:Mermin-Wagner}
In this section, we establish power-law decay of correlations for
the two-dimensional spin~$O(n)$ model with $n\ge 2$ at any positive
temperature. The proof applies in the generality of the spin $O(n)$
model with potential $U$, where $U$ satisfies certain assumptions,
and it is convenient to present it in this context, to highlight the
core parts of the argument. The fact that there is no long-range order was first established by Mermin and Wagner~\cite{mermin1966absence, mermin1967absence}\footnote{A related intuition was mentioned earlier by Herring and Kittel \cite[Footnote 8a]{herring1951theory}.}, with later works providing upper bounds on the rate of decay of the correlations. Power-law decay of correlations for the standard XY model was first established by McBryan and Spencer~\cite{McBSpe77} who used analytic function techniques. The following theorem which generalizes the result to $C^2$ potentials was subsequently proved
by Shlosman \cite{shlosman1978decrease} using methods developed by Dobrushin and Shlosman~\cite{dobrushin1975absence}.
\begin{theorem}\label{thm:Mermin_Wagner_with_potential}
Let $n\ge 2$. Let $U:[-1,1]\to\R$ be twice continuously differentiable.
Suppose that $\sigma:V(\T_L^2)\to\mathbb S^{n-1}$ is
randomly sampled from the two-dimensional spin $O(n)$ model with
potential $U$ (see \eqref{eq:spin_O_n_potential_U}). Then there exist $C_{n,U}, c_{n,U}>0$ such that
\begin{equation}\label{eq:Mermin_Wagner_to_prove}
  |\rho_{x,y}|=|\E(\left\langle\sigma_x,\sigma_y\right\rangle)| \le C_{n,U}\|x-y\|_1^{-c_{n,U}}\quad\text{for all $x,y\in V(\T_L^2)$}.
\end{equation}
\end{theorem}
The proof presented below combines elements of the
Dobrushin--Shlosman~\cite{dobrushin1975absence} and Pfister~\cite{pfister1981symmetry} approaches to the Mermin--Wagner
theorem. The idea to combine the approaches is introduced in a
forthcoming paper of Gagnebin, Mi\l o\'s and Peled \cite{Gagnebin2018},
where it is pushed further to prove power-law decay of correlations
for \emph{any} measurable potential $U$ satisfying only very mild integrability
assumptions. The work \cite{Gagnebin2018} relies further on ideas used
by Ioffe--Shlosman--Velenik \cite{ioffe20022d}, Richthammer \cite{richthammer2007translation} and Mi\l o\'s--Peled \cite{milos2015delocalization}.

For simplicity, we will prove
Theorem~\ref{thm:Mermin_Wagner_with_potential} in the special case
that $n=2$, $x = (1,0)$ and $y = (2^m, 0)$ for some positive integer $m$
(assuming, implicitly, that $L\ge 2^m$). We briefly explain the
necessary modifications for the general case after the proof.

Fix a $C^2$ function $U:[-1,1]\to\R$.
Suppose that $\sigma:V(\T_L^2)\to \mathbb S^1$ is randomly sampled
from the two-dimensional spin $O(2)$ model with potential $U$. It is
convenient to parametrize configurations differently: Identifying
$\mathbb S^1$ with the unit circle in the complex plane, we consider
the angle $\theta_v$ that each vector $\sigma_v$ forms with respect
to the $x$-axis. Precisely, for the rest of the argument, we let
$\theta:V(\T_L^2)\to[-\pi,\pi)$ be randomly sampled from the
probability density
\begin{equation}\label{eq:density_for_theta}
  t(\phi) := \frac{1}{Z}\exp \left[-\sum_{\{u,v\}\in
E(\T_L^2)}U(\cos(\phi_u - \phi_v))\right]
  \prod_{v\in V(\T_L^2)} \one_{(\phi_v\in [-\pi,\pi))},
\end{equation}
with respect to the product uniform measure, where $Z$ is a normalization constant. One checks simply that then
$(\sigma_v)$ is equal in distribution to $(\exp(i\theta_v))$. Thus,
with our choice of the vertices $x$ and $y$, the estimate
\eqref{eq:Mermin_Wagner_to_prove} that we would like to prove
becomes
\begin{equation}\label{eq:correlations_with_theta}
  |\rho_{(1,0),(2^m,0)}| = |\E(\cos(\theta_{(2^m,0)}-\theta_{(1,0)}))| \le C_{n,U}\cdot 2^{-c_{n,U}\cdot m}.
\end{equation}

\medbreak
\noindent{\bf Step 1: Product of conditional correlations.} We start by pointing out a conditional independence
property inherent in the distribution of $\theta$, which is a
consequence of the domain Markov property and the fact that the
interaction term in \eqref{eq:density_for_theta} depends only on the
difference of angles in $\phi$ (the gradient of $\phi$). This part
of the argument is inspired by the technique of Dobrushin and
Shlosman~\cite{dobrushin1975absence}.

We divide the domain into ``layers'', where the $\ell$-th layer, $0 \le \ell \le m$, corresponds to distance $2^\ell$ from the origin. Denote the values and the gradients of $\theta$ on the $\ell$-th layer by
\[ \theta_{=\ell} := \left(\theta_v\colon \|v\|_1 = 2^\ell\right) \qquad\text{and}\qquad \nabla\theta_{=\ell} := \left(\theta_u - \theta_v \colon \|u\|_1,\|v\|_1 = 2^\ell\right) .\]
Similarly, we write $\theta_{\le \ell}$ and $\nabla\theta_{\le \ell}$ for the values/gradients of $\theta$ inside the $\ell$-th layer (i.e., for $u,v$ with $\|u\|_1,\|v\|_1 \le 2^\ell$) and $\theta_{\ge\ell}$ and $\nabla\theta_{\ge\ell}$ for the values/gradients outside (i.e., for $u,v$ with $\|u\|_1,\|v\|_1 \ge 2^\ell$).

\begin{proposition}\label{prop:independence_given_gradients}
	Conditioned on $\nabla\theta_{=\ell}$, we have that $\nabla\theta_{\le \ell}$ and $\theta_{\ge \ell}$ are independent.
\end{proposition}
\begin{proof}
	Consider the random variables $\theta_{<\ell}$ and $\nabla\theta_{<\ell}$, defined in the obvious way. It is straightforward from the definition of the density~\eqref{eq:density_for_theta} of $\theta$ that, conditioned on $\theta_{\ge \ell}$, $\theta_{<\ell}$ almost surely has a density and that this density depends only on $\theta_{=\ell}$. In particular, conditioned on $\theta_{\ge \ell}$, $\nabla\theta_{<\ell}$ has a density which depends only on $\theta_{=\ell}$. Finally, since the interaction term in~\eqref{eq:density_for_theta} depends only on the gradients $\nabla\theta$, we conclude that the density of $\nabla\theta_{<\ell}$ given $\theta_{\ge \ell}$ depends only on the gradients $\nabla\theta_{=\ell}$. Therefore, $\nabla\theta_{<\ell}$, and hence $\nabla\theta_{\le\ell}$, is conditionally independent of $\theta_{\ge \ell}$ given $\nabla\theta_{=\ell}$.
\end{proof}

Proposition~\ref{prop:independence_given_gradients} implies in particular that, conditioned on $\nabla\theta_{=\ell}$, the gradients $\nabla\theta_{\le \ell}$ and $\nabla\theta_{\ge \ell}$ are independent. It follows from abstract arguments that, conditioned on $(\nabla\theta_{=k})_{1 \le k \le m}$, the gradients $\nabla\theta_{\le \ell}$ and $\nabla\theta_{\ge \ell}$ are independent. For convenience, we state this claim in a general form in the following exercise.

\medbreak
\noindent{\bf Exercise.} Suppose $X,Y,Z$ are random variables
satisfying that $X$ is conditionally independent of $Y$ given $Z$.
Then for every two measurable functions $f$ and $g$, $X$ is conditionally
independent of $Y$ given $(f(X), g(Y), Z)$.
\medbreak

In particular, conditioned on $(\nabla\theta_{=k})_{1 \le k \le m}$, the random variables $( \theta_{(2^k,0)}-\theta_{(2^{k-1},0)} )_{1 \le k \le \ell}$ and $( \theta_{(2^k,0)}-\theta_{(2^{k-1},0)} )_{\ell < k \le m}$ are independent. Since this holds for all $1 \le \ell \le m$, it follows again by abstract arguments that, conditioned on $(\nabla\theta_{=k})_{1 \le k \le m}$, the random variables $(\theta_{(2^\ell,0)}-\theta_{(2^{\ell-1},0)})_{1 \le \ell \le m}$ are mutually independent\footnote{In fact, more is true, conditioned on $(\nabla\theta_{=k})_{1 \le k \le m}$, the $\sigma$-algebras of $\nabla\theta_{\ell-1 \le \cdot \le \ell}$ are independent for $1 \le \ell \le m$, where $\nabla\theta_{\ell-1 \le \cdot \le \ell}$ is the collection of gradients $\theta_u - \theta_v$ with $2^{\ell-1} \le \|u\|_1,\|v\|_1 \le 2^\ell$.}. Once again, we state this general claim in an exercise.

\medbreak
\noindent{\bf Exercise.} Suppose $X_1,\dots,X_m,Z$ are random variables
satisfying that, for any $1 \le \ell \le m$, $(X_1,\dots,X_\ell)$ is conditionally independent of $(X_{\ell+1},\dots,X_m)$ given $Z$.
Then $(X_1,\dots,X_m)$ are mutually conditionally independent given $Z$.
\medbreak

The above conditional independence therefore allows us to reexpress the quantity of interest to us as a product of expectations in the following way:
\begin{multline}\label{eq:product_formula_for_estimates}
\E(\cos(\theta_{(2^m,0)}-\theta_{(1,0)}))
 = \Re\E\left(e^{i\left(\theta_{(2^m,0)}-\theta_{(1,0)}\right)}\right)
 = \Re\E\left(\prod_{\ell=1}^m e^{i\left(\theta_{(2^\ell,0)}-\theta_{(2^{\ell-1},0)}\right)} \right) \\
 = \Re\E\left( \prod_{\ell=1}^m \E\left(e^{i\left(\theta_{(2^\ell,0)}-\theta_{(2^{\ell-1},0)}\right)} \mid (\nabla\theta_{=k})_{1 \le k \le m} \right)\right) .
\end{multline}
This will be the starting point for our next step.

\medbreak
\noindent{\bf Step 2: Upper bound on the conditional correlations.} In this step, we estimate the individual conditional expectations in~\eqref{eq:product_formula_for_estimates}, proving that there exists an absolute constant $\eps>0$ for which, almost surely,
\begin{equation}\label{eq:fluctuations_of_gradient}
	\left|\E\left(e^{i\left(\theta_{(2^\ell,0)}-\theta_{(2^{\ell-1},0)}\right)} \mid (\nabla\theta_{=k})_{1 \le k \le m} \right)\right| \le
	1-\eps\quad\text{for all $1\le \ell\le m$},
\end{equation}
immediately implying the required bound~\eqref{eq:correlations_with_theta} as, from~\eqref{eq:product_formula_for_estimates},
\[ \left|\E(\cos(\theta_{(2^m,0)}-\theta_{(1,0)}))\right|
	\le \E\left( \prod_{\ell=1}^m \left| \E\left(e^{i\left(\theta_{(2^\ell,0)}-\theta_{(2^{\ell-1},0)}\right)} \mid (\nabla\theta_{=k})_{1 \le k \le m} \right) \right| \right) \le (1-\eps)^m .\]
This part of the argument is inspired by the technique of Pfister~\cite{pfister1981symmetry}, and the variants used in \cite{richthammer2007translation, milos2015delocalization}. The idea of introducing a \emph{spin wave} which rotates
slowly (our function $\tau$ below and its property
\eqref{eq:finite_energy_tau}) is at the heart of the Mermin--Wagner
theorem.

Define a vector-valued function $g$ on $\R^{V(\T_L^2)}$ by
\begin{equation*}
  g(\phi) := \left(\phi_u - \phi_v\colon u,v\in V(\T_L^2), ~\exists\, 1\le k\le m,~ \|u\|_1 =
  \|v\|_1=2^k \right),
\end{equation*}
so that $g(\theta)$ and $(\nabla\theta_{=k})_{1 \le k \le m}$ represent the same random variable.
Write $dm_{g_0}$ for the lower-dimensional Lebesgue measure supported on
the affine subspace of $\R^{V(\T_L^2)}$ where $g(\phi) = g_0$.
Standard facts (following from Fubini's theorem) imply that
conditioned on $g(\theta) = g_0$, for almost every value of $g_0$ (with respect to the distribution of
$g(\theta)$),
the density of $\theta$ exists with respect to $dm_{g_0}$ and is of
the form (as in \eqref{eq:density_for_theta})
\begin{align*}
  t_{g_0}(\phi) &= \frac{1}{Z_{g_0}}\exp \left[-\sum_{\{u,v\}\in
E(\T_L^2)}U(\cos(\phi_u - \phi_v))\right]
  \prod_{v\in V(\T_L^2)} \one_{(\phi_v\in [-\pi,\pi))}\\
  &= \frac{1}{Z_{g_0}}\exp \left[-\sum_{\{u,v\}\in
E(\T_L^2)}\tilde{U}(\phi_u - \phi_v)\right]
  \prod_{v\in V(\T_L^2)} \one_{(\phi_v\in [-\pi,\pi))},
\end{align*}
where $\tilde{U}:\R\to\R$ is the $2\pi$-periodic $C^2$
function defined by
\begin{equation}\label{eq:U_tilde_def}
  \tilde{U}(\alpha) := U(\cos(\alpha)) .
\end{equation}
In particular,
\begin{equation}\label{eq:tilde_U_Taylor_expansion}
  \tilde{U}(x + \delta) \le \tilde{U}(x) + \tilde{U}'(x)\delta + \frac{\sup_y \tilde{U}''(y)}{2}
  \delta^2\quad\text{for all $x,\delta\in\R$}.
\end{equation}

Fix $1\le \ell \le m$. Define a function $\tau:V(\T_L^2)\to \R$ by
\begin{equation}\label{eq:tau_def}
  \tau_v:=\begin{cases}
    1/2 & \|v\|_1\le 2^{\ell-1}\\
    1 - \frac{\|v\|_1}{2^\ell} & 2^{\ell-1} \le\|v\|_1 \le 2^\ell\\
    0 & \|v\|_1\ge 2^\ell
  \end{cases}
\end{equation}
and define for each $\phi:V(\T_L^2)\to[-\pi,\pi)$ its perturbations
$\phi^+, \phi^-:V(\T_L^2)\to[-\pi,\pi)$ by
\begin{equation}\label{eq:t_plus_minus_def}
  \phi^+_v := \phi_v + \tau_v \pmod{2\pi},\quad \phi^-_v := \phi_v -
  \tau_v \pmod{2\pi}.
\end{equation}
We shall need the following two simple properties of $\tau$ (which the reader may easily verify):
\begin{align}
  &g(\phi^+) = g(\phi^-) =
  g(\phi)\quad \text{for every $\phi:V(\T_L^2)\to\R$}, \label{eq:g_0_preserved}\\
  &\sum_{\{u,v\}\in E(\T_L^2)} (\tau_u - \tau_v)^2 \le C
  \label{eq:finite_energy_tau}
\end{align}
for some absolute constant $C$.

The following is the key calculation of the proof. For every
$\phi:V(\T_L^2)\to[-\pi,\pi)$, setting $g_0:=g(\phi)$,
\begin{multline}\label{eq:geometric_average_densities}
\sqrt{t_{g_0}(\phi^+)t_{g_0}(\phi^-)} = \frac{1}{Z_{g_0}}\exp
\Big[-\frac{1}{2}\sum_{\{u,v\}\in E(\T_L^2)}\tilde{U}(\phi_u - \phi_v +
\tau_u - \tau_v) + \tilde{U}(\phi_u - \phi_v - \tau_u +
\tau_v))\Big]\\
\stackrel{\eqref{eq:tilde_U_Taylor_expansion}}{\ge}\frac{1}{Z_{g_0}}\exp
\Big[-\sum_{\{u,v\}\in E(\T_L^2)}\tilde{U}(\phi_u - \phi_v) -
\frac{\sup_y \tilde{U}''(y)}{2} \sum_{\{u,v\}\in E(\T_L^2)} (\tau_u -
\tau_v)^2\Big] \stackrel{\eqref{eq:finite_energy_tau}}{\ge} c\cdot
t_{g_0}(\phi)
\end{multline}
for an absolute constant $c>0$.

We wish to convert the inequality
\eqref{eq:geometric_average_densities} into an inequality of
probabilities rather than densities. To this end, define for
$a\in\R$,
\begin{equation}\label{eq:E_a_def}
  E_a:=\left\{\phi:V(\T_n^2)\to[-\pi,\pi)\colon\big|\Re e^{i(\phi_{(2^\ell,0)}-\phi_{(2^{\ell-1},0)} - a)}\big|\ge \tfrac{9}{10}\right\},
\end{equation}
and, for almost every $g_0$ with respect to the distribution of
$g(\theta)$,
\begin{equation*}
  I_{a,g_0}:=\int_{E_a} \sqrt{t_{g_0}(\phi^+)t_{g_0}(\phi^-)}dm_{g_0}(\phi).
\end{equation*}
On the one hand, by~\eqref{eq:geometric_average_densities},
\begin{equation}\label{eq:prob_estimate_1}
  I_{a,g_0} \ge c\int_{E_a} t_{g_0}(\phi)dm_{g_0}(\phi) = c\cdot \P\left(\theta\in E_a\; |\; g(\theta)=g_0\right).
\end{equation}
On the other hand, the Cauchy--Schwartz inequality and a change of
variables using \eqref{eq:t_plus_minus_def} and
\eqref{eq:g_0_preserved} yields
\begin{equation}\label{eq:prob_estimate_2}
\begin{split}
  I_{a,g_0}^2&\le \int_{E_a} t_{g_0}(\phi^+) dm_{g_0}(\phi) \cdot \int_{E_a} t_{g_0}(\phi^-) dm_{g_0}(\phi)\\
  &= \P\big(\theta - \tau\in E_a\; |\; g(\theta)=g_0\big)\cdot \P\big(\theta + \tau\in E_a\; |\; g(\theta)=g_0\big).
\end{split}
\end{equation}
Putting together \eqref{eq:prob_estimate_1} and
\eqref{eq:prob_estimate_2} and recalling \eqref{eq:tau_def} and
\eqref{eq:E_a_def} we obtain that, almost surely,
\begin{align*}
&\P\left(\big|\Re e^{i(\theta_{(2^\ell,0)}-\theta_{(2^{\ell-1},0)} + 1/2 - a)}\big|\ge \tfrac{9}{10} \mid g(\theta)\right)\cdot
\P\left(\big|\Re e^{i(\phi_{(2^\ell,0)}-\phi_{(2^{\ell-1},0)} - 1/2 - a)}\big|\ge \tfrac{9}{10} \mid g(\theta)\right)\\
&\qquad\ge c^2\cdot\P\left(\big|\Re e^{i(\phi_{(2^\ell,0)}-\phi_{(2^{\ell-1},0)} - a)}\big|\ge \tfrac{9}{10} \mid g(\theta)\right)^2.
\end{align*}
As this inequality holds for any $a\in\R$, it implies that,
conditioned on $g(\theta)$,
$e^{i\left(\theta_{(2^\ell,0)}-\theta_{(2^{\ell-1},0)}\right)}$ cannot be
concentrated around any single value, proving the inequality
\eqref{eq:fluctuations_of_gradient} that we wanted to show.

\medbreak
\noindent{\bf General vertices $x$ and $y$ and larger values of $n$.}
The inequality~\eqref{eq:Mermin_Wagner_to_prove} for arbitrary vertices $x$ and $y$ follows easily from what we have already shown. Indeed, by symmetry, there is no loss of generality in assuming as before that $x=(1,0)$ and $y\neq(0,0)$. Set $m$ to be the integer satisfying that $2^m \le \|y\|_1 < 2^{m+1}$, so that it suffices to show that $|\rho_{x,y}| \le C_{n,U}\cdot 2^{-c_{n,U}\cdot m}$. Indeed, by Proposition~\ref{prop:independence_given_gradients},
\begin{equation*}
  \E(\cos(\theta_{y}-\theta_{x}))
 = \Re\E\left( \E\left(e^{i\left(\theta_y-\theta_{(2^{m},0)}\right)} \mid \nabla\theta_{=m} \right) \cdot \E\left(e^{i\left(\theta_{(2^m,0)}-\theta_x\right)} \mid \nabla\theta_{=m} \right)\right) .
\end{equation*}
Thus, the required estimate follows from~\eqref{eq:fluctuations_of_gradient} following the decomposition~\eqref{eq:product_formula_for_estimates} (done conditionally on $\nabla\theta_{=m}$).

Let us briefly explain how to adapt the proof to the case that $n\ge 3$. Write $(\sigma^1,\dots,\sigma^n)$ for the components of $\sigma$. The idea is to condition on $(\sigma^3,\dots,\sigma^n)$ and apply the previous argument to the conditional distribution of the remaining two coordinates $(\sigma^1,\sigma^2)$. In more detail, conditioned on $(\sigma^3,\dots,\sigma^n)=h$, the random variable $(\sigma^1,\sigma^2)$ almost surely has a density with respect to the product over $v \in V(\T_L^2)$ of uniform distributions on $r_v \S^1$, where $r_v:=\sqrt{1-\|h(v)\|_2^2}\in[0,1]$. Moreover, after passing to the angle representation $\theta_v$ for each $(\sigma^1_v,\sigma^2_v)$, this density has the form
\begin{equation*}
t_h(\theta) := \frac{1}{Z_h}\exp \left[-\sum_{\{u,v\}\in
	E(\T_L^2)}U\Big(r_u r_v \cos(\theta_u - \theta_v) + \langle h(u),h(v)\rangle\Big)\right]
\prod_{v\in V(\T_L^2)} \one_{(\theta_v\in [-\pi,\pi))}.
\end{equation*}
In particular, we see from this expression that, conditioned on $(\sigma^3,\dots,\sigma^n)$, the distribution of $(\sigma^1,\sigma^2)$ is invariant to global rotations and has the domain Markov property (just as in the $n=2$ case). This allows the first step of the proof to go through essentially without change. In the second step, the function $\tilde{U}$ defined in~\eqref{eq:U_tilde_def} should be replaced by a collection of functions $\tilde{U}_{\{u,v\}}$, one for each edge $\{u,v\} \in E(\T_L^2)$, defined by $\tilde{U}_{\{u,v\}}(\alpha) := U(r_u r_v \cos(\alpha) + \langle h(u),h(v)\rangle)$. It is straightforward to check that, since $U$ is a $C^2$ function and $r_u, r_v\in[0,1]$, the second derivative of $\tilde{U}_{\{u,v\}}$ is bounded above (and below) uniformly in $\{u,v\}$ and $h$. The argument in the second step of the proof now goes through as well, replacing each appearance of $\tilde{U}$ with the suitable $\tilde{U}_{\{u,v\}}$.

\subsection{Long-range order in dimensions $d\ge 3$ - the infra-red
bound}\label{sec:infra-red_bound} In this section we prove that the
spin $O(n)$ model in spatial dimensions $d\ge 3$ exhibits long-range
order at sufficiently low temperatures. This was first proved by
Fr\"ohlich, Simon and Spencer \cite{FroSimSpe76} who introduced the
method of the \emph{infra-red bound} to this end. Our exposition benefitted from the excellent `Marseille notes' of Daniel Ueltschi \cite[Lecture 2, part 2]{Ueltschi2013}, the recent book of Friedli and Velenik \cite[Chapter 10]{friedli2016statistical} and discussions with Michael Aizenman.
We prove the following result.
\begin{theorem}\label{thm:long_range_order}
For any $d\ge 3$ and any $n\ge 1$ there exists a constant $\beta_1(d,n)$ such that the following holds. Suppose $\sigma:V(\T_L^d)\to \mathbb
S^{n-1}$ is randomly sampled from the $d$-dimensional spin $O(n)$
model at inverse temperature $\beta \ge \beta_1(d,n)$. Then
\begin{equation*}
   \frac{1}{|V(\T_L^d)|^2}\sum_{x,y\in
V(\T_L^d)}\E(\left\langle\sigma_x,\sigma_y\right\rangle) \ge c_{d,n,\beta}.
\end{equation*}
Moreover, for any $d \ge 1$, $n \ge 1$ and $\beta>0$, we have the limiting inequality
\[ \liminf_{L \to \infty} \frac{1}{|V(\T_L^d)|^2} \sum_{x,y \in \Lambda} \E\left( \langle \sigma_x \sigma_y \rangle \right) \ge 1 - \frac{n}{2\beta} \int_{[0,1]^d} \frac{1}{\sum_{j=1}^d (1 - \cos(\pi t_j))}dt .\]
Lastly, the above integral is finite when $d \ge 3$ and is asymptotic to $1/d$ as $d \to \infty$.
\end{theorem}

Of course, long-range order for the Ising model ($n=1$) has already been established in
Section~\ref{sec:low_temperature_Ising} so that our main interest is
in the case of continuous spins, when $n\ge 2$. The last part of the theorem establishes long-range order in the spin $O(n)$ model for $\beta\ge \frac{n}{2d}(1+o(1))$ as $d\to\infty$. Comparing with Remark~\ref{rem:Fisher_bound}, we see that this bound has the correct asymptotic dependence on both $n$ and $d$.

The proof presented below, relying on the original paper of Fr\"ohlich, Simon and Spencer \cite{FroSimSpe76} and making use of reflection positivity, remains the main method to establish Theorem~\ref{thm:long_range_order}. Fr\"ohlich and Spencer \cite{frohlich1982massless} developed an alternative approach for the XY model ($n=2$) which relies on the duality transformation explained in Section~\ref{sec:exact_representations} below; see also \cite[Section 5.5]{Bauerschmidt2016}. Kennedy and King~\cite{kennedy1986spontaneous} provided a second alternative approach for the XY model.
However, these alternatives do not apply to the model with larger values of $n$ ($n\ge 3$) where the symmetry group acting on the spins is non-Abelian. For these larger values the only alternative to reflection positivity is due to Balaban who has made rigorous elements of the renormalization group approach to the problem in a formidable series of papers, starting with \cite{balaban1995low}; see also Dimock's review starting with \cite{dimock2013renormalization}. Nevertheless, it would be highly desirable to have additional approaches to prove continuous-symmetry breaking as many questions in this direction are still open, most prominently to establish a phase transition for the \emph{quantum} Heisenberg ferromagnet in dimensions $d\ge3$ (current techniques allow to prove this only for the \emph{antiferromagnet}; See Dyson, Lieb and Simon \cite{dyson1978phase}).

In our treatment we provide additional background on reflection positivity than strictly necessary for the proof of Theorem~\ref{thm:long_range_order} in order to place the arguments in a wider context and to highlight the use of reflection positivity as a general-purpose tool applicable in many settings. The reader is referred to \cite[Lecture 2, part 2]{Ueltschi2013} for a direct route to the proof.

\subsubsection{Introduction to reflection positivity}\label{sec:introduction_to_reflection_positivity}
In this section we provide an introduction to reflection positivity for rather general nearest-neighbor models. Extensions of the theory to certain next-nearest-neighbor and certain long-range interactions are possible and the reader is referred to \cite{frohlich1978phase, frohlich1980phase, biskup2009reflection} and \cite[Chapter 10]{friedli2016statistical} for alternative treatments.

We allow spins to take values in an arbitrary measure space $(S,\cS, \lambda)$. We also consider a general interaction between different values of spins, prescribed by a symmetric measurable function $h \colon S \times S \to [0,\infty)$ which is not essentially zero.
Here symmetric means that
\begin{equation*}
  h(a,b)=h(b,a)\quad \text{for all $a,b \in S$}
\end{equation*}
and not essentially zero means that $\iint h(a,b)d\lambda(a)d\lambda(b)>0$. For simplicity, we assume $(S,\cS, \lambda)$ to be a finite measure space and $h$ to be bounded.

The spin space $(S,\cS, \lambda)$ and the interaction $h$ define a spin model on a finite graph $G$ as follows. The set of configurations is $S^{V(G)}$ and the density of a configuration $\sigma\colon V(G) \to S$ with respect to $d\lambda(\sigma) := \prod_{v \in V(G)} d\lambda(\sigma_v)$ is
\begin{equation}\label{eq:def_density_general_interaction}
\frac{1}{Z_{G,S,\lambda,h}} \prod_{\{u,v\} \in E(G)} h(\sigma_u,\sigma_v),
\end{equation}
where the normalizing constant is given by
\begin{equation*}
  Z_{G,S,\lambda,h} := \int \prod_{\{u,v\} \in E(G)} h(\sigma_u,\sigma_v)d\lambda(\sigma).
\end{equation*}
For this definition to make sense it is required that $0 < Z_{G,S,\lambda,h} < \infty$. The upper bound follows from our assumptions that $\lambda$ is finite and $h$ is bounded. For bipartite $G$, the case of interest to us here, the lower bound follows from the assumption that $h$ is not essentially zero\footnote{It suffices to show that $\iint \prod_{i,j=1}^n h(s_i,t_j)d\lambda(s_i)d\lambda(t_j)>0$ for $n\ge 1$. Fubini's theorem reduces this to $\iint \prod_{i=1}^n h(s_i,t)d\lambda(s_i)d\lambda(t)>0$, which then follows from Fubini's theorem and the assumption on $h$.} (the assumption does not suffice for general graphs).

The spin $O(n)$ model with potential $U$ can be recovered in this setting by taking the spin space $(S,\cS,\lambda)$ to be the uniform probability space on $\S^{n-1}$ and defining the interaction $h$ by $h(a,b):=\exp(-U(\left\langle a,b\right\rangle))$.

In order to discuss reflections, the graph $G$ should have suitable symmetries. From here on, we consider only the torus graph $G=\T_L^d$. Denote the vertices of the torus by
\begin{equation*}
  \Lambda := V(\T_L^d) = \{-L+1,\dots,L\}^d.
\end{equation*}
The torus graph admits hyperplanes of reflection which pass through vertices and hyperplanes of reflection which pass through edges. We discuss these two cases separately.

\medbreak
\noindent{\bf Reflections through vertices.}
We split the vertices of the torus into two partially overlapping subsets $V_0$ and $V_1$ of vertices, the `left' and `right' halves, by defining
\[ V_0 := \big\{ v \in \Lambda : v_1 \notin \{1,\dots,L-1 \} \big\} \qquad\text{and}\qquad V_1 := \big\{ v \in \Lambda : v_1 \in \{0,1,\dots,L \} \big\},\]
where we write each $v\in \Lambda$ as $(v_1,v_2,\ldots, v_d)$. Note that $V_0 \cup V_1 = \Lambda$ and
\[ V_0 \cap V_1 = \big\{ v \in \Lambda : v_1 \in \{0,L\} \big\} .\]
Define a function $R \colon \Lambda \to \Lambda$ by
\[ Rv := \begin{cases}
	(-v_1,v_2,\dots,v_n) &\text{if }v_1 \neq L\\
	(L,v_2,\dots,v_n) &\text{if }v_1 = L
\end{cases}.\]
Thus, $R$ is the reflection through the vertices $V_0 \cap V_1$.
Note in particular that $R$ is an involution which fixes $V_0 \cap V_1$.
Geometrically, the reflection is done across the hyperplane
orthogonal to the $x$-axis which passes through the vertices having
$x$-coordinate $0$ (or equivalently, the
hyperplane passing through the vertices having $x$-coordinate $L$). One may similarly consider reflections through other planes orthogonal to one of the coordinate axes, however, for concreteness, we focus on the reflection above. We denote by $R$ also the naturally induced mapping on configurations $\sigma \in S^\Lambda$ which is defined by $(R\sigma)_v := \sigma_{Rv}$.

Let $\cF$ denote the set of bounded measurable functions $f \colon S^\Lambda \to \C$.
Let $\cF_0 \subset \cF$ be the subset of functions $f$ which depend only on the values of the spins in $V_0$, i.e., $f(\sigma)$ is determined by $\sigma|_{V_0}$.
We define a bilinear form on $\cF_0$ by
\[ (f,g) := \E\left(f(\sigma) \overline{g(R\sigma)}\right) \qquad\text{for }f,g \in \cF_0 .\]

\begin{proposition}[Reflection positivity through vertices]
	The bilinear form defined above is positive semidefinite, i.e.,
\begin{equation}\label{eq:reflection_positivity_vertices}
  (f,f) \ge 0\quad\text{ for all $f \in \cF_0$}.
\end{equation}
\end{proposition}
\begin{proof}
	The domain Markov property implies that after conditioning on $\sigma|_{V_0 \cap V_1}$ the random variables $\sigma|_{V_0}$ and $(R\sigma)|_{V_0}$ become independent and identically distributed. Thus,
	\begin{align*}
	(f,f) = \E\left(f(\sigma) \overline{f(R\sigma)}\right)
	 &= \E\left(\E\left(f(\sigma) \overline{f(R\sigma)} \mid \sigma|_{V_0 \cap V_1} \right)\right) \\
	 &= \E\left(\E\left(f(\sigma) \mid \sigma|_{V_0 \cap V_1} \right) \cdot \overline{\E\left(f(R\sigma) \mid \sigma|_{V_0 \cap V_1} \right)} \right) \\
	 &= \E\left(\left|\E\left(f(\sigma) \mid \sigma|_{V_0 \cap V_1} \right)\right|^2  \right) \ge 0 . \qedhere
	\end{align*}
\end{proof}

The reflection positivity property \eqref{eq:reflection_positivity_vertices} (used for all hyperplanes of reflection passing through vertices) implies a version of the important ``chessboard estimate''. We do not state this estimate here, as a version of it for reflections through edges is given in Proposition~\ref{prop:chessboard_estimate_edges} below, and refer instead to \cite{biskup2009reflection} and \cite[Chapter 10]{friedli2016statistical} for more details.

\medbreak
\noindent{\bf Reflections through edges.}
We split the vertices of the torus into two non-overlapping subsets $V_0$ and $V_1$ of vertices, the `left' and `right' halves, by
\[ V_0 := \big\{ v \in \Lambda : v_1 \le 0 \big\} \qquad\text{and}\qquad V_1 := \big\{ v \in \Lambda : v_1 \ge 1 \big\} .\]
Note that $V_0 \cup V_1 = \Lambda$ and that $V_0 \cap V_1 = \emptyset$.
Define a function $R \colon \Lambda \to \Lambda$ by
\[ Rv := (1-v_1,v_2,\dots,v_n) .\]
Thus, $R$ is the reflection through the edges between $V_0$ and $V_1$.
Note in particular that $R$ is an involution with no fixed points.
Geometrically, the reflection is done across the hyperplane
orthogonal to the $x$-axis which passes through the edges between
$x$-coordinate $0$ and $x$-coordinate $1$ (or equivalently, the
hyperplane passing through the edges between $x$-coordinate $L$ and
$x$-coordinate $-L+1$). One may similarly consider reflections
through other planes orthogonal to one of the coordinate axes,
however, for concreteness, we focus on the reflection above.
We again denote by $R$ also the naturally induced mapping on configurations $\sigma \in S^\Lambda$ which is defined by $(R\sigma)_v := \sigma_{Rv}$.

Let $\cF$ denote the set of bounded measurable functions $f \colon S^\Lambda \to \C$.
Let $\cF_0 \subset \cF$ be the subset of functions $f$ which depend only on the values of the spins in $V_0$, i.e., $f(\sigma)$ is determined by $\sigma|_{V_0}$.
We define a bilinear form on $\cF_0$ by
\begin{equation}\label{eq:bilinear_form_edge_reflection}
(f,g) := \E\left(f(\sigma) \overline{g(R\sigma)}\right) \qquad\text{for }f,g \in \cF_0 .
\end{equation}

\begin{proposition}[Reflection positivity through edges]\label{prop:reflection_positivity_through_edges}
	Suppose that the interaction $h$ may be written as follows: there exists a measure space $(T,\mathcal{T},\nu)$, where $\nu$ is a finite (non-negative) measure, and a bounded measurable function $\alpha \colon T \times S \to \C$ such that
\begin{equation}\label{eq:h_semidefinite_representation}
  h(a,b) = \int \alpha(t,a) \overline{\alpha(t,b)} d\nu(t) \qquad\text{for $(\lambda\times\lambda)$-almost every $a,b \in S$}.
\end{equation}
	Then the bilinear form defined above is positive semidefinite, i.e., $(f,f) \ge 0$ for all $f \in \cF_0$.
\end{proposition}

We remark that for finite spin spaces $S$ the assumption in the proposition holds if and only if the interaction $h$, regarded as a real symmetric $S\times S$ matrix, is positive semidefinite. Indeed, if $h$ has eigenvalues $(\lambda_t)$ and associated (real) eigenvectors $(\alpha_t)$ then $h(a,b) = \sum_t \alpha_t(a)\alpha_t(b)\lambda_t$ so that being positive semidefinite yields a representation of the form \eqref{eq:h_semidefinite_representation}. Conversely, having such a representation implies that $\sum_{a,b\in S} v(a)h(a,b)v(b)\ge 0$ for all $v:S\to\R$ whence $h$ is positive semidefinite. This argument may be viewed as saying that, for finite spin spaces, having a representation of the form \eqref{eq:h_semidefinite_representation} is a \emph{necessary} condition for the conclusion that $(\cdot, \cdot)$ is positive semidefinite when the graph $G$ is the single-edge graph $\T_1^1$. Further details on necessary and sufficient conditions for reflection positivity may be found in \cite{frohlich1978phase, frohlich1980phase, biskup2009reflection} and \cite[Chapter 10]{friedli2016statistical}.

\begin{proof}[Proof of Proposition~\ref{prop:reflection_positivity_through_edges}]
	By the definition~\eqref{eq:def_density_general_interaction} of the density of $\sigma$, we have
	\begin{align*}
	\E\left(f(\sigma) \overline{f(R\sigma)}\right) = \int f(\sigma) \overline{f(R\sigma)} h_0(\sigma) h_0(R\sigma) \prod_{\{u,v\} \in E(V_0,V_1)} h(\sigma_u,\sigma_v) d\lambda(\sigma) ,
	\end{align*}
	where $h_0 \in \cF_0$ accounts for the part of the interaction coming from edges within $V_0$, and $E(V_0,V_1)$ denotes the set of edges between $V_0$ and $V_1$. Using the assumption \eqref{eq:h_semidefinite_representation} and writing $\alpha_t := \alpha(t,\cdot)$, we see that the above is equal to
\begin{equation*}
  \begin{split}
    \iint f(\sigma) \overline{f(R\sigma)} h_0(\sigma) h_0(&R\sigma) \prod_{\{u,v\} \in E(V_0,V_1)} \alpha_{t_{\{u,v\}}}(\sigma_u) \overline{\alpha_{t_{\{u,v\}}}(\sigma_v)} d\nu(t_{\{u,v\}}) d\lambda(\sigma) =\\
    \int \prod_{\{u,v\} \in E(V_0,V_1)}d\nu(t_{\{u,v\}})&\int f(\sigma)h_0(\sigma)\prod_{\substack{\{u,v\} \in E(V_0,V_1)\\u\in V_0}} \alpha_{t_{\{u,v\}}}(\sigma_u)\prod_{u\in V_0} d\lambda(\sigma_u)\\
    \cdot&\int \overline{f(R\sigma)}h_0(R\sigma)\prod_{\substack{\{u,v\} \in E(V_0,V_1)\\v\in V_1}} \overline{\alpha_{t_{\{u,v\}}}(\sigma_v)}\prod_{v\in V_1} d\lambda(\sigma_v) =\\
    \int \prod_{\{u,v\} \in E(V_0,V_1)}d\nu(t_{\{u,v\}})&\bigg|\int f(\sigma)h_0(\sigma)\prod_{\substack{\{u,v\} \in E(V_0,V_1)\\ u\in V_0}} \alpha_{t_{\{u,v\}}}(\sigma_u)\prod_{u\in V_0} d\lambda(\sigma_u)\bigg|^2\ge 0,
  \end{split}
\end{equation*}
where in the second equality we used the fact that $Rv=u$ when $\{u,v\} \in E(V_0,V_1)$ to write
\begin{multline*}
  \int \overline{f(R\sigma)}h_0(R\sigma)\prod_{\substack{\{u,v\} \in E(V_0,V_1)\\v\in V_1}} \overline{\alpha_{t_{\{u,v\}}}(\sigma_v)}\prod_{v\in V_1} d\lambda(\sigma_v)=\\
  \overline{\int f(\sigma)h_0(\sigma)\prod_{\substack{\{u,v\} \in E(V_0,V_1)\\u\in V_0}} \alpha_{t_{\{u,v\}}}(\sigma_u)\prod_{u\in V_0} d\lambda(\sigma_u)}
\end{multline*}
and in the last inequality we used that $\nu$ is a non-negative measure.
\end{proof}

Let us consider an important example of a representation of the form \eqref{eq:h_semidefinite_representation}.

\smallskip
\noindent{\bf Example.} Let $S = \R^n$. Let $\tilde{h}:\R^n\to[0,\infty)$ be a continuous positive-definite function (in particular, $\tilde{h}(-x)=\tilde{h}(x)$). This is equivalent, by Bochner's theorem, to $\tilde{h}$ being the Fourier transform of a finite (non-negative) measure $\nu$ on $\R^n$. Suppose that $h \colon \R^n \times \R^n \to [0,\infty)$ is given by $h(a,b) := \tilde{h}(a-b)$. Then we may write
\begin{equation*}
  h(a,b) = \tilde{h}(a-b) = \int e^{ita}\overline{e^{itb}}d\nu(t),\quad a,b\in\R^n,
\end{equation*}
yielding a representation of the form \eqref{eq:h_semidefinite_representation}. We remark that the example generalizes to the case that $S$ is a locally compact Abelian group.

A particular function $\tilde{h}$ which we will be interested in later on is the one arising from the Gaussian interaction, namely, $\tilde{h}(x) := e^{-\frac{\beta}{2}\|x\|_2^2}$. In this case, the Fourier transform of $\tilde{h}$ is itself a scaled Gaussian density which is, in particular, non-negative. Thus $h(a,b):=\tilde{h}(a-b)$ admits a representation of the form \eqref{eq:h_semidefinite_representation}.
\medskip

The reflection positivity property allows to prove the following ``chessboard estimate''.

\begin{proposition}[Chessboard estimate]\label{prop:chessboard_estimate_edges}
	Let $\sigma$ be sampled from the density~\eqref{eq:def_density_general_interaction} and suppose that the bilinear form defined in~\eqref{eq:bilinear_form_edge_reflection} is positive semidefinite.
	Then for any collection of real-valued bounded measurable functions $(f_v)_{v \in \Lambda}$ on $(S,\cS)$, we have
	\[ \left|\E\left(\prod_{v \in \Lambda} f_v(\sigma_v) \right)\right|^{|\Lambda|} \le \prod_{w \in \Lambda} \E\left(\prod_{v \in \Lambda} f_w(\sigma_v)\right) .\]
\end{proposition}

More general versions of the chessboard estimate are available and we refer the reader once again to \cite{biskup2009reflection} and \cite[Chapter 10]{friedli2016statistical} for more details.
\begin{proof}[Proof of Proposition~\ref{prop:chessboard_estimate_edges}]
	Let $(f_v)_{v \in \Lambda}$ be real-valued bounded measurable functions on $(S,\cS)$. We first prove the following weaker inequality:
\begin{equation}\label{eq:chessboard_estimate_edges_max_form}
\left|\E\left(\prod_{v \in \Lambda} f_v(\sigma_v) \right)\right| \le \max_{w \in \Lambda}\, \E\left(\prod_{v \in \Lambda} f_w(\sigma_v)\right) .
\end{equation}
	For every $\tau \colon \Lambda \to \Lambda$, denote
	\[ P(\tau) := \E\left(\prod_{v \in \Lambda} f_{\tau(v)}(\sigma_v)\right) \quad\text{and}\quad M(\tau) := \left|\big\{ \{u,v\} \in E(\T_L^d) : \tau(u) \neq \tau(v) \big\} \right| .\]
Note that $\prod_{v \in \Lambda} f_v(\sigma_v)$ changes sign under the substitution $f_v \mapsto -f_v$ for any single $v$, while $\prod_{v \in \Lambda} f_w(\sigma_v)$ remains unchanged (since $|\Lambda|$ is even). Thus, \eqref{eq:chessboard_estimate_edges_max_form} amounts to showing that some minimizer of $M$ is a maximizer of $P$, i.e., that there exists $\tau_*$ such that $M(\tau_*)=0$ and $P(\tau) \le P(\tau_*)$ for all $\tau$. Let $\tau_*$ be a maximizer of $P$ having $M(\tau_*)$ as small as possible among such maximizers. Assume towards a contradiction that $M(\tau_*) \ge 1$. Then there exist $u,w \in \Lambda$ such that $\tau(u) \neq \tau(w)$ and $\{u,w\} \in E(\T_L^d)$. Since rotations and
	translations of $\T_L^d$ preserve the distribution of $\sigma$ we may assume without loss of generality that $u \in V_0$ and $w \in V_1$. Define two functions $F^0_\tau(\sigma) := \prod_{v \in V_0} f_{\tau(v)}(\sigma_v)$ and $F^1_\tau(\sigma) := \prod_{v \in V_1} f_{\tau(v)}(\sigma_{Rv})$, and observe that both functions belong to $\cF_0$. Thus,
	\[ P(\tau) = \E\Big(F^0_\tau(\sigma) \cdot F^1_\tau(R\sigma) \Big) = (F^0_\tau, F^1_\tau) .\]
	Since the above bilinear form is positive semidefinite by assumption, the Cauchy--Schwartz inequality and the fact that $\tau_*$ is a maximizer of $P$ imply that
	\[ P(\tau_*) \le \sqrt{(F^0_{\tau_*},F^0_{\tau_*}) \cdot (F^1_{\tau_*},F^1_{\tau_*})} = \sqrt{P(\tau_0) \cdot P(\tau_1)} \le P(\tau_*),\]
	where $\tau_0$ and $\tau_1$ are defined by $\tau_i|_{V_i}=\tau_*|_{V_i}$ and $\tau_i = \tau_i \circ R$. Thus, both $\tau_0$ and $\tau_1$ are maximizers of $P$. Since $\tau_0(u')=\tau_0(w')$ and $\tau_1(u')=\tau_1(w')$ for $\{u',w'\} \in E(\T_L^d)$, $u'\in V_0$ and $w'\in V_1$, we see that $M(\tau_0)+M(\tau_1) < 2M(\tau_*)$, which is a contradiction to the choice of $\tau_*$.
	
We now show how to obtain the proposition from~\eqref{eq:chessboard_estimate_edges_max_form}. For $w \in \Lambda$, define
\[ a_w := \E\left(\prod_{v \in \Lambda} f_w(\sigma_v)\right) ,\]
and note that~\eqref{eq:chessboard_estimate_edges_max_form} implies that $a_w \ge 0$ (as can be seen by taking all functions to be equal). Let $\eps>0$ and define functions $(g_v)_{v \in \Lambda}$ on $(S,\cS)$ by
\[ g_v(s) := \frac{f_v(s)}{(a_v+\eps)^{1/|\Lambda|}} .\]
As these functions are bounded and measurable, \eqref{eq:chessboard_estimate_edges_max_form} implies that
\[ \left|\E\left(\prod_{v \in \Lambda} f_v(\sigma_v) \right)\right| \le \prod_{v \in \Lambda} (a_v+\eps)^{1/|\Lambda|} \cdot \max_{w \in \Lambda}\, \frac{a_w}{a_w + \eps} \le \prod_{v \in \Lambda} (a_v+\eps)^{1/|\Lambda|} .\]
Letting $\eps$ tend to zero now yields the proposition.
\end{proof}

In the special case when each $f_v$ is taken to be the indicator of some $E_v \in \cS$, the chessboard estimate implies that the probability that ``$E_v$ occurs at $v$ for all $v$'' is maximized when all the sets $\{E_v\}_v$ are equal. For convenience, and as this is the only use we make of the chessboard estimate in the next section, we state this explicitly as a corollary.

\begin{cor}\label{cor:chessboard_estimate_edges_simplified}
	Let $\sigma$ be sampled from the density~\eqref{eq:def_density_general_interaction} and suppose that the bilinear form defined in~\eqref{eq:bilinear_form_edge_reflection} is positive semidefinite.
	Then for any collection of measurable sets $(E_v)_{v \in \Lambda}$ in $(S,\cS)$, we have
	\[ \Pr\Big(\sigma_v \in E_v\text{ for all }v \in \Lambda \Big) \le \max_{w \in \Lambda}\, \Pr\Big(\sigma_v \in E_w\text{ for all }v \in \Lambda \Big) .\]
\end{cor}

\subsubsection{Gaussian domination}
Recall that $\Lambda$ denotes the set of vertices of $\T_L^d$.
For $\tau \colon \Lambda \to \R^n$, denote
\begin{equation}\label{eq:W_def}
W(\tau) := \exp\left[-\frac{\beta}{2}\sum_{\{u,v\}\in E(\T_L^d)}
\|\tau_u - \tau_v\|_2^2\right] ,
\end{equation}
where $\|\cdot\|_2$ denotes the Euclidean norm of a vector.
Recall that $\Omega = \left(\S^{n-1}\right)^\Lambda$ denotes the
space of configurations of the spin $O(n)$ model on $\T_L^d$, and note that since $\|\sigma_v\|_2^2=1$ at each vertex $v$ for $\sigma\in\Omega$, the function $W$ is closely related to the density of the spin $O(n)$ model (see~\eqref{eq:spin_O_n_def}), namely,
\begin{equation}\label{eq:W_and_spin_O(n)}
\exp \left[\beta \sum_{\{u,v\}\in
	E(G)}\left\langle\sigma_u,\sigma_v\right\rangle\right] = e^{-\beta d|\Lambda|}\cdot W(\sigma) \qquad\text{for all }\sigma \in \Omega .
\end{equation}
A key part of the argument is the study of the function
$Z:\left(\R^n\right)^\Lambda\to\R$ defined by
\begin{equation*}
Z(\tau) := \int_\Omega W(\sigma+\tau) d\sigma.
\end{equation*}
Using~\eqref{eq:W_and_spin_O(n)}, we see that the
function $Z(\tau)$ at the zero function $\tau=0$ is closely related to the
partition function of the spin $O(n)$ model (see \eqref{eq:Z_def}), namely,
\begin{equation}\label{eq:Z_zero_and_spin_O(n)}
Z(0) = e^{\beta d|\Lambda|}\cdot Z^{\text{spin}}_{\T_L^d,n,\beta}.
\end{equation}
The main step in the proof of Theorem~\ref{thm:long_range_order} is
the verification of the following \emph{Gaussian domination}
inequality,
\begin{equation}\label{eq:Gaussian_domination_inequality}
Z(\tau) \le Z(0)\quad\text{for all $\tau:\Lambda\to\R^n$},
\end{equation}
which may be reinterpreted as an inequality of expectations in the spin $O(n)$ model. Indeed, if $\sigma$ is sampled from the spin $O(n)$ model on $\T_L^d$ at inverse temperature $\beta$, then, by~\eqref{eq:W_and_spin_O(n)}, \eqref{eq:Z_zero_and_spin_O(n)} and~\eqref{eq:Gaussian_domination_inequality},
\begin{equation}\label{eq:Gaussian_domination_inequality_expectation}
\E\left(\frac{W(\sigma+\tau)}{W(\sigma)}\right) = \frac{Z(\tau)}{Z(0)} \le 1 .
\end{equation}

We establish~\eqref{eq:Gaussian_domination_inequality} using the method of reflection positivity as
described in the previous section, or, more precisely, using the chessboard estimate given in Proposition~\ref{prop:chessboard_estimate_edges} and Corollary~\ref{cor:chessboard_estimate_edges_simplified}. To this end we first define a suitable spin system specified by a finite measure space $(S,\cS, \lambda)$ and bounded symmetric interaction $h \colon S \times S \to [0,\infty)$ to which we can apply the results of the previous section.

Consider the spin system on $\T_L^d$ whose configurations are pairs $\bar{\sigma} = (\sigma,\tau)$, where for each $v \in \Lambda$, the spins $\sigma_v$ and $\tau_v$ take values in $\R^n$. Let $\eta_0$ be the Lebesgue measure on a bounded open set in $\R^n$ containing the origin. Denote $S := \R^n \times \R^n$ (with Borel $\sigma$-algebra) and let $\lambda$ be the product of the uniform probability measure on $\S^{n-1}$ and $\eta_0$. Let the interaction $h \colon S \times S \to [0,\infty)$ be $h((a,a'),(b,b')) := \tilde{h}(a+a'-b-b')$, where $\tilde{h}(x) := e^{-\frac{\beta}{2}\|x\|_2^2}$. Suppose $\bar{\sigma}=(\sigma,\tau)$ is sampled according to the density~\eqref{eq:def_density_general_interaction} with respect to $d\lambda(\bar{\sigma})$, and observe that this density is exactly given by $W(\sigma+\tau)/Z_{\T_L^d,S,\lambda,h}$. In particular, the marginal distribution of $\tau$ has density $Z(\tau)/Z_{\T_L^d,S,\lambda,h}$ with respect to $\eta := \prod_{v \in V(\T_L^d)} \eta_0$. For a function $t \colon \Lambda \to \R^n$ and $\eps>0$, define the event
\[ E_{t,\eps} := \left\{ |\tau_v - t_v| < \eps\text{ for all }v \in \Lambda \right\} .\]
It follows that, for $\eta$-almost every $t$, we have
\begin{equation}\label{eq:Z_t_as_limit}
 \frac{Z(t)}{Z_{\T_L^d,S,\lambda,h}} = \lim_{\eps \downarrow 0} ~ \frac{\Pr(E_{t,\eps})}{\eta(E_{t,\eps})} = \lim_{\eps \downarrow 0} ~ (C_n\eps^n)^{-|\Lambda|} \cdot \Pr(E_{t,\eps}) ,
\end{equation}
where $C_n$ is a positive constant depending only on $n$ and the second equality uses that $\eta_0$ is supported on an open set.

Proposition~\ref{prop:reflection_positivity_through_edges} and the example following its proof imply that the bilinear form defined by \eqref{eq:bilinear_form_edge_reflection}, with $\bar{\sigma}$ substituted for $\sigma$, is positive semidefinite. Thus Corollary~\ref{cor:chessboard_estimate_edges_simplified} may be used for the distribution of $\bar{\sigma}$. It implies that for each $t \colon \Lambda \to \R^n$, $\Pr(E_{t,\eps})\le \Pr(E_{c(t),\eps})$, where $c(t):\Lambda\to\R^n$ is a constant function. Combining this with \eqref{eq:Z_t_as_limit} and using that $Z(0)=Z(c)$ for any constant $c$, we obtain $Z(t) \le Z(0)$ for $\eta$-almost every $t$. The Gaussian domination inequality~\eqref{eq:Gaussian_domination_inequality} now follows from the continuity of $Z$ and the fact that $\eta_0$ had an arbitrarily large support.

\medbreak
\noindent{\bf Where the name ``Gaussian domination'' comes from.}
Let us give a short explanation as to the why~\eqref{eq:Gaussian_domination_inequality} is referred to as Gaussian domination. In the previous section, we considered a general spin model with density~\eqref{eq:def_density_general_interaction} with respect to the product of some a priori finite measure space. In fact, even when the a priori measure is not finite, in certain cases one can still make sense of the same density. For instance, if this a priori measure space is the Lebesgue measure on $\R^n$ and the interaction $h$ is of the same form as considered above, i.e., $h(a,b) : = e^{-\frac{\beta}{2}\|a-b\|_2^2}$, then the distribution of $\sigma$ is well-defined when considered up to a global addition of a constant (i.e., $\sigma$ takes values in the quotient space $(\R^n)^\Lambda/\R^n$ in which two configurations are equivalent if they differ by a constant; alternatively, one could introduce a boundary condition by normalizing $\sigma$ to be $0$ at some vertex). This model is called the \emph{discrete Gaussian free field} (see also Section~\ref{sec:heuristic_for_Berezinskii-Kosterlitz-Thouless} below). Since the Lebesgue measure is invariant to translations, it follows that the  function $Z$ corresponding to this model satisfies $Z(\tau)=Z(0)$ for all $\tau$. For this reason, \eqref{eq:Gaussian_domination_inequality} may be viewed as a comparison to the Gaussian case. Indeed, \eqref{eq:Gaussian_domination_inequality} implies that certain quantities in the spin $O(n)$ model are dominated by the corresponding quantities in the discrete Gaussian free field. For instance, the infra-red bound given by \eqref{eq:infra_red_bound} below becomes equality in the Gaussian case.

\subsubsection{The infra-red bound}
In this section, we prove an upper bound on the Fourier transform of the correlation function.

Recall that $\Lambda = \{-L+1,\dots,L\}^d$ is the set of vertices of $\T_L^d$.
We begin by introducing the \emph{discrete Laplacian operator} $\Delta$ on $\C^\Lambda$ defined by
\[ (\Delta f)_u := \sum_{v\colon \{u,v\}\in E(\T_L^d)} (f_v - f_u), \qquad\text{for }f \in \C^\Lambda .\]
Thus, one may regard $\Delta$ as a $\Lambda \times \Lambda$ matrix given by
\[ \Delta_{xy} := \begin{cases} -2d &\text{if }x=y\\ 1 &\text{if }\{x,y\} \in E(\T_L^d) \\0 &\text{otherwise} \end{cases} .\]
Denote the inner-product on $\C^\Lambda$ by $(\cdot,\cdot)$, i.e.,
\[ (f,g) := \sum_{u \in \Lambda} f_u \overline{g_u} ,\qquad\text{for }f,g \in \C^\Lambda .\]
Recall now the \emph{discrete Green identity}:
\begin{equation*}
\sum_{\{u,v\}\in E(\T_L^d)} (f_u  - f_v) \overline{(g_u - g_v)} = (f, -\Delta g) ,\qquad\text{for }f,g \in \C^\Lambda.
\end{equation*}
With a slight abuse of notation, we also write $\Delta$ and $(\cdot,\cdot)$ for the Laplacian and inner-product on $(\C^n)^\Lambda$, so that $\Delta f = (\Delta f^1,\dots, \Delta f^n)$ and $(f,g) = \sum_{j=1}^n (f^j,g^j)$ for $f,g \in (\C^n)^\Lambda$.
Using this notation, we can rewrite~\eqref{eq:W_def} as
\[ W(\tau) = \exp\left[\tfrac{1}{2} \beta (\tau, -\Delta \tau) \right] , \qquad\text{for }\tau \in (\R^n)^\Lambda .\]
Thus, if $\sigma$ is sampled from the $d$-dimensional spin $O(n)$ model on $\T_L^d$ at inverse temperature $\beta$, then the Gaussian domination inequality~\eqref{eq:Gaussian_domination_inequality_expectation} becomes
\[ \E\left(\exp\left[-\frac{\beta}{2}\Big((\sigma+\tau, - \Delta\sigma-\Delta\tau) - (\sigma,-\Delta\sigma) \Big) \right]\right) \le 1 , \qquad\text{for }\tau \in (\R^n)^\Lambda, \]
or, equivalently, using that $\sigma$ and $\tau$ are real-valued and that $\Delta$ is symmetric,
\begin{equation}\label{eq:Gaussian_domination_inequality_laplacian_form}
\E\Big(\exp\big[\beta (\sigma, \Delta\tau)\big]\Big) \le \exp\Big[-\tfrac{1}{2} \beta (\tau,\Delta \tau) \Big] ,\qquad\text{for }\tau \in (\R^n)^\Lambda .
\end{equation}
Substituting $\alpha \tau$ in~\eqref{eq:Gaussian_domination_inequality_laplacian_form} for $\alpha>0$, and expanding both sides of the inequality using the Taylor's series for $e^t$, yields
\[ 1 + \alpha \beta \cdot \E\left((\sigma,\Delta\tau)\right) + \tfrac12 \alpha^2 \beta^2 \cdot \E\left((\sigma,\Delta\tau)^2\right) + O(\alpha^3) \le 1 - \tfrac12 \alpha^2 \beta (\tau,\Delta\tau) + O(\alpha^4) .\]
Letting $\alpha$ tend to zero, and using that $\E((\sigma,\Delta\tau))=0$ by the invariance of the measure to the transformation $\sigma \mapsto -\sigma$, we obtain
\begin{equation}\label{eq:Gaussian_domination_inequality_final_form}
\E\left((\sigma,\Delta\tau)^2\right) \le \frac{(\tau,-\Delta\tau)}{\beta} ,\qquad\text{for }\tau \in (\R^n)^\Lambda .
\end{equation}

At this point, it seems reasonable that diagonalizing the Laplacian may prove useful, and indeed we proceed to do so. As the Laplacian matrix $\Delta$ is cyclic, it is diagonalized in the \emph{Fourier basis}, which we now define. Let $\Lambda^* := \frac{\pi}{L} \Lambda$ denote the vertices of the dual torus. The Fourier basis elements are $\{ F^k \}_{k \in \Lambda^*}$, where
\[ F^k_v := e^{i \langle k,v \rangle} , \qquad k \in \Lambda^*,~v \in \Lambda ,\]
and where we use the notation $\langle \cdot,\cdot\rangle$ also for the inner-product on $\R^d$. A straightforward computation now yields that each $F^k$ is an eigenvector of $(-\Delta)$ with eigenvalue $\lambda_k$ given by
\begin{equation}\label{eq:Laplacian_eigenvalues}
\lambda_k := 2 \sum_{j=1}^d (1 - \cos(k_j)) .
\end{equation}
It is also straightforward to check that $(F^k,F^k) = |\Lambda|$ and that $(F^k,F^{k'})=0$ for $k \neq k'$, so that the Fourier basis is an orthogonal basis.
Thus, we may write any $f \in \C^\Lambda$ in this basis as
\[ f_v = \frac{1}{|\Lambda|} \sum_{k \in \Lambda^*} \hat{f}_k e^{i\langle k,v\rangle} , \qquad v \in \Lambda ,\]
where $\{\hat{f}_k\}_{k \in \Lambda^*}$ are the Fourier coefficients of $f$ given by
\[ \hat{f}_k := (f,F^k) = \sum_{v \in \Lambda} f_v e^{-i\langle k,v\rangle}, \qquad k \in \Lambda^* .\]

Returning to the inequality~\eqref{eq:Gaussian_domination_inequality_final_form}, we now substitute a particular choice for $\tau$. Let $k \in \Lambda^*$ and let $j \in \{1,\dots,n\}$. Define $\tau := e_j F^k$, i.e., $\tau_u = e^{i\left\langle k, u\right\rangle} \cdot e_j$ for all $u \in \Lambda$, where $e_j$ is the $j$-th standard basis vector in $\R^n$. Then, since $\Delta \tau = -\lambda_k \tau$, $(\tau,\tau)=|\Lambda|$ and $(\sigma,\tau)=\hat{\sigma}^j_k$, applying~\eqref{eq:Gaussian_domination_inequality_final_form} to both the real and imaginary parts of $\tau$, we obtain
\[ \E\left(|(\sigma,\Delta\tau)|^2\right) = \lambda_k^2 \cdot \E\left(|\hat{\sigma}^j_k|^2\right) \le \frac{(\tau,-\Delta\tau)}{\beta} = \frac{\lambda_k |\Lambda|}{\beta} .\]
Hence,
\begin{equation}\label{eq:infra_red_bound}
\E\left(|\hat{\sigma}^j_k|^2\right) \le \frac{|\Lambda|}{\beta \lambda_k} ,\qquad\text{for any }k \in \Lambda^* \setminus \{0\},~1 \le j \le n .
\end{equation}
This inequality is called the infra-red bound.

The inequality \eqref{eq:infra_red_bound} can be expressed as an upper bound on the Fourier transform of the two-point correlation function $\rho_{x-y}:=\E(\left\langle\sigma_x,\sigma_y\right\rangle)$. Indeed, for any $k \in \Lambda^*$,
\begin{equation*}
  \hat{\rho}_k = \sum_{v \in \Lambda} \rho_v e^{-i\langle k,v\rangle} = \frac{1}{|\Lambda|}\sum_{x,y\in \Lambda} \E(\left\langle\sigma_x,\sigma_y\right\rangle)e^{-i\langle k,(x-y)\rangle} = \sum_{j=1}^n\E\Big|\sum_{x\in \Lambda}\sigma^j_x e^{-i\langle k,x\rangle}\Big|^2 = \sum_{j=1}^n \E\left(|\hat{\sigma}^j_k|^2\right).
\end{equation*}
As we will see in the next section, it implies that at low temperature, the Fourier transform of the two-point function in the infinite-volume limit must have a non-trivial atom at $k=0$, implying long-range order.

\subsubsection{Long-range order}

By Parseval's identity,
\[ \|f\|_2^2 = (f,f) = \frac{1}{|\Lambda|} \sum_{k \in \Lambda^*} |\hat{f}_k|^2 ,\qquad\text{for }f \in \C^\Lambda .\]
Thus,
\[ 1 = \frac{1}{|\Lambda|} \sum_{v \in \Lambda} \|\sigma_v\|_2^2 = \frac{1}{|\Lambda|} \sum_{j=1}^n \|\sigma^j\|_2^2 = \frac{1}{|\Lambda|^2} \sum_{j=1}^n \sum_{k \in \Lambda^*} (\hat{\sigma}^j_k)^2 , \qquad\text{for }\sigma \in \left(\S^{n-1}\right)^\Lambda .\]
Therefore, the infra-red bound~\eqref{eq:infra_red_bound} implies that
\begin{equation}\label{eq:long_range_order_bound_on_zero_momenta}
\frac{1}{|\Lambda|^2} \sum_{j=1}^n \E\left((\hat{\sigma}^j_0)^2\right) \ge 1 - \frac{1}{|\Lambda|} \sum_{k \in \Lambda^* \setminus \{0\}} \frac{n}{\beta \lambda_k} .
\end{equation}
Note that the left-hand side of~\eqref{eq:long_range_order_bound_on_zero_momenta} is precisely the quantity we want to estimate, i.e., the quantity appearing in the statement of Theorem~\ref{thm:long_range_order}, as can be seen from
\[ \sum_{j=1}^n \E\left((\hat{\sigma}^j_0)^2\right) = \sum_{j=1}^n \E\left( \left( \sum_{v \in \Lambda} \sigma^j_v \right)^2 \right) = \sum_{j=1}^n \sum_{x,y \in \Lambda} \E\left( \sigma^j_x \sigma^j_y \right) = \sum_{x,y \in \Lambda} \E\left( \langle \sigma_x \sigma_y \rangle \right) .\]
Plugging in the value of $\lambda_k$ from~\eqref{eq:Laplacian_eigenvalues} into the right-hand side of~\eqref{eq:long_range_order_bound_on_zero_momenta}, we identify a Riemann sum, and thus obtain
\begin{align*}
\liminf_{L \to \infty} \frac{1}{|\Lambda|^2} \sum_{x,y \in \Lambda} \E\left( \langle \sigma_x \sigma_y \rangle \right) &\ge 1 - \frac{n}{2\beta (2\pi)^d} \int_{[-\pi,\pi]^d} \frac{1}{\sum_{j=1}^d (1 - \cos(t_j))}dt \\&= 1 - \frac{n}{2\beta} \int_{[0,1]^d} \frac{1}{\sum_{j=1}^d (1 - \cos(\pi t_j))}dt .
\end{align*}
This completes the proof of the moreover part of Theorem~\ref{thm:long_range_order}. To deduce the first part of the theorem,
note that the integral is finite in dimensions $d\ge 3$, since $1 - \cos(t)$ is of order $t^2$ when $|t|$ is small. Thus, in dimensions $d \ge 3$, when $\beta$ is sufficiently large, the quantity of interest, $|\Lambda|^{-2} \sum_{x,y \in \Lambda} ( \langle \sigma_x \sigma_y \rangle )$, is bounded from below uniformly in $L$ (for bounded values of $L$, we appeal directly to~\eqref{eq:long_range_order_bound_on_zero_momenta} without taking a limit).
Finally, we note that the latter integral is asymptotic to $1/d$ as $d \to \infty$, as one can deduce using the law of large numbers.

\medskip
As a final remark we note that the proof of Theorem~\ref{thm:long_range_order} adapts verbatim to other a priori single-site measures (other than the uniform measure on $\S^{n-1}$), with the only change being the bound on $\sum_j \E((\hat{\sigma}^j_0)^2)$ in~\eqref{eq:long_range_order_bound_on_zero_momenta}, due to the fact that we can no longer use Parseval's identity to obtain a simple deterministic bound on the sum of squares of the Fourier coefficients of $\sigma$. See, e.g., \cite[Section 3.2]{Bauerschmidt2016} for details.

\subsection{Slow decay of correlations in spin $O(2)$ models - heuristic for the Berezinskii--Kosterlitz--Thouless transition and a theorem of Aizenman}\label{sec:Berezinskii-Kosterlitz-Thouless}
In this section we consider the question of proving a
\emph{power-law lower bound} on the decay of correlations in the
two-dimensional spin $O(2)$ model. As described in
Section~\ref{sec:spin_O_n_results_and_conjectures}, this was
achieved for the XY model at sufficiently low temperatures in the
celebrated work of Fr\"ohlich and Spencer on the
Berezinskii--Kosterlitz--Thouless transition \cite{FroSpe81}. The proof is too
difficult to present within the scope of our notes (see~\cite{kharash2017fr} for a recent presentation) and instead we
start by giving a heuristic reason for the existence of the
transition. The heuristic suggests that a power-law lower bound on
correlations will \emph{always} hold in the spin $O(2)$ model with a
potential $U$ of bounded support (as explained below). We then
proceed by presenting a theorem of Aizenman \cite{Aiz94}, following
earlier predictions by Patrascioiu and Seiler \cite{PatSei92}, who
made rigorous a version of the last statement.

\subsubsection{Heuristic for the Berezinskii--Kosterlitz--Thouless transition and vortices in the XY model}
\label{sec:heuristic_for_Berezinskii-Kosterlitz-Thouless}
To motivate the result, let us first give a heuristic
argument for the Berezinskii--Kosterlitz--Thouless phase transition.
Let $h:V(\T_L^2)\to\R$ be a randomly sampled \emph{discrete Gaussian
free field}. By this, we mean that $h((0,0)) := 0$ and $h$ is
sampled from the probability measure
\begin{equation}\label{eq:discrete_Gaussian_free_field_def}
  \frac{1}{Z^{\text{DGFF}}_{\T_L^2,\beta}} \exp \left[-\beta \sum_{\{u,v\}\in
E(G)} (h_u - h_v)^2\right] \prod_{\substack{v\in V(\T_L^2)\\v\neq (0,0)}} dm(h_v),
\end{equation}
with $Z^{\text{DGFF}}_{\T_L^2,\beta}$ a suitable normalization constant and $dm$ standing for the Lebesgue measure on $\R$. As the expression in the exponential is a quadratic form in $h$, it follows that $h$ has a multi-dimensional Gaussian distribution with zero mean. Moreover, the matrix of this quadratic form is proportional to the graph Laplacian of $\T_L^2$, whence the covariance structure of $h$ is proportional to the Green's function of $\T_L^2$. In particular,
\begin{equation}\label{eq:var_DGFF}
  \var(h_x) = \var(h_x - h_0)\approx \frac{a}{\beta}\log\|x-y\|_1
\end{equation}
for large $\|x-y\|_1$, with a specific constant $a>0$. Now consider the random configuration $\sigma:V(\T_L^2)\to \mathbb S^1$, with $\mathbb S^1$ identified with the unit circle in the complex plane, obtained by setting
\begin{equation}\label{eq:sigma_exponential_def}
  \sigma_v := \exp(i h_v).
\end{equation}
This configuration has some features in common with a sample of the XY model (normalized to have $\sigma_{(0,0)} = 1$). Although its density is not a product of nearest-neighbor terms, one may imagine that the main contribution to it does come from nearest-neighbor interactions, at least for large $\beta$ when the differences $h_u - h_v$ of nearest neighbors tend to be small. The interaction term $-\beta(h_u-h_v)^2$ in \eqref{eq:discrete_Gaussian_free_field_def} is then rather akin to an interaction term of the form $\frac{\beta}{2}\left\langle\sigma_u,\sigma_v\right\rangle$ as in the XY model (as $\left\langle s,t\right\rangle$ is the cosine of the difference of arguments between $s$ and $t$ and one may consider its Taylor expansion around $s=t$). The main advantage in this definition of $\sigma$ is that it allows a precise calculation of correlation. Indeed, as $h_x$ has a centered Gaussian distribution with variance given by \eqref{eq:var_DGFF}, it follows that
\begin{equation}\label{eq:algebraic_decay_from_height_function}
  \rho_{x,(0,0)}:=\E(\left\langle\sigma_x,\sigma_{(0,0)}\right\rangle) = \E(\cos(h_x)) = e^{-\frac{\var(h_x)}{2}} \approx \|x-y\|_1^{-\frac{a}{\beta}},
\end{equation}
and thus $\sigma$ exhibits power-law decay of correlations.

There are many reasons why the analogy between the definition \eqref{eq:sigma_exponential_def} and samples of the XY model should not hold. Of these, the notion of vortices has been highlighted in the literature. Suppose now that $\sigma:V(\T_L^2)\to\mathbb S^1$ is an \emph{arbitrary} configuration. Associate to each directed edge $(u,v)$, where $\{u,v\}\in E(\T_L^2)$, the difference $\theta_{(u,v)}$ in the arguments of $\sigma_u$ and $\sigma_v$, with the convention that $\theta_{(u,v)} \in [-\pi, \pi)$. Call a $2\times 2$ `square' in the graph $\T_L^2$ a \emph{plaquette} (these are exactly the simple cycles of length $4$ in $\T_L^2$). For a plaquette $P$, set $s_P$ to be the sum of $\theta_{(u,v)}$ on the edges around the plaquette going in `clockwise' order, say. We necessarily have that $s_P\in\{-2\pi, 0,2\pi\}$ and one says that there is a \emph{vortex} at $P$ if $s_P\neq 0$, with charge plus or minus according to the sign of $s_P$. Vortices form an obstruction to defining a \emph{height function} $h$ for which \eqref{eq:sigma_exponential_def} holds, as one would naturally like the differences of this $h$ to be the $\theta_{(u,v)}$, but then one must have $s_P = 0$ for all plaquettes. Existence of vortices means that $h$ needs to be a multi-valued function, with a non-trivial \emph{monodromy} around plaquettes with $s_P\neq 0$.

Now take $\sigma$ to be a sample of the $XY$ model on $\T_L^2$ at
inverse temperature $\beta$. When $\beta$ is small, the model is
disordered as one may deduce from the high-temperature expansion
(Section~\ref{sec:high-temperature_expansion}) and there are
vortices of both charges in a somewhat chaotic fashion (a `plasma'
of vortices), making the analogy with the definition
\eqref{eq:sigma_exponential_def} rather weak. Indeed, in this case
there is exponential decay of correlations violating
\eqref{eq:algebraic_decay_from_height_function}. However, when
$\beta$ is large, it can be shown (e.g., by a version of the chessboard
estimate, see Section~\ref{sec:introduction_to_reflection_positivity}) that large differences $\theta_{(u,v)}$ in the angles are
rare, whence vortices are rare too. Thus, one may hope vortices to
\emph{bind} together, coming in structures of small diameter of
overall neutral charge (the smallest structure is a \emph{dipole},
having one plus and one minus vortex). When this occurs, the height
function $h$ can be defined as a single-valued function at most
vertices and one may hope that the analogy
\eqref{eq:sigma_exponential_def} is of relevance so that, in
particular, power-law decay of correlations holds. This gives a
heuristic reason for the Berezinskii--Kosterlitz--Thouless transition.

\subsubsection{Slow decay of correlations for Lipschitz spin $O(2)$
models}\label{sec:Aizenman_slow_decay_of_correlations}

The above heuristic suggests the consideration of the spin $O(2)$ model with a potential $U$ of \emph{bounded support}. By this we mean a measurable $U:[-1,1]\to(-\infty,\infty]$ (allowing here $U(r) = \infty$) which satisfies
\begin{equation*}
  U(r) = \infty\quad\text{when $r<r_0\in(-1,1)$}.
\end{equation*}
This property constrains the corresponding $O(2)$ model so that
adjacent spins have difference of arguments at most $\arccos(r_0)$.
Such a spin configuration may naturally be called \emph{Lipschitz}
(as in a Lipschitz function). If $r_0\ge 0$, the maximal difference
allowed is at most $\frac{\pi}{2}$ which implies that the spin
configuration is \emph{free of vortices} with probability one. If
indeed vortices are the reason behind the
Berezinskii--Kosterlitz--Thouless transition, then one may expect such
models to always exhibit power-law decay of correlations.
Patrascioiu and Seiler \cite{PatSei92} predicted, based on rigorous
mathematical statements and certain yet unproven conjectures, that a
phenomenon of this kind should hold. Aizenman \cite{Aiz94} then
gave a beautiful proof of a version of the above statement, which we
now proceed to present.

\begin{theorem}\label{thm:Aizenman_lower_bound_on_correlations}
Let $U:[-1,1]\to(-\infty,\infty]$ be non-increasing and satisfy
\begin{equation}\label{eq:Aizenman_theorem_angle_constraint}
  U(r) = \infty\quad\text{when $r<\frac{1}{\sqrt{2}}$}.
\end{equation}
Suppose that $\sigma:V(\T_L^2)\to \mathbb S^1$ is
randomly sampled from the two-dimensional spin $O(2)$ model with
potential $U$. Then, for any integer $1\le \ell\le L$,
\begin{equation}\label{eq:Aizenman_lower_bound_on_correlations}
  \max_{\substack{x,y\in V(\T_L^2)\\\|x-y\|_1\ge \ell}} \rho_{x,y} = \max_{\substack{x,y\in V(\T_L^2)\\\|x-y\|_1\ge \ell}} \E(\left\langle\sigma_x,\sigma_y\right\rangle) \ge \frac{1}{2\ell^2}.
\end{equation}
\end{theorem}
We make a few remarks regarding the statement. First, one would
expect that $\rho_{x,y}$ is at least a power of $\|x-y\|_1$ for all
$x,y\in V(\T_L^2)$. The bound
\eqref{eq:Aizenman_lower_bound_on_correlations} is a little weaker
in that it only shows existence of a pair $x,y$ with this property
(the proof actually gives a slightly stronger statement, see \eqref{eq:stronger_conclusion_Aizenman_theorem}
below), but is still enough to rule out exponential decay of
correlations in the sense we saw occurs at high temperatures (see
Section~\ref{sec:high-temperature_expansion}). Second, the bound
\eqref{eq:Aizenman_lower_bound_on_correlations} can be said to hold
at all temperatures in that it will continue to hold if we multiply
the potential $U$ by any constant. Third, the constraint
\eqref{eq:Aizenman_theorem_angle_constraint} is stronger than the
constraint discussed above which would prohibit vortices
($U(r)=\infty$ if $r<0$). This stronger assumption is used in the
proof and it remains open to understand the behavior with other
versions of the constraint. Lastly, the fact that correlations decay \emph{at least} as fast as a power-law under the assumptions of the theorem is a special case of the results of \cite{Gagnebin2018}.

We proceed to the proof of Theorem~\ref{thm:Aizenman_lower_bound_on_correlations}. Let $U$ be a potential as in the theorem and $\sigma:V(\T_L^2)\to \mathbb S^1$ be
randomly sampled from the two-dimensional spin $O(2)$ model with potential $U$.

\medbreak
\noindent{\bf Step 1: Passing to $\{-1,1\}$-valued random variables.}
A main idea in the proof, suggested in the work of Patrascioiu and Seiler \cite{PatSei92}, is to consider the configuration $\sigma$ conditioned on the $y$ coordinate of each spin and identify an Ising-type model which is embedded in the configuration. In fact, we have already used this same idea in Section~\ref{spin_nonnegativity_of_correlations} when proving the non-negativity of correlations for the spin $O(n)$ model with $n \ge 2$.
Recall the definitions of the signs $\eps = (\eps_v)_{v \in V(\T_L^2)}$ and the coordinate spin values $(\sigma^1,\sigma^2) = (\sigma^1_v,\sigma^2_v)_{v \in V(\T_L^2)}$ given just prior to Theorem~\ref{thm:conditioning_gives_Ising}.
Recall also that $|\sigma^1|$ is determined by $\sigma^2$ and that $\sigma$ is determined by $(\eps,\sigma^2)$.
By Theorem~\ref{thm:conditioning_gives_Ising}, we have
 \begin{equation}\label{eq:non-negativity_of_conditional_correlations}
   \E(\eps_x \eps_y\; |\; \sigma^2)\ge 0\quad\text{for every $x,y\in V(\T_L^2)$, almost surely.}
\end{equation}
Moreover, as in the proof of Theorem~\ref{thm:conditioning_gives_Ising}, we have
\begin{equation}\label{eq:correlations_in_term_of_sign_variables}
  \rho_{x,y} = 2\E(|\sigma_x^1|\cdot |\sigma_y^1|\cdot\E(\eps_x \eps_y\; |\; \sigma^2)),
\end{equation}

\medbreak
\noindent{\bf Step 2: A lower bound on correlations in terms of connectivity.} A key idea in the analysis of Aizenman \cite{Aiz94} is the consideration of the following random set of vertices
\begin{equation*}
  V_0 := \left\{v\in V(\T_L^2)\colon |\sigma_v^1| \ge \frac{1}{\sqrt{2}}\right\}.
\end{equation*}
Note that this set is measurable with respect to $\sigma^2$. Let us consider the relevance of this set to the conditional correlations $\E(\eps_x \eps_y\; |\; \sigma^2)$ discussed above.

For reasons that will become clear in the next step, we introduce a second adjacency relation on the vertices $V(\T_L^2)$. We say that $u,v\in V(\T_L^2)$ are $\boxtimes$-adjacent if $\{u,v\}\in E(\T_L^2)$ or $u,v$ are next-nearest-neighbors in $\T_L^2$ which differ in both coordinates (they are diagonal neighbors). Now observe that, almost surely,
\begin{equation*}
\text{if $u,v$ are $\boxtimes$-adjacent and both $u,v\in V_0$ then $\eps_u = \eps_v$}.
\end{equation*}
This is a consequence of the bounded support constraint \eqref{eq:Aizenman_theorem_angle_constraint} and it is here that the number $\frac{1}{\sqrt{2}}$ in that constraint is important (as we are allowing next-nearest-neighbors). Together with the non-negativity property \eqref{eq:non-negativity_of_conditional_correlations}, it follows that
\begin{equation*}
  \E(\eps_x \eps_y\; |\; \sigma^2) \ge \one(E_{x,y})\quad\text{for every $x,y\in V(\T_L^2)$, almost surely,}
\end{equation*}
where we write $\one(E_{x,y})$ for the indicator function of the event
\begin{equation*}
  E_{x,y}:=\{\text{$x$ and $y$ are connected in the graph on $V_0\subseteq V(\T_L^2)$ with the $\boxtimes$-adjacency}\}.
\end{equation*}
Plugging this relation back into the identity \eqref{eq:correlations_in_term_of_sign_variables} for the correlation $\rho_{x,y}$ shows that
\begin{equation}\label{eq:rho_lower_bound_by_E_x_y}
  \rho_{x,y} \ge 2\E(|\sigma_x^1|\cdot |\sigma_y^1| \cdot \one(E_{x,y})) \ge \P(E_{x,y}),
\end{equation}
where we used that $|\sigma_x^1|\cdot |\sigma_y^1|\ge \frac{1}{2}$ when $x,y\in V_0$. We now proceed to deduce Theorem~\ref{thm:Aizenman_lower_bound_on_correlations} from this lower bound.

\medbreak
\noindent{\bf Step 3: Duality for vertex crossings.}
Fix an integer $1\le \ell\le L$ and define the discrete square $R := \{1, \ldots, \ell\}^2 \subseteq V(\T_L^2)$.

\medbreak
\noindent{\bf Geometric fact:}
For any subset $R_0\subseteq R$, either there is a top-bottom crossing of $R$ with vertices of $R_0$ and the $\boxtimes$-adjacency or there is a left-right crossing of $R$ with vertices of $R\setminus R_0$ and the standard nearest-neighbor adjacency (that of $\T_L^2$).

\smallskip\noindent The fact is intuitive, though finding a simple proof
requires some ingenuity. We refer the reader to Tim\'ar \cite{Tim13}
for this and related statements.

Now consider the two events
\begin{align*}
  E &:= \{\text{there is a top-bottom crossing of $R$ with vertices of $V_0$ and the $\boxtimes$-adjacency}\},\\
  F &:= \{\text{there is a left-right crossing of $R$ with vertices of $V(\T_L^2)\setminus V_0$ and the standard adjacency}\}.
\end{align*}
By rotational-symmetry of the configuration $\sigma$ (its distribution is invariant under applying a global rotation of the spins), we have $\P(F) = \P(\tilde{F})$, where
\begin{equation*}
  \tilde{F}:=\{\text{there is a left-right crossing of $R$ with vertices of $V_0$ and the standard adjacency}\}.
\end{equation*}
In particular, as $R$ is a square and since it easier to be connected in the $\boxtimes$-adjacency than in the nearest-neighbor adjacency, we conclude that
\[
  \P(E) \ge \P(F).
\]
Lastly, the geometric fact implies that $\P(E\cup F) = 1$, whence
\begin{equation}\label{eq:stronger_conclusion_Aizenman_theorem}
  1 = \P(E\cup F) \le \P(E) + \P(F) \le 2\P(E) \le 2\sum_{\substack{x=(a,1),\, 1\le a\le \ell\\ y=(b,\ell),\, 1\le b\le \ell}}\P(E_{x,y})\stackrel{\eqref{eq:rho_lower_bound_by_E_x_y}}{\le} 2\sum_{\substack{x=(a,1),\, 1\le a\le \ell\\ y=(b,\ell),\, 1\le b\le \ell}}\rho_{x,y},
\end{equation}
from which Theorem~\ref{thm:Aizenman_lower_bound_on_correlations} follows.

\subsection{Exact representations}\label{sec:exact_representations}

In this section, we show that the XY model in two dimensions admits an exact representation as an integer-valued height function. Such representations are sometimes called \emph{dual models}. We mention as another example that the dual model of the \emph{Villain model} is the integer-valued (discrete) Gaussian free field. The reader may also consult~\cite[Appendix~A]{FroSpe81} or~\cite[Section 6.1]{kharash2017fr} for additional details.
We mention also in this regard that the loop $O(n)$ model, discussed in Section~\ref{sec:loop-model} below, may be regarded as an \emph{approximate} (graphical) representation for the spin $O(n)$ model (an exact representation for $n=1$); See Section~\ref{sec:loop-spin-relation} for details. Another exact representation for the spin $O(n)$ model, which is not discussed here, is the Brydges--Fr\"ohlich--Spencer random walk representation, inspired by pioneering work of Symanzik~\cite{Symanzik1968}; see~\cite{brydges1982random, fernandez2013random} for details.

We begin the treatment here in the general context of the spin $O(n)$ model with $n=2$ and potential $U$ on an arbitrary finite graph $G$ as defined in~\eqref{eq:spin_O_n_potential_U}.
As the spins take values in the unit circle, we may reparameterize the spin variables according to their angle, to obtain
\begin{equation}\label{eq:spin_exact_Z}
Z^{\text{spin}}_{G,n,U} = \int_\Omega \prod_{\{u,v\}\in
    E(G)} \exp \Big[-U(\left\langle\sigma_u,\sigma_v\right\rangle)\Big]
d\sigma = \int_{\Omega'} \prod_{\{u,v\}\in
    E(G)} g(\theta_u - \theta_v)
d\theta ,
\end{equation}
where $d\theta$ is the Lebesgue measure on $\Omega' := [0,1)^{V(G)}$ and $g \colon \R \to \R$ is the $1$-periodic function defined by $g(t) := \exp[-U(\cos(2\pi t))]$.
When the potential $U$ is sufficiently nice, $g$ has a Fourier expansion:
\[ g(t) = \sum_{k=-\infty}^\infty \hat{g}(k) e^{2\pi ikt}, \qquad\text{where}\quad \hat{g}(k) := \int_0^1 g(t) e^{-2\pi ikt} dt .\]
Note that, since $g$ is real and even, we have that $\hat{g}$ is real and symmetric.
Having in mind that we want to plug the Fourier series of $g$ into~\eqref{eq:spin_exact_Z}, we note that $\theta_u - \theta_v$ is defined for $\{u,v\} \in E(G)$ up to its sign. For this reason, it is convenient to work with the directed edges of $G$, which we denote by $\vec{E} := \{ (u,v) : \{u,v\} \in E(G) \}$. We say a function $k \colon \vec{E} \to \Z$ is \emph{anti-symmetric} if $k_{(u,v)}=-k_{(v,u)}$ for all $(u,v) \in \vec{E}$. Note that for such a function, $k_{(u,v)}(\theta_u - \theta_v)$ is well-defined for any undirected edge $\{u,v\} \in E(G)$.
Now, plugging in the Fourier series of $g$ into~\eqref{eq:spin_exact_Z} yields
\[ Z^{\text{spin}}_{G,n,U} = \sum_{\substack{k \colon \vec{E} \to \Z\\k\text{ anti-symmetric}}} \int_{\Omega'} \prod_{\{u,v\}\in E(G)} \hat{g}(k_{(u,v)}) e^{2\pi ik_{(u,v)}(\theta_u - \theta_v)}
d\theta = \sum_{\substack{k \colon \vec{E} \to \Z\\k\text{ anti-symmetric}}} \omega_k I_k ,\]
where
\[ \omega_k := \prod_{\{u,v\}\in E(G)} \hat{g}(k_{(u,v)}) \qquad\text{and}\qquad I_k := \int_{\Omega'} \prod_{\{u,v\}\in E(G)} e^{2\pi ik_{(u,v)}(\theta_u - \theta_v)} d\theta .\]
Denoting $k_u := \sum_{\{u,v\} \in E(G)} k_{(u,v)}$ for $u \in E(G)$, we may rewrite $I_k$ as
\[ I_k = \prod_{u\in V(G)} \int_{\Omega'} e^{2\pi i k_u \theta_u} d\theta .\]
From this we see that $I_k$ is either 1 or 0 according to whether $k$ is a \emph{flow}, i.e., it satisfies $k_u=0$ for all $u \in V(G)$. Therefore, we have shown that
\[ Z^{\text{spin}}_{G,n,U} = \sum_{\substack{k \colon \vec{E} \to \Z\\k\text{ flow}}} \omega_k .\]
When the weights $\omega_k$ are non-negative, we interpret this relation as prescribing a probability measure on flows, where the probability of a flow $k$ is proportional to $\omega_k$.

In order to obtain a model of height functions, we henceforth assume that $G$ is a finite planar graph (embedded in the plane). In this case, the set of flows on $G$ are in a `natural' bijection with (suitably normalized) integer-valued height functions on the dual graph of $G$.
The dual graph of $G$, denoted by $G^*$, is the planar graph obtained by
placing a vertex at the center of every face of $G$, so
that each (directed) edge $e$ in $G$ corresponds to the unique (directed) edge $e^*$ in $G^*$ which intersects $e$ (and is rotated by 90 degrees in the clockwise direction). Note that $G^*$ has a distinguished vertex $x_0$ corresponding to the unique infinite face of $G$. Let $\cF$ be the set of functions $f \colon V(G^*) \to \Z$ having $f(x_0)=0$, which we call \emph{height-functions}. Given a function $f \in \cF$, define $k^f \colon \vec{E} \to \Z$ by $k^f_{(x,y)^*} := f(x)-f(y)$. It is straightforward to check that $k^f$ is a flow and that $f \mapsto k^f$ is injective. It remains to show that any flow is obtained in this manner. Let $k$ be a flow and define $f \colon V(G^*) \to \Z$ as follows.
For any path $p=(x_0,x_1,\dots,x_m)$ in $G^*$ starting at $x_0$, we define $f(x_m) := \phi(p)$, where $\phi(p) := \sum_{i=1}^m k_{(x_{i-1},x_i)^*}$. To show that $f$ is well-defined, we must check that $\phi(p)=\phi(p')$ for any two paths $p$ and $p'$ starting at $x_0$ and ending at the same vertex. This in turn, is the same as checking that $\phi(q)=0$ for any path $q$ starting and ending at $x_0$. It is easy to see that it suffices to check this only for any cycle $q$ in a set of cycles $Q$ which generates the \emph{cycle space} of $G^*$. To this end, we use the fact that the cycle space of a planar graph is generated by the basic cycles which correspond to the faces. Thus, we may take $Q$ to be the basic cycles in $G^*$ corresponding to the \emph{vertices} of $G$. That is, for every vertex $v \in V(G)$, we have a cycle $q_v \in Q$ which consists of the dual edges $e^*$ of the edges $e$ incident to $v$.
Finally, the property $\phi(q_v)=0$ is precisely the defining property $k_u=0$ of a flow.
It is now straightforward to verify that $k^f=k$.
Thus, when $\hat{g}$ is non-negative, we obtain a probability measure on height-functions, where the probability of $f \in \cF$ is proportional to
\[ \omega_f := \prod_{\{x,y\} \in E(G^*)} \hat{g}(f(x)-f(y)) .\]

We now specialize to the XY model, i.e., the ordinary spin $O(2)$ model as defined in~\eqref{eq:spin_O_n_def}, in which case the relevant potential is $U(t) = -\beta t$ so that $g(t) = \exp(\beta \cos(2\pi t))$. In this case, the Fourier coefficients are given by the modified Bessel functions:
\[ \hat{g}(k) = I_k(\beta) := \sum_{m=0}^\infty \frac{1}{m!\cdot (m+k+1)!} (\beta/2)^{k+2m} .\]
Since these are positive, we have indeed found a random height-function representation for the XY model in two dimensions.

As mentioned above, the Villain model also admits a similar representation. The model is defined through~\eqref{eq:spin_exact_Z} by taking the function $g$ to be the ``periodized Gaussian'' given by
\[ g(t) := \sum_{k=-\infty}^\infty e^{-\beta (t+k)^2/2} .\]
In this case, the Fourier coefficients are themselves Guassian,
\[ \hat{g}(k) = \sqrt{2\pi/\beta} \cdot e^{-2 \pi^2 k^2 / \beta} ,\]
thus yielding a height-function representation for the Villain model in two dimensions.

\section{The Loop $O(n)$ model}
\label{sec:loop-model}

\subsection{Definitions}
\label{sec:loop-model-def}

Let $\HH$ denote the hexagonal lattice.
A \emph{loop} is a
finite subgraph of $\HH$ which is isomorphic to a simple cycle. A
\emph{loop configuration} is a spanning subgraph of $\HH$ in which
every vertex has even degree; see Figure~\ref{fig:loop-configs}. The
\emph{non-trivial finite} connected components of a loop
configuration are necessarily loops, however, a loop configuration
may also contain isolated vertices and infinite simple paths. We
shall often identify a loop configuration with its set of edges,
disregarding isolated vertices. A \emph{domain} $H$ is
a non-empty finite connected induced subgraph of $\HH$ whose
complement $V(\HH) \setminus V(H)$ induces a connected subgraph of
$\HH$ (in other words, it does not have ``holes''). Given a domain $H$, we
denote by $\LC(H)$ the collection of all loop configurations
$\omega$ that are contained in $H$. Finally,
for a loop configuration $\omega$, we denote by
$L(\omega)$ the number of loops in $\omega$ and by $o(\omega)$ the number of edges of $\omega$.

Let $H$ be a domain and let $n$ and $x$ be positive real numbers. The loop $O(n)$ measure on $H$ with edge
weight $x$ is the probability measure
$\Pr_{H,n,x}$ on $\LC(H)$ defined by
  \begin{equation}\label{eq:loop O n model def}
  \Pr_{H,n,x}(\omega) := \frac{x^{o(\omega)} n^{L(\omega)}}{Z_{H,n,x}^{\text{loop}}},
  \end{equation}
  where $Z_{H,n,x}^{\text{loop}}$, the partition function, is given by
  \[ Z_{H,n,x}^{\text{loop}} := \sum_{\omega \in \LC(H)} x^{o(\omega)} n^{L(\omega)} .\]

\medbreak \noindent {\bf The $x=\infty$ model.} We also
consider the limit of the loop $O(n)$ model as the edge weight $x$
tends to infinity. This means restricting the model to `optimally
packed loop configurations', i.e., loop configurations having the
maximum possible number of edges.

  Let $H$ be a domain and let $n>0$. The loop $O(n)$ measure on $H$ with edge weight $x = \infty$ is the probability measure on $\LC(H)$ defined by
\[
\Pr_{H,n,\infty}(\omega) := \lim_{x \to \infty} \Pr_{H,n,x}(\omega) =
 \begin{cases} \frac{n^{L(\omega)}}{Z_{H,n,\infty}} &\text{if }o(\omega)=o_H \\ 0 &\text{otherwise} \end{cases} ,
\]
where $o_H := \max \{ o(\omega) : \omega \in \LC(H) \}$ and $Z_{H,n,\infty}$ is the unique constant making
$\Pr_{H,n,\infty}$ a probability measure.
We note that if a loop configuration $\omega\in\LC(H)$ is
\emph{fully-packed}, i.e., every vertex in $V(H)$ has degree $2$,
then $\omega$ is optimally packed, i.e., $o(\omega)=o_H$.
In particular, if such a configuration exists for the domain $H$, then the measure $\Pr_{H,n,\infty}$ is supported on fully-packed loop configurations.

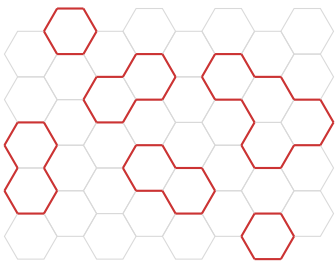
\begin{figure}
	\centering
	
	\begin{tikzpicture}[scale=0.35, every node/.style={scale=0.35}]
	\hexagonGrid[7][4][0];
	\trivialLoop[1][4];
	\trivialLoop[6][-3];
	\doubleLoopUp[0][1];
	\doubleLoopRight[2][2];
	\doubleLoopLeft[3][-1];
	\hexagonEdges[edge-on][5][0][0/1/-3,0/1/-2,0/1/-1,0/1/0,0/1/1,1/0/0,1/0/1,2/-1/0,2/-1/1,2/-1/2,2/-1/3,1/-1/2,1/-1/3,1/-1/4,1/-1/5,1/0/-2];
	\end{tikzpicture}
	
	\caption{A loop configuration is a subgraph of the hexagonal lattice in which every vertex has even degree.}
	\label{fig:loop-configs}
\end{figure}

\bigbreak
Like in the spin model, special cases of the loop $O(n)$ model have names of their own:
\begin{itemize}[noitemsep,topsep=0.5em]
\setlength\itemsep{0.25em}
\item
When $n=0$, one formally obtains the self-avoiding walk (SAW); see Section~\ref{sec:saw}.

\item
When $n=1$, the model is equivalent to the Ising model on the triangular lattice under the correspondence $x = e^{-2\beta}$ (the loops represent the interfaces between spins of different value), which in turn is equivalent via the Kramers--Wannier duality \cite{KraWan41} to an Ising model on the dual hexagonal lattice.
\begin{itemize}[label=$\circ$]	
  \item The special case $x=1$, corresponding to the Ising model at infinite temperature, is critical site percolation on the triangular lattice.
  \item The special case $x=\infty$, corresponding to the anti-ferromagnetic Ising model at zero temperature, is a uniformly picked fully-packed loop configuration, whence its complement is a uniformly picked perfect matching of the vertices in the domain. The model is thus equivalent to the dimer model.
\end{itemize}

\item
When $n \ge 2$ is an integer, the model is a marginal of a discrete random Lipschitz function on the triangular lattice. When $n=2$ this function takes integer values and when $n \ge 3$ it takes values in the $n$-regular tree. See Section~\ref{sec:loop-o-n-height} for more details.
The special case $n=2$ and $x=\infty$ is equivalent to uniform proper $4$-colorings of the triangular lattice~\cite{baxter1970colorings} (the loops are obtained from a proper coloring with colors $\{0,1,2,3\}$ as the edges bordering hexagons whose colors differ by $\pm 1$ modulo 4).

\item
When $n=\infty$ and $nx^6 = \text{const}$, the model becomes the hard-hexagon model. See Section~\ref{sec:hard-hex} for more details.

\item When $n$ is the square root of a positive integer, the model is a marginal of the dilute Potts model on the triangular lattice. See Section~\ref{sec:loop-O-n-exact-representations} for more details.

\end{itemize}

\subsection{Relation to the spin $O(n)$ model}
\label{sec:loop-spin-relation}

We reiterate that the loop $O(n)$ model is defined for any positive \emph{real}
$n$, whereas the spin $O(n)$ model is only defined for positive \emph{integer} $n$. For integer $n$, there is a connection between the
loop and the spin $O(n)$ models on a domain $H
\subset \HH$. Rewriting the partition function
$Z^{\text{spin}}_{H,n,\beta}$ given by \eqref{eq:Z_def} using the
approximation $e^t \approx 1+t$ gives
\[
\begin{split}
Z^{\text{spin}}_{H,n,\beta} =\int\displaylimits_{\Omega} \prod_{\{u,v\}\in E(H)} \exp \left[\beta \langle  \sigma_u,\sigma_v \rangle \right] \,d\sigma
\\\approx \int\displaylimits_{\Omega} \prod_{\{u,v\}\in E(H)} (1 + \beta \langle  \sigma_u,\sigma_v \rangle) \,d\sigma &= \sum_{\omega \subset E(H)} (\beta/n)^{o(\omega)} \int\displaylimits_{\Omega} \prod_{\{u,v\} \in E(\omega)} \langle \sqrt{n} \cdot \sigma_u, \sqrt{n} \cdot \sigma_v \rangle \,d\sigma,\\
& = \sum_{\omega \in \LC(H)} (\beta/n)^{o(\omega)} n^{L(\omega)},
\end{split}
\]
where the last equality follows by splitting the integral into a product of integrals on each connected component of $\omega$ and then using the following calculation.

\medbreak
\noindent {\bf Exercise.}
Let $E \subset E(\HH)$ be finite and connected. Show that
\[ \int\displaylimits_{\Omega} \prod_{\{u,v\} \in E} \langle \sqrt{n} \cdot \sigma_u, \sqrt{n} \cdot \sigma_v \rangle \,d\sigma = \begin{cases}n&\text{if $E$ is a loop}\\0&\text{otherwise}\end{cases} .\]
(see~\cite[Appendix A]{DCPSS14} for the calculation) \medbreak

Therefore, substituting $x$ for $\beta/n$, we obtain
\[
Z_{H,n,nx}^{\text{spin}} \approx Z_{H,n,x}^{\text{loop}}.
\]
In the same manner, the correlation $\rho_{u,v}$ for $u,v\in V(H)$ in the spin $O(n)$ model at inverse temperature $\beta = nx$ may be approximated as follows.
\begin{equation}\label{eq:relation spin loop}
    \rho_{u,v} = \frac{ \displaystyle\int_{\Omega} \langle \sigma_u, \sigma_v \rangle \prod_{\{w,z\}\in E(H)} \exp \left[\beta \langle \sigma_w,\sigma_z \rangle \right]}{Z_{H,n,\beta}^{\text{spin}}} \approx n\cdot \frac{\displaystyle \sum_{\lambda \in \LC(H,u,v)} x^{o(\lambda)} n^{L'(\lambda)} J(\lambda)}{Z_{H,n,x}^{\text{loop}}},
\end{equation}
where $\LC(H,u,v)$ is the set of spanning subgraphs of $H$
in which the degrees of $u$ and $v$ are odd and the degrees of all
other vertices are even. Here, for $\lambda \in \LC(H,u,v)$, $o(\lambda)$ is the number of edges of $\lambda$, $L'(\lambda)$ is the number of loops in $\lambda$ after removing an arbitrary simple path in $\lambda$ between $u$ and $v$, and
$J(\lambda):=\tfrac{3n}{n+2}$ if there are three disjoint paths in
$\lambda$ between $u$ and $v$ and $J(\lambda):=1$ otherwise (in which
case, there is a unique simple path in $\lambda$ between $u$ and
$v$).

\medbreak \noindent {\bf Exercise.} Use the approximation $e^t
\approx 1+t$ to obtain the asserted representation
in~\eqref{eq:relation spin loop} (see~\cite[Appendix A]{DCPSS14} for
the calculation). \medbreak

We remark that for $n=1$, since $e^{\beta s} = \cosh(\beta)(1 + s\cdot \tanh(\beta))$ for $s\in\{-1,1\}$, the above expansion can be made exact by choosing $x=\tanh(\beta)$. This yields an exact duality between the ferromagnetic Ising model on the hexagonal lattice and the ferromagnetic Ising model on the triangular lattice; a special case of the so-called Kramers--Wannier duality \cite{KraWan41}. Such duality maps the high-temperature region to the low-temperature region providing a self-dual point on self-dual lattices such as $\Z^2$ (and also for the hexagonal-triangular lattice pair, using an auxiliary star-triangle transformation), which turns out to be the critical point~\cite{Ons44}.

Unfortunately the above approximation is not justified for any
$x>0$ when $n>1$. Nevertheless, \eqref{eq:relation spin loop} provides a
heuristic connection between the spin and the loop $O(n)$ models and
suggests that both these models reside in the same universality
class. For this reason, it is natural to ask whether the prediction
about the absence of phase transition is valid for the loop $O(n)$
model.

\medbreak
\noindent{\bf Question:}
\label{quest:loop-On-quest}
Does the quantity on the right-hand side of~\eqref{eq:relation spin loop} decay exponentially fast in the distance between $u$ and $v$, uniformly in the domain $H$, whenever $n>2$ and $x>0$?
\medbreak

This question is partially answered in~\cite{DCPSS14}, where it is shown that for all sufficiently
large $n$ and any $x>0$, the quantity on the right-hand side
of~\eqref{eq:relation spin loop} decays exponentially fast for a
large class of domains $H$. The result is a consequence of a more
detailed understanding of the loop $O(n)$
model with large $n$, which we elaborate on in Section~\ref{sec:large-n}.

\subsection{Conjectured phase diagram and rigorous results}\label{sec:loop_O_n_phase_diagram}

\begin{figure}[!ht]
	\centering
	\begin{subfigure}[t]{.49\textwidth}
		\includegraphics[width=\textwidth]{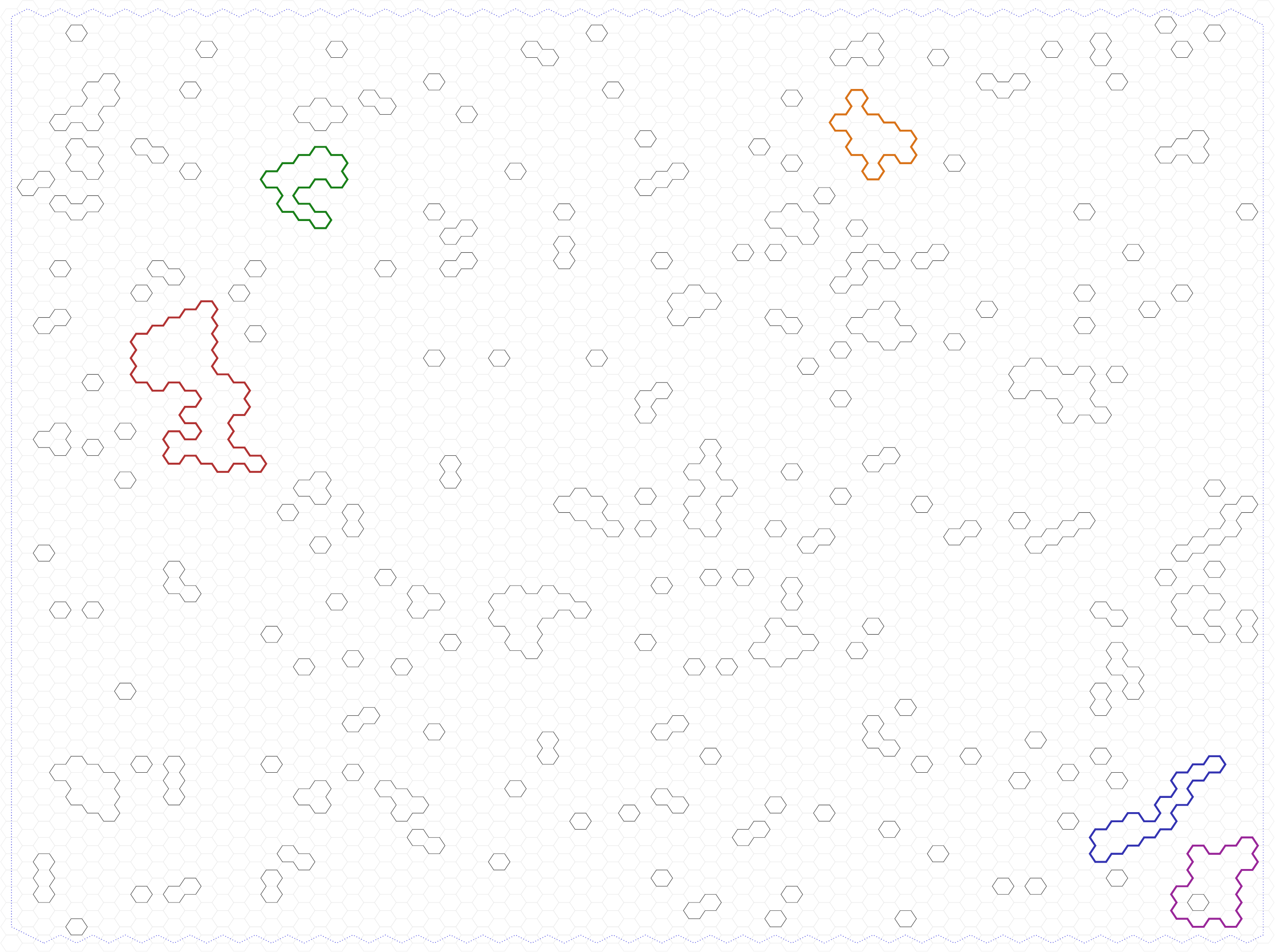}
		\caption{$n=1.4$ and $x=0.57<x_c(n)$.}
		\label{fig:loop-sample-n1.4-x0.57}
	\end{subfigure}%
	\begin{subfigure}{10pt}
		\quad
	\end{subfigure}%
	\begin{subfigure}[t]{.49\textwidth}
		\includegraphics[width=\textwidth]{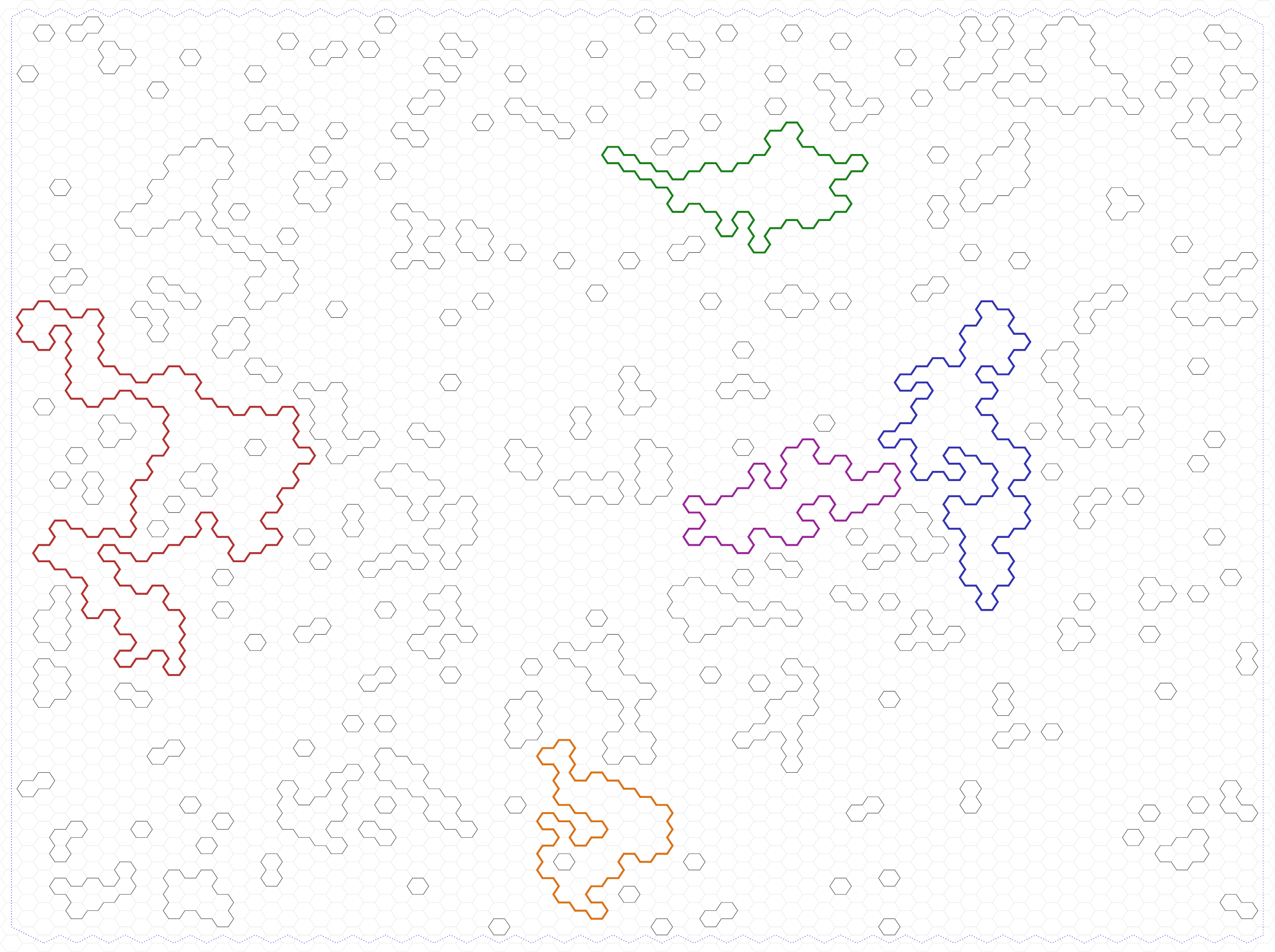}
		\caption{$n=1.4$ and $x=x_c(n) \approx 0.6$.}
		\label{fig:loop-sample-n1.4-x0.6}
	\end{subfigure}
	\medbreak
	\begin{subfigure}[t]{.49\textwidth}
		\includegraphics[width=\textwidth]{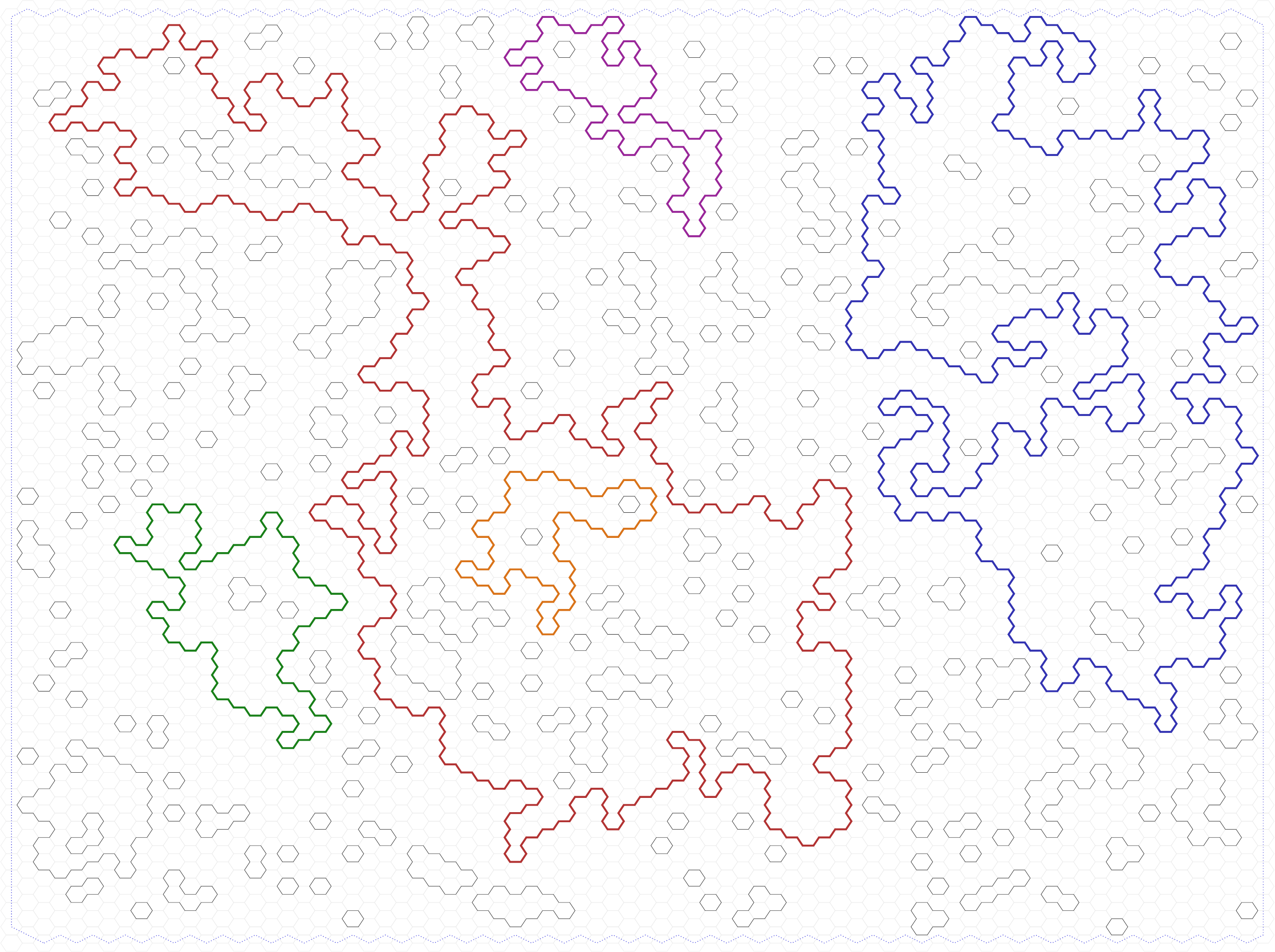}
		\caption{$n=1.4$ and $x=0.63>x_c(n)$.}
		\label{fig:loop-sample-n1.4-x0.63}
	\end{subfigure}%
	\begin{subfigure}{10pt}
		\quad
	\end{subfigure}%
	\begin{subfigure}[t]{.49\textwidth}
		\includegraphics[width=\textwidth]{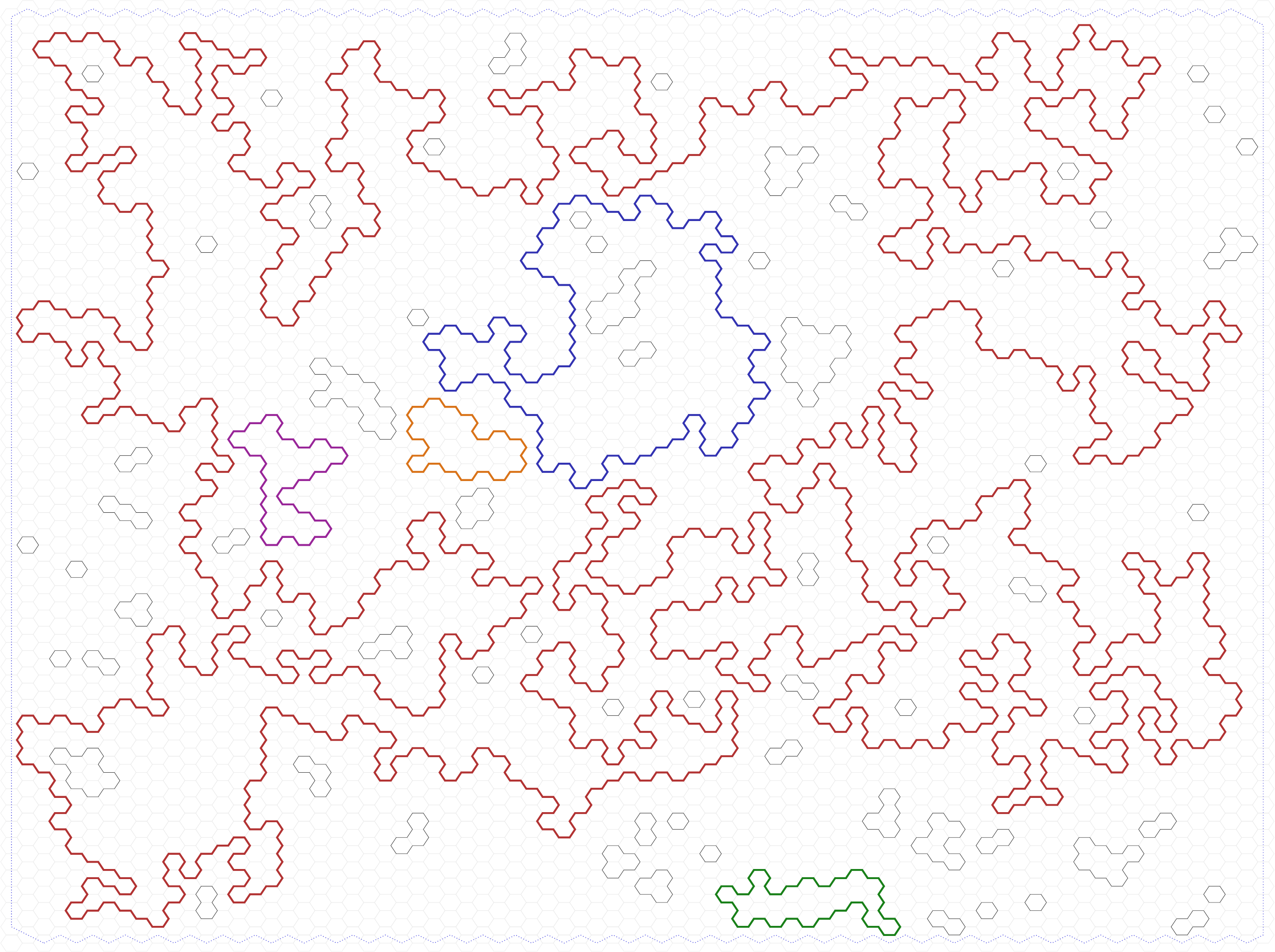}
		\caption{$n=0.5$ and $x=0.6>x_c(n)$.}
		\label{fig:loop-sample-n0.5-x0.6}
	\end{subfigure}
	\caption{Samples of random loop configurations on and around the critical line. Configurations are on a $80\times60$ rectangular-shaped domain and are sampled via Glauber dynamics for 100 million iterations started from the empty configuration. The longest loops are highlighted (from longest to shortest: red, blue, green, purple, orange).}
	\label{fig:loop-samples}
\end{figure}

\begin{figure}[!ht]
	\centering
	\begin{subfigure}[t]{.49\textwidth}
		\includegraphics[width=\textwidth]{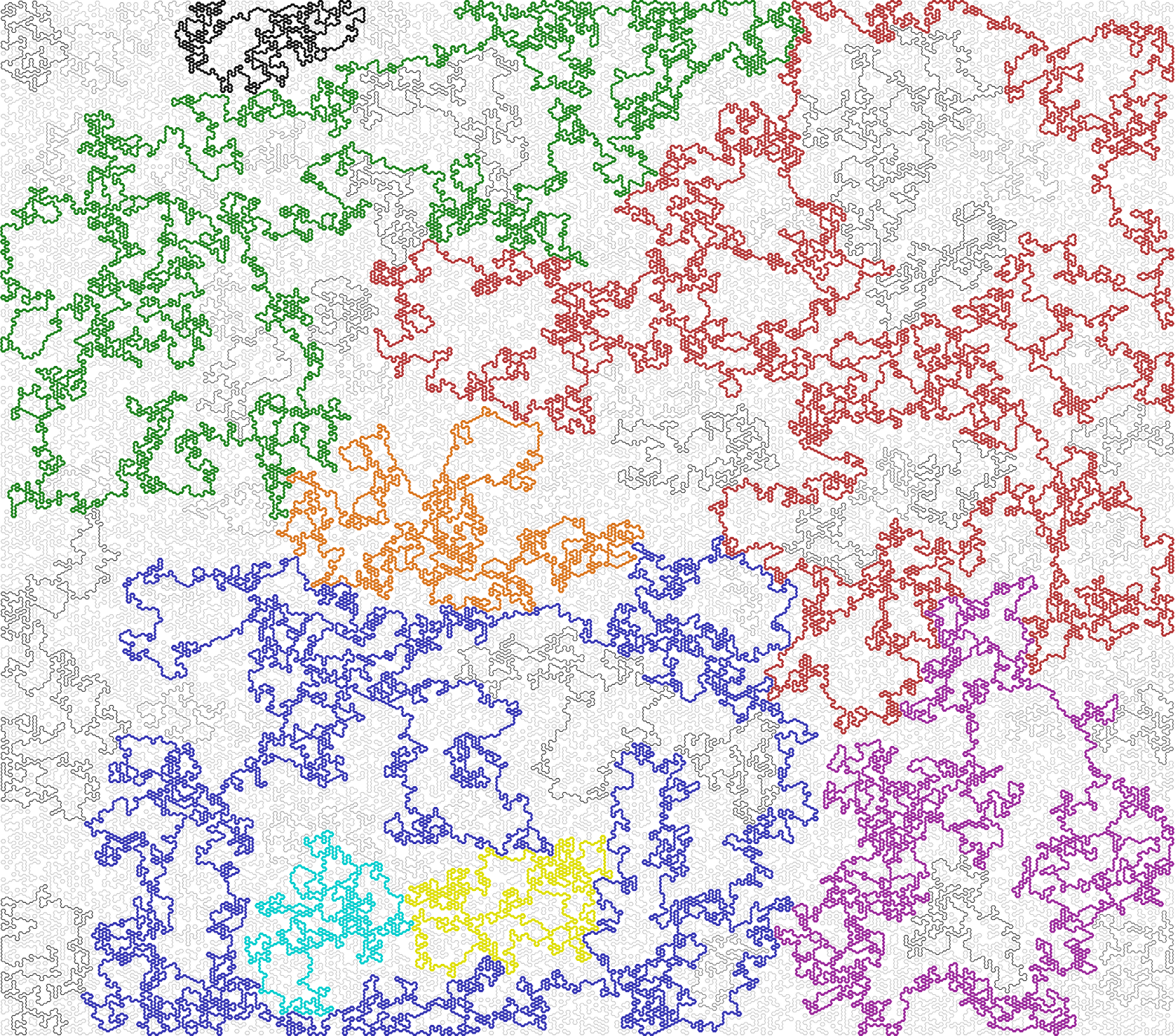}
		\caption{$n=1.5$ and $x=1$.}
		\label{fig:loop-sample-n1,5-x1}
	\end{subfigure}%
	\begin{subfigure}{10pt}
		\quad
	\end{subfigure}%
	\begin{subfigure}[t]{.49\textwidth}
		\includegraphics[width=\textwidth]{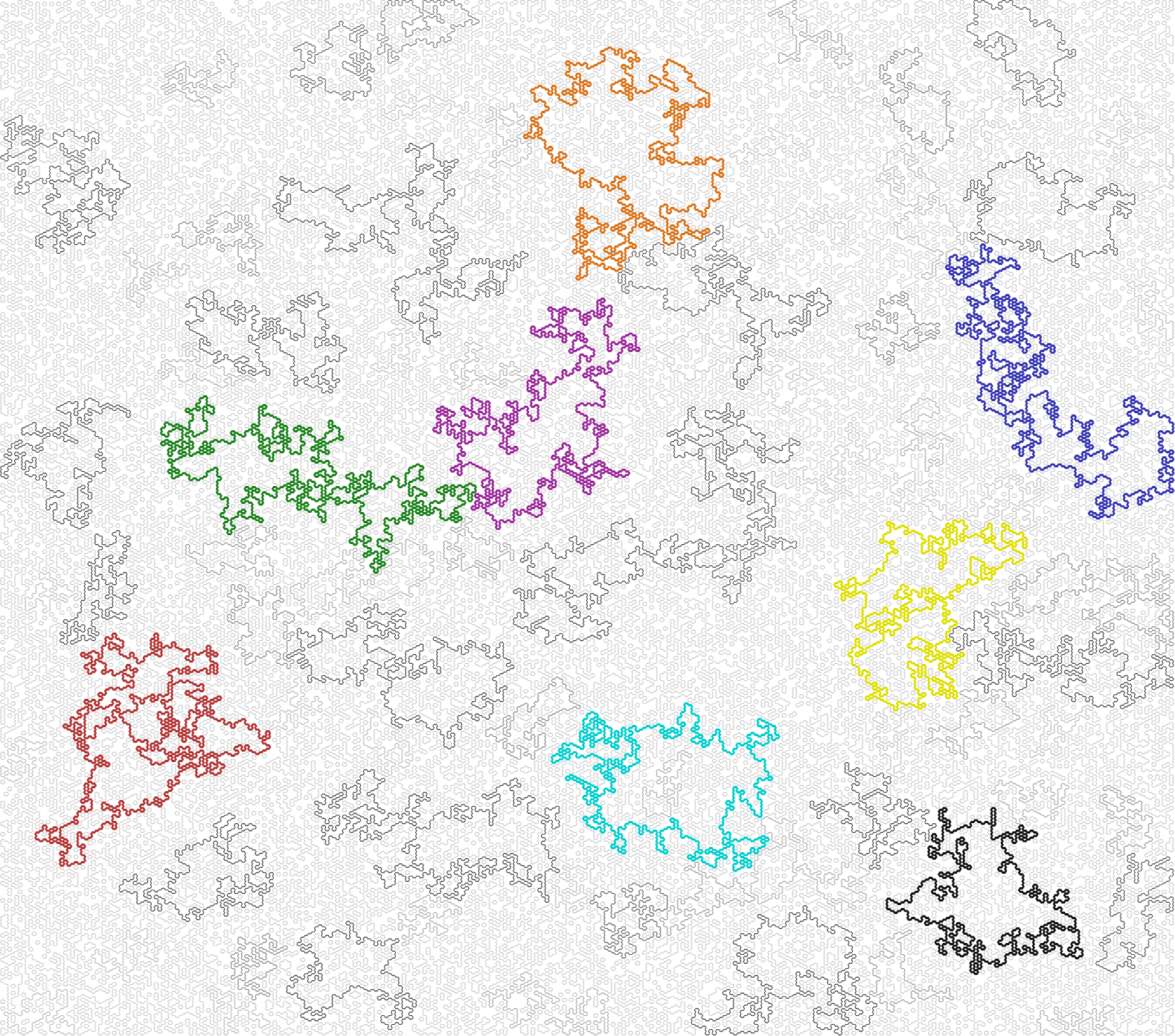}
		\caption{$n=2.5$ and $x=1$.}
		\label{fig:loop-sample-n2,5-x1}
	\end{subfigure}
	\caption{Samples of random loop configurations on a $340\times300$ rectangular-shaped domain.}
	\label{fig:loop-samples-large-scale}
\end{figure}

It is predicted~\cite{Nie82,KagNie04,Smi06} that the loop $O(n)$ model exhibits critical behavior when $n\in [0,2]$; see Figure~\reffig{fig:loop-samples}. In this regime, the model should have a critical value $x_c(n)$ with the formula
\begin{equation}\label{eq:x c def}
  x_c(n) := \frac{1}{\sqrt{2 + \sqrt{2 - n}}}.
\end{equation}
The prediction is that for $x<x_c$ the model is \emph{sub-critical} in the sense that the probability that a loop passing through a given point has length longer than $t$ decays \emph{exponentially} in~$t$. When $x\ge x_c$, the model should be \emph{critical}, with the same probability decaying only as a \emph{power-law} in $t$ and with the model exhibiting a conformally-invariant scaling limit. Furthermore, there should be two critical regimes: when $x = x_c$ and $x>x_c$, each characterized by its own conformally-invariant scaling limit (the same one for all $x>x_c$ and a different one for $x = x_c$). Kager and Nienhuis~\cite[Section 5.6]{KagNie04} predict that in both cases, the loops should scale in a suitable limit to random \emph{Schramm L\"owner evolution} (SLE) curves, introduced by Schramm~\cite{Sch00}, with parameter~$\kappa$ satisfying
\begin{equation}\label{eq:kappa and n relation}
  n = -2\cos\left(\frac{4\pi}{\kappa}\right),
\end{equation}
where, however, we take the solution of the above equation to
satisfy $\kappa\in [\frac{8}{3},4]$ when $x=x_c$ and
$\kappa\in[4,8]$ when $x>x_c$. When the parameter $n$ satisfies
$n>2$, it is predicted that the model is always sub-critical in the
sense of exponential decay of loop lengths described above. These predictions have been mathematically validated only in very special cases. See Figure~\reffig{fig:loop-samples} and Figure~\reffig{fig:loop-samples-large-scale} for samples from the loop $O(n)$ model. See also the two bottom figures on the cover page which show samples of the model with $n=0.5$ and $x=0.6$.

The physics literature considers the loop $O(n)$ model also with negative $n$, where the model is still defined by~\eqref{eq:loop O n model def} but is now a signed measure. Critical behavior is then predicted for $n\in [-2,2]$, with the same critical value~\eqref{eq:x c def} for $x$; see~\cite{Nie82}. Presumably formula~\eqref{eq:kappa and n relation} continues to describe the parameter $\kappa$ of the scaling limit of the model throughout this range. However, the precise meaning of these predictions for negative $n$ is less clear.

\begin{figure}[t]
	\begin{center}
		\includegraphics[scale=1.5]{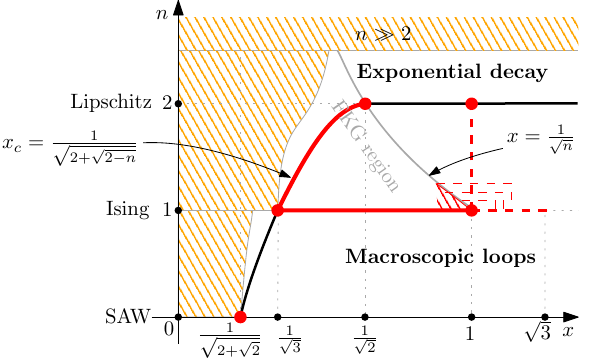}
	\end{center}
	\caption{The predicted phase diagram for the loop $O(n)$ model. The critical line $x_c$ separating the regime of exponential decay from the regime of macroscopic loops is plotted. The region where a dichotomy between the two behaviors is proved is denoted FKG regime. Orange lines illustrate regions where exponential decay is proved. Red dots or lines mark regions where macroscopic loops are proven to occur. Dotted red lines denote regions where exponential decay is ruled out. Picture adapted from Glazman--Manolescu~\cite{glazman2018exponential}.}
	\label{fig:phase diagram}
\end{figure}

\medbreak
We list the main rigorous results on the loop $O(n)$ model.

In the critical percolation case, $n= x = 1$, Smirnov~\cite{Smi01} proved that crossing probabilities have a conformally-invariant
scaling limit (given by Cardy's formula) and sketched a proof~\cite{Smi01,Smi06} for convergence of the exploration path to SLE(6), following an argument of Schramm~\cite{Sch00}. Camia and Newman~\cite{CamNew07} proved this latter convergence and further showed~\cite{camia2004continuum,camia2005full,CamNew06} that the full scaling limit is CLE(6), a member of the family of \emph{conformal loop ensembles} introduced by Sheffield~\cite{sheffield2009exploration}.

In the Ising model case, $n=1$, it is known that $x = x_c(1) = \frac{1}{\sqrt{3}}$ is critical~\cite{Ons44} with its interface scaling to SLE(3)~\cite{Smi10,CheSmi12,HonKyt13,CheDumHon14,izyurov2015smirnov} and its loops scaling to CLE(3)~\cite{benoist2016scaling}.

In the self-avoiding walk case, $n=0$, it was proved by Duminil-Copin and Smirnov~\cite{DumSmi12} that $x=x_c(0) = 1/\sqrt{2 + \sqrt{2}}$ is critical (it is the inverse of the connective constant of the hexagonal lattice; see Section~\ref{sec:saw}), though conformal invariance and convergence
to SLE have not been established. Furthermore, it was shown that for $x > x_c(0)$ the self-avoiding walk is space-filling~\cite{DumKozYad11}.

For large values of $n$, it has been shown by Duminil-Copin--Peled--Samotij--Spinka~\cite{DCPSS14} that there is exponential decay of loop lengths for all values of $x$ (see Section~\ref{sec:large-n} for details).

For $n \in [1,2]$, it has been shown by Duminil-Copin--Glazman--Peled--Spinka~\cite{macroscopicloops2017} that the model exhibits macroscopic loops at the critical point $x=x_c(n)$. A main tool in the proof is the observation that the Ising-type \emph{spin representation} of the loop $O(n)$ model, in which the loops form the interfaces between $\pm1$ spins on the triangular lattice, satisfies the FKG lattice condition (i.e., has strong positive association) when $n\ge 1$ and $nx^2\le 1$. Based on this and ideas from~\cite{DumSidTas13} the authors deduce a \emph{dichotomy} theorem when $n$ and $x$ are in this range: Either the length of the loop passing through a given vertex has exponential tail decay, or the model satisfies Russo--Seymour--Welsh (RSW) type estimates, i.e., for some $c\in(0,1)$ and any given annulus whose outer radius is twice its inner radius, the probability to see a loop which winds around the annulus is between $c$ and $1-c$. In this range of parameters, using a technique of~\cite{georgii2000percolation}, it is further shown that the loop $O(n)$ model has a unique Gibbs measure.

The high-temperature (ferromagnetic) Ising case, when $n=1$ and $\frac{1}{\sqrt{3}}<x<1$, can be shown to exhibit RSW type estimates using the techniques of \cite{tassion2016crossing}. Another possibility is to rely on the aforementioned dichotomy: Exponential tail decay can be ruled out by noting that the Ising model has a unique Gibbs measure in this range~\cite[Theorem 3.25]{friedli2016statistical}, which then cannot have an infinite connected component of $+$'s, nor an infinite connected component of $-$'s (by an argument of Zhang which rules out coexistence of the infinite components; see, e.g.,~\cite[Theorem 14.3]{haggstrom2006uniqueness}), so that every vertex is surrounded by infinitely many loops (domain walls).
Remarkably, the question of convergence of the loops to CLE(6) for the high-temperature Ising model remains open.

We briefly mention some very recent developments:
Beffara--Gayet~\cite{beffara2017percolation} prove that RSW type estimates hold for the (very-)high-temperature antiferromagnetic Ising model, when $n=1$ and $1<x<1+\eps$ for some small $\eps>0$. Glazman--Manolescu~\cite{glazman2018uniform} prove RSW type estimates for $n=2$ and $x=1$, where the loops can be viewed as the level lines of a uniform Lipschitz function (see Section~\ref{sec:loop-O-n-exact-representations} below). Crawford--Glazman--Harel--Peled~\cite{crawford2019} rule out the possibility that the length of loops has exponential tail decay for $1 \le n \le 1+\eps$ and $1 - \eps \le x \le 1+\eps$ for some small $\eps>0$. This implies RSW type estimates for $1\le n\le 1+\eps$ and $1-\eps\le x\le \frac{1}{\sqrt{n}}$ by the aforementioned dichotomy. They further show that when $1\le n\le 2$ and $x=1$ the model on a toroidal domain has a \emph{non-contractible} loop with non-negligible probability. Lastly, it is proved there that when $n=1$ and $1<x\le\sqrt{3}$ (antiferromagnetic Ising) the model has loops of large diameter (comparable to that of the domain) with non-negligible probability.
On the other side, Taggi~\cite{taggi2018shifted} established exponential decay of loop lengths when $n>0$ and $x\le (\sqrt{2+\sqrt{2}})^{-1} + \eps(n)$, with $\eps(n)>0$ some function of $n$. Glazman--Manolescu~\cite{glazman2018exponential} further showed exponential decay for any $n > 1$ and $x<\frac{1}{\sqrt{3}}+\eps(n)$, with $\eps(n)>0$ another function of~$n$.

Many interesting questions remain open for the loop $O(n)$ model, with some of the more notable ones being: proving conformal invariance at any point except $n=1, x=\frac{1}{\sqrt{3}}$ and $n=x=1$, and showing the existence of large loops in the remaining parts of the phase diagram: for $0<n<1$ and any $x$, or for $n=1$ and any $x\in(\sqrt{3},\infty]$ (it is unknown even for the dimer model case, $x=\infty$), or for $1<n\le 2$ and $x>x_c(n)$ (apart from the case $n=2, x=1$ and from the neighborhood of $n=x=1$ mentioned above) .

\subsection{Equivalent models}
The equivalent models discussed below do not reside on the hexagonal lattice, but rather on its dual, the triangular lattice $\T$, which is obtained by placing a vertex at the center of every face (hexagon) of $\HH$, so that each edge $e$ of $\HH$ corresponds to the unique edge $e^*$ of $\T$ which intersects $e$. Since vertices of $\T$ are identified
with faces of $\HH$, they will be called \emph{hexagons} instead of
vertices. We also say that a vertex or an edge of $\HH$
\emph{borders} a hexagon if it borders the corresponding face of
$\HH$.

\subsubsection{The hard-hexagon model}
\label{sec:hard-hex}

As noted already in the paper~\cite{DomMukNie81} where the loop $O(n)$ model was introduced,
taking the limit $n\to\infty$ and $nx^6\to\lambda$ leads formally to
the hard-hexagon model. As non-trivial loops (loops having length longer than $6$) become
less and less likely in this limit, hard-hexagon configurations
consist solely of trivial loops, with each such loop contributing a
factor of $\lambda$ to the weight. Thus, the hard-hexagon model is
the hard-core lattice gas model on the triangular lattice $\T$ with
fugacity $\lambda$. For this model, Baxter \cite{Bax80} (see also
\cite[Chapter 14]{Bax89}) computed the critical fugacity
\begin{equation*}
\lambda_c = \left(2\cos\left(\frac{\pi}{5}\right)\right)^5 = \frac{1}{2}\left(11 + 5\sqrt{5}\right) \approx 11.09017,
\end{equation*}
and showed that as $\lambda$ increases beyond the threshold
$\lambda_c$, the model undergoes a fluid-solid phase transition, from
a homogeneous phase in which the sublattice occupation frequencies
are equal, to a phase in which one of the three sublattices is
favored. Additional information is obtained on the critical behavior,
including the fact that the mean density of hexagons is equal for
each of the three sublattices \cite[Equation (13)]{Bax80} and the
fact that the transition is of second order \cite[Equation
(9)]{Bax80}. Baxter's arguments use certain assumptions on the model
which appear not to have been mathematically justified. Still, this
exact solution may suggest that the loop $O(n)$ model with large $n$
will also have a unique transition point $x_c(n)$, that $nx_c(n)^6$
will converge to $\lambda_c$ as $n$ tends to infinity and that the
transition in $x$ is of second order, with the model having a unique
Gibbs state when $x = x_c(n)$.

\subsubsection{Exact representations as spin models with local interactions}
\label{sec:loop-O-n-exact-representations}

As explained in the previous section, the loop $O(n)$ is an \emph{approximation} of the spin $O(n)$ model, a spin model on $\HH$ with local interactions. Here we develop \emph{exact} representations of the loop $O(n)$ model as spin models on $\T$ with local interactions (see also~\cite{cardy2008conformal}).

The spin space here will always be a discrete set $S$ (finite or countably infinite) and we shall restrict ourselves to the set $\Phi$ of spin configurations $\varphi \in S^\T$ satisfying the condition that $|\{ \varphi(y),\varphi(z),\varphi(w)\}| \le 2$ for any three mutually adjacent hexagons $y,z,w \in \T$.
Define the `domain walls' of a configuration $\varphi \in \Phi$ by
\begin{equation*}
\omega_\varphi := \big\{e \in E(\HH) \,:\, \text{the edge $e$ borders hexagons $y,z\in\T$ satisfying $\varphi(y)\neq \varphi(z)$}\big\},
\end{equation*}
and observe that $\omega_\varphi$ is a loop configuration.
For a domain $H \subset \HH$ and a fixed $s_0 \in S$, let $\Phi(H)$ be the set of $\varphi\in \Phi$ satisfying the boundary condition $\varphi(z)=s_0$ for any hexagon $z \in \T$ which is not entirely contained in $H$. Note that $\omega_\varphi\in\LC(H)$ for $\varphi\in\Phi(H)$.

We now define a spin model in a similar manner as in Section~\ref{sec:spin-O-n}, with one important difference: as we are now working on the triangular lattice, rather than the square lattice, it is natural to consider triangular interactions, rather than pairwise interactions. Precisely, given a (non-zero) symmetric interaction $h \colon S^3 \to [0,\infty)$, i.e., $h=h \circ \tau$ for any permutation $\tau \in S_3$, we consider the probability distribution on $\Phi(H)$ in which the probability of a configuration $\varphi \in \Phi(H)$ is proportional to
\begin{equation}\label{eq:triangle-interaction}
\prod_{\{y,z,w\}} h(\varphi(y),\varphi(z),\varphi(w)) ,
\end{equation}
where the product is over triples $\{y,z,w\}$ of mutually adjacent hexagons $y,z,w \in \T$, at least one of which has an edge in $H$. We note that in order for this distribution to be well-defined when $S$ is infinite, one must impose an implicit condition on $h$ to ensure that the sum of the above weights is finite. Note also that this distribution is entirely defined by the collection of numbers $(h_a)_{a \in S}$ and $(h_{a,b})_{a,b \in S, a \neq b}$, where $h_a := h(a,a,a)$ and $h_{a,b} := h(a,b,b)$ ($h_{a,b}$ need not equal $h_{b,a}$ in general).

Any such choice of spin space $S$ and interaction $h$, gives rise via the map $\varphi \mapsto \omega_\varphi$ to a probability distribution on $\LC(H)$. The goal is then to choose $S$ and $h$ in such a manner that this distribution coincides with the loop $O(n)$ measure $\Pr_{H,n,x}$.
As we now show, there is in fact a general recipe for constructing such examples.

Let $G$ be a simple graph on vertex set $S$. We focus on the case that $h$ imposes the hard-core constraint that $h_{a,b}=0$ unless $\{a,b\}$ is an edge of $G$. In order words, the corresponding distribution on $\Phi(H)$ is supported on $\text{Lip}(G)$, the set of configurations $\varphi \in S^\T$ satisfying the Lipschitz condition: if $y,z \in \T$ are adjacent hexagons then either $\varphi(y)=\varphi(z)$ or $\varphi(y)$ is adjacent to $\varphi(z)$ in $G$. We note that, in general, neither $\Phi$ nor $\text{Lip}(G)$ is contained in the other. However, in the case that $G$ contains no triangles, we have $\text{Lip}(G) \subset \Phi$.

For simplicity, we now restrict ourselves to the case that $S$ is finite.
Let $A$ be the adjacency matrix of the graph $G$, i.e., $A$ is a real symmetric matrix, indexed by the set $S$ and defined by $A_{a,b} := \one_{\{\{a,b\} \in E(G)\}}$ for $a,b \in S$. Let $\psi$ be the Perron--Frobenius eigenvector corresponding to the largest eigenvalue $\lambda$ of $A$, i.e., the components of $\psi$ are non-negative and $A\psi=\lambda\psi$. We now choose the interaction to be $h_a := 1$ for all $a\in S$ and $h_{a,b} := x (\psi_a/\psi_b)^{1/6}$ for all adjacent $a,b \in S$.

Let us now show that if $\varphi$ is a random spin configuration sampled according to the distribution corresponding to the above choice of $h$, then $\omega_\varphi$ is distributed according to $\Pr_{H,\lambda,x}$.
To this end, we must show that, for any fixed $\omega \in \LC(H)$, the sum of weights in~\eqref{eq:triangle-interaction} over configurations in $\varphi \in \Phi(H)$ having $\omega_\varphi = \omega$ is proportional to $x^{o(\omega)}n^{L(\omega)}$. To see this, observe that $\omega$ may have nested loops and that by considering these loops one-by-one, from the innermost to the outermost, it suffices to show that the contribution of any single loop $\ell$ is $x^{|\ell|}\lambda$.
More precisely, for any $a \in S$, the sum of weights in~\eqref{eq:triangle-interaction} over configurations $\varphi$ which equal $a$ on the exterior side of $\ell$ and satisfy $\omega_\varphi=\ell$ is $x^{|\ell|}\lambda$. Indeed, this sum is precisely $\sum_{b \in S} h_{b,a}^m h_{a,b}^{m'}$, where $m$ and $m'$ are the number of vertices of $\ell$ which are incident to an edge in the exterior and interior sides of~$\ell$, respectively. Geometrically, if one traverses $\ell$ in counterclockwise direction, then $m$ and $m'$ are the number of left-hand and right-hand turns, respectively. In particular, it always holds that $m=m'+6$. Thus,
\[ \sum_{b \in S} h_{b,a}^m h_{a,b}^{m'} = x^{m+m'} \sum_{b:\{a,b\} \in E(G)} (\psi_b/\psi_a)^{(m-m')/6} = x^{|\ell|} \frac{(A \psi)_a}{\psi_a} = x^{|\ell|} \lambda .\]

We have thus shown that if there exists a finite graph whose adjacency matrix has maximum eigenvalue $n$, then one may find an exact representation of the loop $O(n)$ model with any value of the edge-weight $x$ as a spin model with local interactions (and finite spin space). Not all values of $n>0$ are obtainable as such.
The set of possible $n$ in $(0,2)$ is known; They are the eigenvalues of the ADE diagrams and form an infinite set in $[1,2)$ having 2 as its sole accumulation point.

We remark that the above construction can sometimes be extended to the case when $G$ is an infinite, locally finite graph (i.e., all vertices have finite degrees). In this case, the Perron--Frobenius eigenvector $\psi$ is replaced by a non-zero element $\psi \in \R^S$ such that $\psi \ge 0$ and $\lambda \psi_a = \sum_{b:\{a,b\} \in E(G)} \psi_b$ for some $\lambda>0$ and all $a\in S$. If such a $\psi$ exists, then the arguments above continue to hold without change.

\medbreak
\noindent{\bf Lipschitz functions.}
\label{sec:loop-o-n-height}
When $n$ is a positive integer, the loop $O(n)$ model admits a height function representation~\cite{DomMukNie81}. Let $G=T_n$ be the $n$-regular tree (so that $T_1 = \{+,-\}$ and $T_2=\Z$) rooted at an arbitrary vertex $\rho$. Here, $\text{Lip}(T_n)$ is the set of 1-Lipschitz functions from $\T$ to $T_n$ (where the metrics are the graph distances), and moreover, $\text{Lip}(T_n) \subset \Phi$ as $T_n$ does not contain triangles. In this case, one may regard $\omega_\varphi$ as the `level lines' of the height function $\varphi \in \text{Lip}(T_n)$.
Since $\lambda=n$ is an eigenvalue of $T_n$ (in the sense discussed above; the eigenvector $\psi$ is the constant function), we see that if one samples a random function $\varphi\in \text{Lip}(T_n) \cap \Phi(H)$ with probability proportional to $x^{|\omega_\varphi|}$, then $\omega_\varphi$ is distributed according to $\Pr_{H,n,x}$. In particular, the height function representation of the loop $O(1)$ model is an Ising model (which may be either ferromagnetic or antiferromagnetic according to whether $x<1$ or $x>1$) and the height function representation of the loop $O(2)$ model is a restricted Solid-On-Solid model (an integer-valued Lipschitz function). Andrews--Baxter--Forrester~\cite{andrews1984eight} studied a related type of restricted Solid-On-Solid models.

\medbreak
\noindent{\bf The dilute Potts model.}
Let $q \ge 1$ be an integer and set $S := \{0,1,\dots,q\}$.
Let $G$ be star graph on $S$ in which $0$ is the center, i.e., the edges of $G$ are $\{0,i\}$ for $1 \le i \le q$. Here, $\text{Lip}(G) \subset \Psi$ and the elements of this set can be thought of as configurations in a \emph{dilute Potts model}: the value $0$ represents a vacancy and a positive value $i$ represent a particle/spin of type $i$. In this case, the Perron--Frobenius eigenvector is given by $\psi(0) := \sqrt{q}$ and $\psi(i) := 1$ for $1 \le i \le q$, and its corresponding eigenvalue is $\lambda = \sqrt{q}$. Thus, this model gives a representation of the loop $O(n)$ model for any $n$ which is the square root of an integer.

Nienhuis~\cite{nienhuis1991locus} proposed a slightly different version of the dilute Potts model, similar to the above representation. A configuration of this model in a domain of the triangular lattice is an
assignment of a pair $(s_z, t_z)$ to each vertex $z$ of the domain,
where $s_z\in\{1,\ldots, q\}$ represents a spin and $t_z\in\{0,1\}$
denotes an occupancy variable. The probability of configurations
involves a hard-core constraint that nearest-neighbor occupied sites
must have equal spins (reminiscent of the Edwards--Sokal coupling of
the Potts and random-cluster models) and single-site,
nearest-neighbor and triangle interaction terms involving the
occupancy variables as in~\eqref{eq:triangle-interaction}. With a certain choice of coupling constants,
the marginal of the model on the product variables $(s_z t_z)_z$ has the same distribution as $\varphi$ (with the above choice of $G$), and thus, the marginal on the occupancy variables is equivalent to the loop $O(n)$ model (with $n=\sqrt{q}$). Nienhuis
predicts this choice of parameters to be part of the critical
surface of the dilute Potts model. This prediction is partially confirmed in~\cite{macroscopicloops2017} for the loop $O(n)$ model with parameters $n\ge 1$ and $nx^2\le 1$.

\subsection{Self-avoiding walk and the connective constant}
\label{sec:saw}

The loop $O(n)$ model as defined in Section~\ref{sec:loop-model-def} is said to have \emph{vacant boundary conditions}. In this case, the probability of any non-empty loop configuration tends to zero as $n$ tends to zero. Thus, under vacant boundary conditions, the $n=0$ model is trivial.
However, as can be done for the spin $O(n)$ model, here too one may impose different \emph{boundary conditions} on the model, where the states of certain edges are pre-specified. Taking boundary conditions for which precisely two edges $e_1$ and $e_2$ on the boundary of the domain $H$ are present, one forces a self-avoiding path between these two edges within the domain (in addition to possible loops). Under such boundary conditions, in the limit as $n \to 0$, one obtains a random self-avoiding walk. The probability of such a given self-avoiding walk $\gamma$ is proportional to $x^{\text{length}(\gamma)}$.
The partition function, $Z_{H,x,e_1,e_2}^{\text{saw}}$, is given by
\[ Z_{H,x,e_1,e_2}^{\text{saw}} := \sum_{\substack{\gamma: e_1 \to e_2\\\gamma \subset H}} x^{\text{length}(\gamma)} = \sum_{k=0}^\infty s_{H,e_1,e_2,k} x^k ,\]
where $s_{H,k,e_1,e_2}$ is the number of self-avoiding walks of length $k$ from $e_1$ to $e_2$ in $H$.

We consider the related partition function of all self-avoiding walks starting at a fixed vertex $v$, given by
\[ Z_x^{\text{saw}} := \sum_{\gamma:\,\gamma_0=v} x^{\text{length}(\gamma)} = \sum_{k=0}^\infty s_k x^k ,\]
where $s_k$ is the number of self-avoiding walks of length $k$ starting at $v$. The series defining $Z_x^{\text{saw}}$ has a radius of convergence $x_c \in [0,\infty]$ so that $Z_x^{\text{saw}}<\infty$ when $0<x<x_c$ and $Z_x^{\text{saw}}=\infty$ when $x>x_c$. This is the critical point of the model. The critical value $x_c$ is directly related to the exponential rate of growth of $s_k$.

An important and simple observation is that $s_k$ is sub-multiplicative. That is,
\[ s_{k+m} \le s_k s_m .\]
It follows that the limit
\[ \mu := \lim_{k \to \infty} s_k^{1/k} \]
exists and is finite. The number $\mu$, called the \emph{connective constant} of the hexagonal lattice, clearly relates to the critical value via $\mu = 1/x_c$.

\medbreak
\noindent {\bf Exercise.}
Show that $\mu$ is well-defined and that $\mu = \inf_k s_k^{1/k}$.
\medbreak
\noindent {\bf Exercise.}
Show that $2^{k/2} \le s_k \le 3 \cdot 2^{k-1}$ and deduce that $\sqrt{2} \le \mu \le 2$.
\medbreak

Recently, Duminil-Copin and Smirnov~\cite{DumSmi12} showed the following remarkable result.

\begin{theorem}
The connective constant of the hexagonal lattice is
\[ \mu = \sqrt{2 + \sqrt{2}} .\]
\end{theorem}

We do not give the proof in these notes and refer the interested reader to~\cite{DumSmi12}.

\subsection{Large $n$}
\label{sec:large-n}

It is believed that the loop $O(n)$ model, although only an
approximation of the spin $O(n)$ model, resides in the same
universality class as the spin $O(n)$ model. Thus, as in the case of
the spin $O(n)$ model, it has been conjectured that the loop $O(n)$
model exhibits exponential decay of correlations when $n>2$. Duminil-Copin, Peled, Samotij and
Spinka~\cite{DCPSS14} established this for large $n$, showing
that long loops are exponentially unlikely to occur, uniformly in
the edge weight $x$. This result is the content of the first
theorem below.

We begin with some definitions (see Figure~\reffig{fig:loop-config-ground-state} for their illustration).
Recall that the triangular lattice $\T$ is the dual of the hexagonal lattice.
Fix a proper 3-coloring of $\mathbb{T}$ (there is a unique such coloring up to permutations of the colors), and let $\mathbb{T}^0$, $\mathbb{T}^1$ and $\mathbb{T}^2$ denote the color classes of this coloring.
The $0$-phase ground state $\ground^0$ is defined to be the (fully-packed) loop configuration consisting of trivial loops (loops of length 6) around each hexagon in $\mathbb{T}^0$.
A domain $H \subset \HH$ is said to be \emph{of type 0} if no edge on its boundary belongs to $\ground^0$, or equivalently, if every edge bordering a hexagon in $\mathbb{T}^0$ has either both or neither of its endpoints in $V(H)$.
Finally, we say that a loop surrounds a vertex $u$ of $\HH$ if any infinite simple path in $\HH$ starting at $u$ intersects a vertex of this loop. In particular, if a loop passes through a vertex then it surrounds it as well.

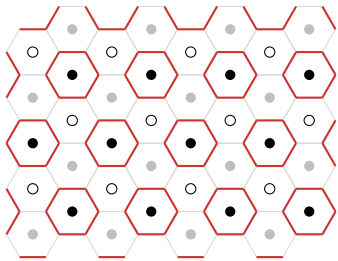
\begin{figure}
	\centering
	
	\begin{tikzpicture}[scale=0.35, every node/.style={scale=0.35}]
	\groundState[7][4][1];
	\hexagonEdges[edge-on][0][0][0/0/-1,0/1/-2,0/3/-1,0/4/-2];
	\hexagonEdges[edge-on][0][0][0/0/3,2/-1/3,4/-2/3,6/-3/3];
	\hexagonEdges[edge-on][7][-2][0/0/1,0/1/2];
	\hexagonEdges[edge-on][0][5][0/0/3,1/-1/-1,1/-1/1,2/-1/3,3/-2/-1,3/-2/1,4/-2/3,5/-3/-1,5/-3/1,6/-3/3,7/-4/-1,7/-4/1];
	\end{tikzpicture}
	
	\caption{A proper $3$-coloring of the triangular lattice $\mathbb{T}$ (the dual of the hexagonal lattice $\HH$), inducing a partition of $\mathbb{T}$ into three color classes $\mathbb{T}^0$, $\mathbb{T}^1$, and $\mathbb{T}^2$. The $0$-phase ground state $\ground^0$ is the (fully-packed) loop configuration consisting of trivial loops around each hexagon in~$\mathbb{T}^0$.}
	\label{fig:loop-config-ground-state}
\end{figure}

\begin{theorem}\label{thm:no-large-loops}
  There exist $n_0,c > 0$ such that for any $n \ge n_0$, any $x\in(0,\infty]$ and any domain $H$ of type 0 the following holds.
  Suppose $\omega$ is sampled from the loop $O(n)$ model in domain $H$ with edge weight $x$. Then, for any vertex $u \in V(H)$ and any integer $k > 6$,
  \[
  \Pr(\text{there exists a loop of length $k$ surrounding $u$}) \le n^{-ck}.
  \]
\end{theorem}

\begin{figure}[t!]
   \centering
   \begin{subfigure}[t]{.5\textwidth}
       \includegraphics[scale=0.46]{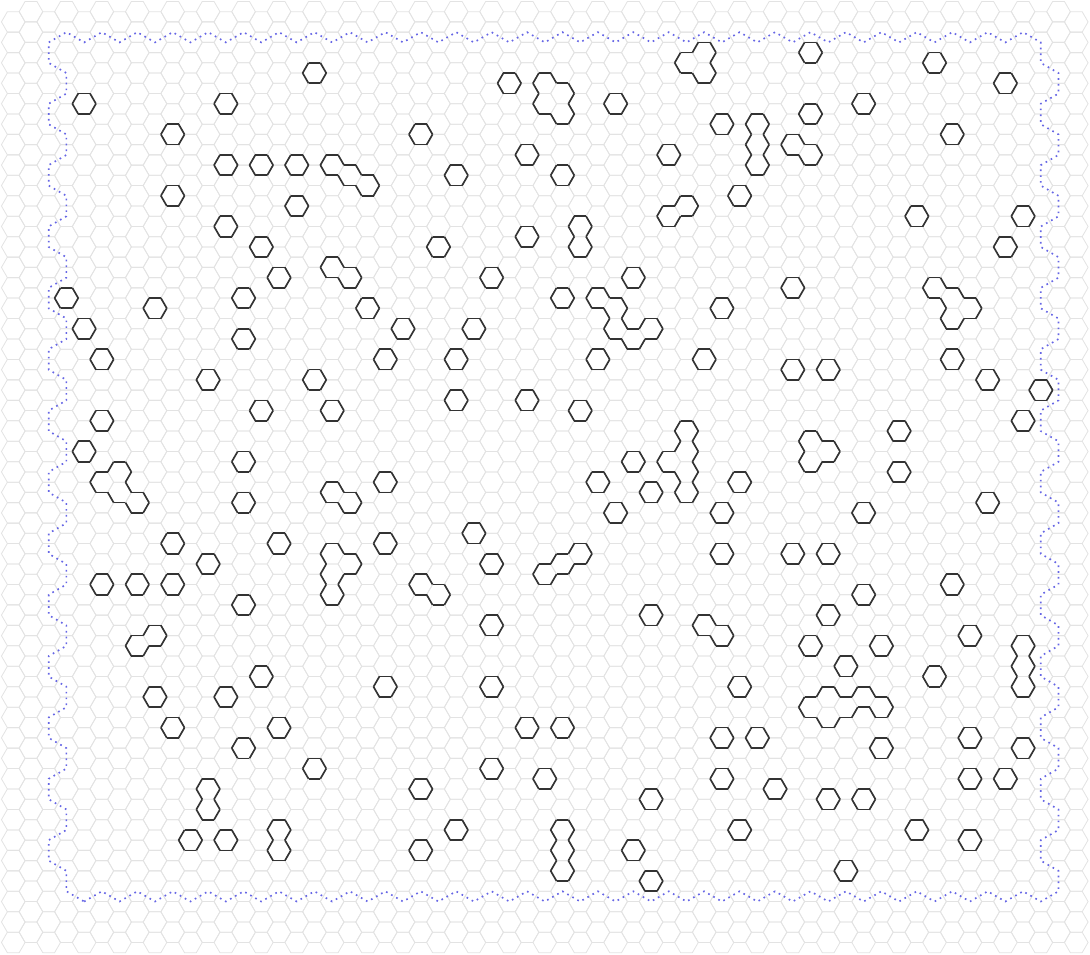}
       \caption{$n=8$ and $x=0.5$. When $x$ is small, the limiting measure is unique for domains with vacant boundary conditions, and the model is in a dilute, disordered phase.}
       \label{fig:loop-sample-n=8,x=0.5}
   \end{subfigure}%
   \begin{subfigure}{20pt}
       \quad
   \end{subfigure}%
   \begin{subfigure}[t]{.5\textwidth}
       \includegraphics[scale=0.46]{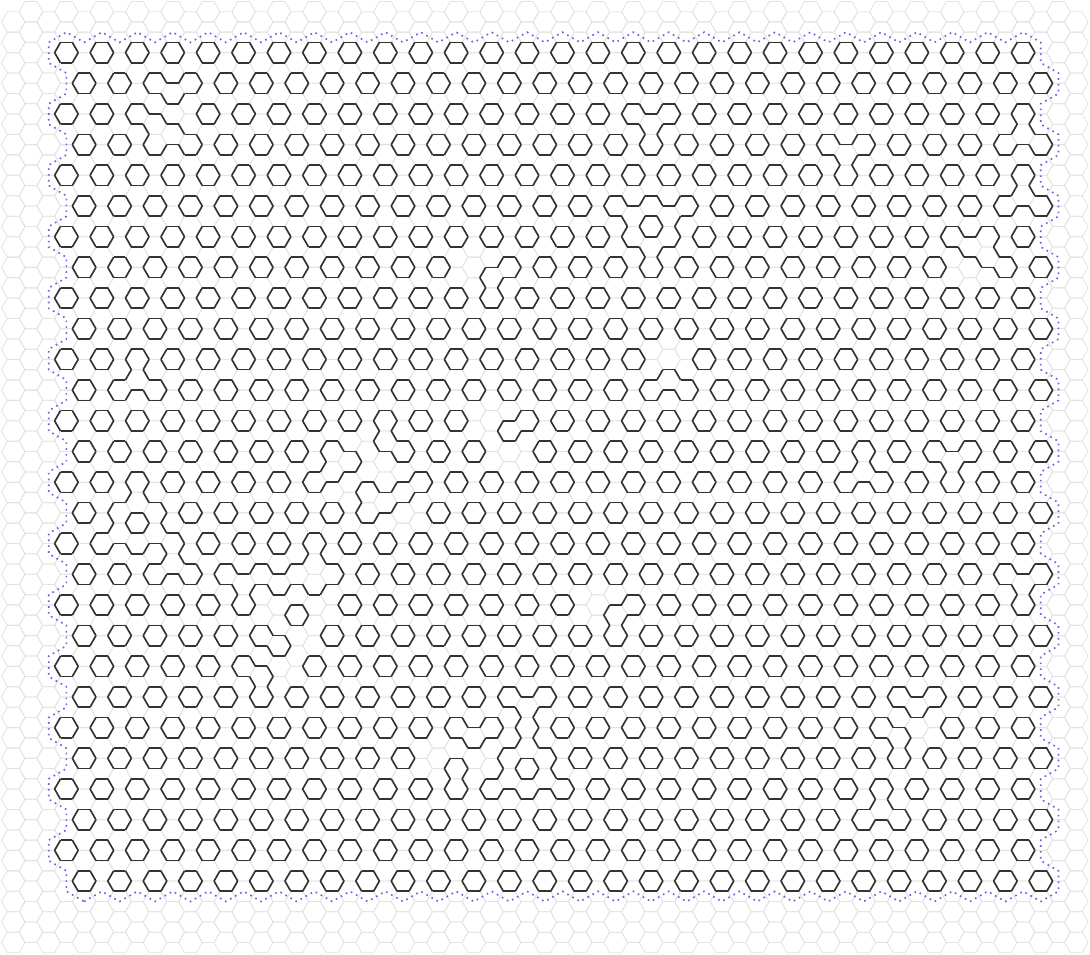}
       \caption{$n=8$ and $x=2$. When $n$ is large and $x$ is not small, the model is in an ordered phase where typical configurations are small perturbations of the ground state.}
       \label{fig:loop-sample-n=8,x=2}
   \end{subfigure}
   \caption{Two samples of random loop configurations with large $n$. Configurations are on a $60\times45$ domain of type $0$ and are sampled via Glauber dynamics for 100 million iterations started from the empty configuration.}
   \label{fig:loop-samples-large-n}
\end{figure}

The reasons behind this exponential decay are quite different when $x$ is small or large. While there is no transition to slow decay of loop lengths as $x$ increases, there is a different kind of transition in terms of the structure of the random loop configuration and, in particular, in how the loops pack in the domain.
When $x$ is small, the model is dilute and disordered, whereas, when $x$ is large, the model is dense and ordered (a small perturbation of the $0$-phase ground state $\ground^0$); these behaviors are depicted in Figure~\ref{fig:loop-samples-large-n}.
We remark that it is this latter behavior that makes the assumption that $k>6$ necessary in the above theorem. The next theorem makes these statements precise.

Given a loop configuration $\omega$ and two vertices $u$ and $v$ in $\HH$, we say that $u$ and $v$ are \emph{loop-connected} if there exists a path between $u$ and $v$ consisting only of vertices which belong to loops in $\omega$, and we say that $u$ and $v$ are \emph{ground-connected} if there exists a path between $u$ and $v$ consisting only of vertices which belong to loops in $\omega \cap \ground^0$.

\begin{theorem}\label{thm:large-n-transition}
    There exist $C,c>0$ such that for any $n>0$, any $x\in(0,\infty]$ and any domain $H$ of type 0 the following holds. Suppose $\omega$ is sampled from the loop $O(n)$ model in domain $H$ with edge weight $x$.
    Then, for any vertex $u \in V(H)$, on the one hand,
    \[ \Pr(\text{$u$ is loop-connected to a vertex at distance $k$ from $u$}) \le (C (n+1) x^6)^{ck} , \qquad k \ge 1, \]
    and, on the other hand,
    \[ \Pr(\text{$u$ is ground-connected to a vertex on the boundary of }H) \ge 1 - C(n \min\{x^6,1\})^{-c} .\]
\end{theorem}

Note that the first bound is non-trivial when both $x$ and $nx^6$ are sufficiently small, while the second bound is non-trivial when both $n$ and $nx^6$ are sufficiently large.
Thus, when $n$ is large, the theorem establishes a change in behavior as $nx^6$ transitions from small to large values.
In particular, when $nx^6$ is small, any fixed vertex is unlikely to be surrounded by a loop (of any size). On the other hand, when $nx^6$ is large, any fixed hexagon in $\mathbb{T}^0$ is very likely to be surrounded by a trivial loop.
The proof for small $x$ is very similar in nature to the high-temperature case of the spin $O(n)$ model, as described in Section~\ref{sec:high-temperature_expansion}, while the proof for large $x$ is more intricate.

\medbreak
We remark that several rigorous results on the behavior of general loop models on the $\Z^d$ lattice were obtained by Chayes, Pryadko and Shtengel~\cite{chayes2000intersecting}. These include theorems of a similar nature to our Theorem~\ref{thm:no-large-loops} and Theorem~\ref{thm:large-n-transition}. The proofs there rely on reflection positivity and are thus tied to the $\Z^d$ lattice structure and require as well that $n$ be integer (which is a built-in feature of the loop models studied in~\cite{chayes2000intersecting}). As we have not found a representation for the loop $O(n)$ model (even with integer $n$) which is reflection positive for large values of $n$ and $x$, our proofs proceed by different means.
\medbreak

In these notes, we give an extended overview of the proofs of Theorem~\ref{thm:no-large-loops} and Theorem~\ref{thm:large-n-transition}, omitting most of the technical details.
The techniques of the proofs are combinatorial in nature and rely on a general principle captured by the following simple lemma.
\begin{lemma}
    \label{lem:prob-inequality-tool}
    Let $p,q>0$ and let $E$ and $F$ be two events in a discrete probability space. If there exists a map $\sfT \colon E \to F$ such that $\Pr(\sfT(e))\ge p\cdot\Pr(e)$ for every $e\in E$, and $|\sfT^{-1}(f)|\le q$ for every $f\in F$, then
    \[ \Pr(E) \leq \frac{q}{p}\cdot\Pr(F) .\]
\end{lemma}
\begin{proof}
    We have
    \[ p\cdot \Pr(E)  \leq \sum_{e \in E} \Pr(\sfT(e))
    = \sum_{e \in E} \sum_{f \in F} \Pr(f) \mathbf{1}_{\{\sfT(e)=f\}}
    = \sum_{f \in F} |\sfT^{-1}(f)|\cdot\Pr(f) \leq q\cdot \Pr(F). \qedhere \]
\end{proof}

The results for small $x$ are obtained via a fairly standard, and short, Peierls argument, by applying the above lemma to a map which removes loops (see Lemma~\ref{lem:given-loops-are-unlikely} below). Thus, the primary focus here lies in the study of the loop $O(n)$ model for large
$x$. In this regime, the main idea is to identify the region having an atypical structure (which is called the breakup) and apply the above lemma to a suitably defined `repair map'.
This map takes a configuration $\omega$ sampled in a domain of type 0 and having a large breakup, and
returns a `repaired' configuration in which the breakup is significantly reduced (see Figure~\reffig{fig:proof-illustration}).
In order to use Lemma~\ref{lem:prob-inequality-tool}, it is important that
the number of preimages of a given loop configuration is exponentially smaller than the probability gain. This yields the main lemma, Lemma~\ref{lem:prob-outer-circuit}, from which the results for large $x$ are later deduced.

\medbreak
\noindent{\bf Basic definitions.}
A {\em circuit} is a simple closed path in $\T$ of length at least $3$. We may view a circuit $\gamma$ as a sequence of hexagons $(\gamma_0,\dots,\gamma_m)$ with $\gamma_0=\gamma_m$.
Define $\gamma^*$ to be the set of edges
$\{\gamma_i,\gamma_{i+1}\}^* \in \EH$ for $0\le i<m$.
We now state two standard geometric facts regarding circuits and domains, which may be seen as a discrete version of the Jordan curve theorem. Proofs of these facts can be found in~\cite[Appendix B]{DCPSS14}.

\begin{fact}\label{fact:gamma-int-ext}
If $\gamma$ is a circuit then the removal of $\gamma^*$ splits $\HH$ into exactly two connected components, one of which is infinite, denoted by $\Ext\gamma$, and one of which is finite, denoted by $\Int\gamma$. Moreover, each of these are induced subgraphs of $\HH$.
\end{fact}

\begin{fact}\label{fact:circuit-domain-bijection}
    Circuits are in one-to-one correspondence with domains via $\gamma \leftrightarrow \Int\gamma$.
\end{fact}

Hence, every domain $H$ may be written as $H=\Int\gamma$ for some circuit $\gamma$.
Note also that $H$ is of type $0$ if and only if $\gamma \subset \T \setminus \T^0$.
We denote the vertex sets and edge sets
of $\Int\gamma,\Ext\gamma$ by $\IntVert\gamma,\ExtVert\gamma$ and
$\IntEdge\gamma,\ExtEdge\gamma$, respectively. Note that $\{
\IntVert\gamma, \ExtVert\gamma \}$ is a partition of $\VH$ and
that $\{ \IntEdge\gamma, \ExtEdge\gamma, \gamma^* \}$ is a partition
of $\EH$. We also define $\IntHex\gamma$ to be the set of faces
of $\Int\gamma$, i.e., the set of hexagons $z\in\T$ having all their
six bordering vertices in $\IntVert\gamma$. Since $\Int\gamma$ is
induced, this is equivalent to having all six bordering edges in
$\IntEdge\gamma$.

\begin{definition}[$\clr$-flower, $\clr$-garden, $\clr$-cluster, vacant circuit; see Figure~\reffig{fig:garden}]
  Let $\clr\in\{0,1,2\}$ and let $\omega$ be a loop configuration. A hexagon $z\in\T^\clr$ is a {\em $\clr$-flower of} $\omega$ if it is surrounded by a trivial loop in $\omega$. A subset $E \subset \EH$ is a \emph{$\clr$-garden of $\omega$} if there exists a circuit $\sigma \subset \T \setminus \T^\clr$ such that $E=\IntEdge\sigma \cup \sigma^*$ and every $z\in \T^\clr\cap\partial \IntHex\sigma$ is a $\clr$-flower of $\omega$. In this case, we denote $\sigma(E) :=
  \sigma$. A garden of $\omega$ is a $\clr$-garden of $\omega$ for some $\clr\in\{0,1,2\}$. A subset $E \subset \EH$ is a {\em $\clr$-cluster} of $\omega$ if
it is a $\clr$-garden of $\omega$ and it is not contained in any
other garden of $\omega$. A cluster of $\omega$ is a $\clr$-cluster of $\omega$ for some $\clr\in\{0,1,2\}$. A circuit $\sigma$ is {\em vacant} in $\omega$ if $\omega\cap\sigma^*=\emptyset$.
\end{definition}

We stress the fact that a garden/cluster is a {\em subset of the edges of $\HH$}.
We remark that distinct clusters of $\omega$ are edge disjoint and that, moreover, distinct $\clr$-clusters (for some $\clr$) are slightly separated from one another.
Here and below, when $A$ is a subset of vertices of a graph $G$, we use $\partial A$ to denote the \emph{(vertex) boundary} of $A$, i.e.,
\[
\partial A := \big\{u\in A ~:~ \{u,v\}\in E(G)\text{ for some }v \not\in A \big\}.
\]


\begin{figure}
   \centering

   \begin{tikzpicture}[scale=0.35, every node/.style={scale=0.45}]
    \hexagonGridOneClass[12][6][2][1];
    \begin{scope}[yshift=1*0.866cm, xshift=1.5cm]
    \trivialLoop[3][-1];\trivialLoop[5][-2];\trivialLoop[7][-3];
    \trivialLoop[2][1];\trivialLoop[1][3];\trivialLoop[2][4];
    \trivialLoop[4][3];\trivialLoop[6][2];\trivialLoop[8][1];
    \trivialLoop[9][-1];\trivialLoop[8][-2];
    \trivialLoop[6][0];
    \doubleLoopUp[0][1];
    \trivialLoop[10][0];
    \doubleLoopRight[9][-4];

    \hexagonEdges[edge-on][3][1][0/1/-3,0/1/-2,0/1/-1,0/1/0,0/1/1,1/0/0,1/0/1,2/-1/-1,2/-1/-2,1/-1/2,1/-1/3,1/-1/4,1/-1/5,1/0/-2];
    \end{scope}
    \begin{scope}[yshift=1*0.866cm, xshift=1.5cm, xscale=1.5, yscale=0.866]
    \draw [domain-path]
    {(0,6)--(0,8)--(1,9)--(1,11)--(2,12)--(3,11)--(4,12)--(5,11)--(6,12)--(7,11)--(8,12)--(9,11)--(9,9)--(10,8)--(10,6)--(9,5)--(9,3)--(8,2)--(8,0)--(7,-1)--(6,0)--(5,-1)--(4,0)--(3,-1)--(2,0)--(2,2)--(1,3)--(1,5)--cycle };
    \end{scope}
   \end{tikzpicture}

   \caption{A garden. The dashed line denotes a vacant circuit $\sigma \subset \T \setminus \T^\clr$, where $\clr\in\{0,1,2\}$. The edges inside $\sigma$, along with the edges crossing $\sigma$, then comprise a $\clr$-garden of $\omega$, since every hexagon in $\T^\clr \cap \partial \IntHex\sigma$ is surrounded by a trivial loop.}
   \label{fig:garden}
\end{figure}
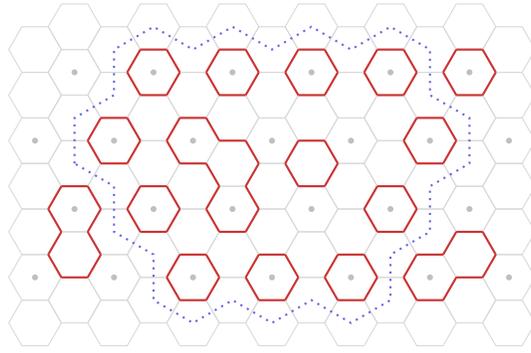


\medbreak
\noindent{\bf Statement of the main lemma.}
\label{sec:main lemma}
For a loop configuration $\omega$ and a vacant circuit $\gamma$ in
$\omega$, denote by $V(\omega,\gamma)$ the set of vertices $v \in
\IntVert\gamma$ such that the three edges of $\HH$ incident to $v$
are not all contained in the same cluster of $\omega \cap \IntEdge\gamma$. One may check that a vertex $v\in\IntVert\gamma$ satisfies $v\in V(\omega, \gamma)$ if and only if $v$ is incident to an edge which is not in any such cluster or each of its incident edges lies in a different such cluster.
The set $V(\omega, \gamma)$ specifies the deviation in $\omega$ from the $0$-phase ground state along the interior boundary of $\gamma$. The main lemma shows that having a large deviation is exponentially unlikely.

\begin{lemma}
    \label{lem:prob-outer-circuit}
    There exists $c>0$ such that for any $n>0$, any $x \in (0,\infty]$ and any circuit $\gamma\subset \T\setminus\T^0$ the following holds. Suppose $\omega$ is sampled from the loop $O(n)$ model in domain $\Int\gamma$ with edge weight $x$. Then, for any positive integer $k$,
    \[
    \Pr\big(\partial\IntVert\gamma\subset V(\omega,\gamma)\text{ and }|V(\omega,\gamma)|\ge k \big) \le (cn \cdot \min\{x^6,1\})^{-k/15} .
    \]
\end{lemma}

\medbreak
\noindent{\bf Definition of the repair map.}
\label{sec:repair-map}
Fix a circuit $\gamma\subset \T\setminus \T^0$ and set $H := \Int\gamma$.
 Consider a loop configuration $\omega$ such that $\gamma$ is vacant in $\omega$. The idea of the repair map is to modify $\omega$ in the interior of $\gamma$, keeping the configuration unchanged in the exterior of $\gamma$, as follows (see Figure~\reffig{fig:proof-illustration} for an illustration):
\begin{itemize}[noitemsep,topsep=0.5em]
\setlength\itemsep{0.25em}
\item
  Edges in $1$-clusters are shifted down ``into the $0$-phase''.
\item
  Edges in $2$-clusters are shifted up ``into the $0$-phase''.
\item
  Edges in $0$-clusters are left untouched.
\item
  The remaining edges which are not inside (the shifted) clusters, but are in the interior of $\gamma$ (these edges will be called {\em bad}), are overwritten to ``match'' the $0$-phase ground state, $\ground^0$.
\end{itemize}

In order to formalize this idea, we need a few definitions. A
\emph{shift} is a graph automorphism of $\T$ which maps every
hexagon to one of its neighbors. We henceforth fix a shift $\DIRup$
which maps $\T^0$ to $\T^1$ (and hence, maps $\T^1$ to $\T^2$ and
$\T^2$ to $\T^0$), and denote its inverse by $\DIRdown$. A shift
naturally induces mappings on the vertices and
edges of $\HH$. We shall use the same symbols, $\DIRup$ and
$\DIRdown$, to denote these mappings. Endow $\T$ with the coordinate system given by
$(0,2)\Z + (\sqrt3,1)\Z$ and recall that $(\T^0, \T^1, \T^2)$ are the color
classes of an arbitrary proper $3$-coloring of $\T$. In the figures,
we make the choice that $(0,0) \in \T^0$ and $(0,2) \in \T^1$ so
that $\DIRup$ is the map $(a,b) \mapsto (a,b+2)$.

For a loop configuration $\omega \in \LC(H)$ and $\clr \in
\{0,1,2\}$, let $E^\clr(\omega)\subset \EH$ be the union of all
$\clr$-clusters of $\omega$, and define
\begin{align}
  \label{eq:def-bad-edges}
  \BadEdges(\omega) &:= \big(\IntEdge\gamma \cup \gamma^*\big) \setminus \big( E^0(\omega) \cup E^1(\omega)^{\DIRdown} \cup E^2(\omega)^{\DIRup} \big),
\\
  \label{eq:def-bad-edges-star}
  \BadEdgesBefore(\omega) &:= \big(\IntEdge\gamma\cup\gamma^*\big)\setminus \big( E^0(\omega) \cup E^1(\omega) \cup E^2(\omega) \big).
\end{align}
One may check that $\{ E^0(\omega),
E^1(\omega), E^2(\omega), \BadEdgesBefore(\omega) \}$ is a partition
of $\IntEdge\gamma \cup \gamma^*$ so that $\omega \cap E^0(\omega)$, $\omega \cap E^1(\omega)$, $\omega \cap E^2(\omega)$ and $\omega \cap \BadEdgesBefore(\omega)$ are pairwise disjoint loop configurations.
Finally, we define the
\emph{repair map}
\[ \shiftFunc \colon \LC(H) \to \LC(H) \]
by
  \begin{equation*}
    \label{eq:def-repair-map}
      \shift\omega :=  \big(\omega \cap E^0(\omega)\big)  \cup \big(\omega \cap E^1(\omega)\big)^{\DIRdown} \cup \big(\omega \cap E^2(\omega)\big)^{\DIRup}\cup\big(\ground^0 \cap \BadEdges(\omega)\big)  .
  \end{equation*}
The fact that the mapping is well-defined, i.e., that $\shift\omega$
is indeed in $\LC(H)$, is not completely straightforward.
However, it is indeed well-defined and, moreover,
  \[ \omega \cap E^0(\omega), \quad (\omega \cap E^1(\omega))^{\DIRdown}\cup(\omega \cap E^2(\omega))^{\DIRup} \quad\text{and}\quad\ground^0 \cap \BadEdges(\omega) \]
  are pairwise disjoint loop configurations in $\LC(H)$.

%
%

\afterpage{
\newgeometry{left=15mm,bottom=40mm,top=10mm}
\pagestyle{empty}
\begin{figure}
    \centering
    \begin{subfigure}[t]{.5\textwidth}
        \includegraphics[scale=0.74]{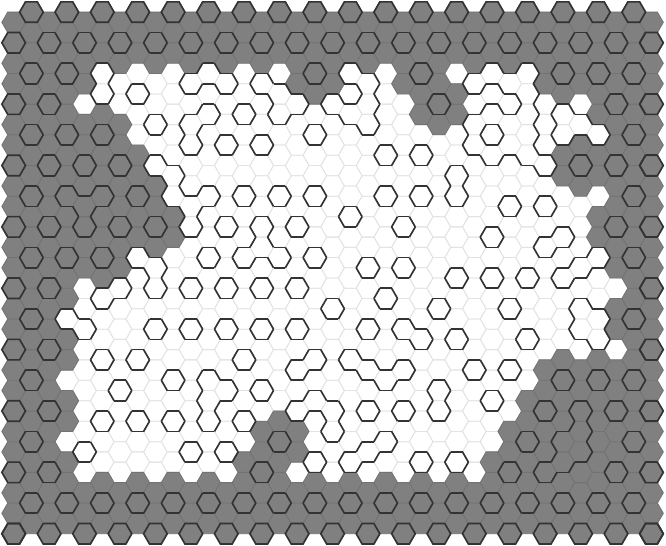}
        \caption{The breakup is found by exploring $0$-flowers from the boundary.}
        \label{fig:proof-illustration-1}
    \end{subfigure}%
    \begin{subfigure}{15pt}
        \quad
    \end{subfigure}%
    \begin{subfigure}[t]{.5\textwidth}
        \includegraphics[scale=0.74]{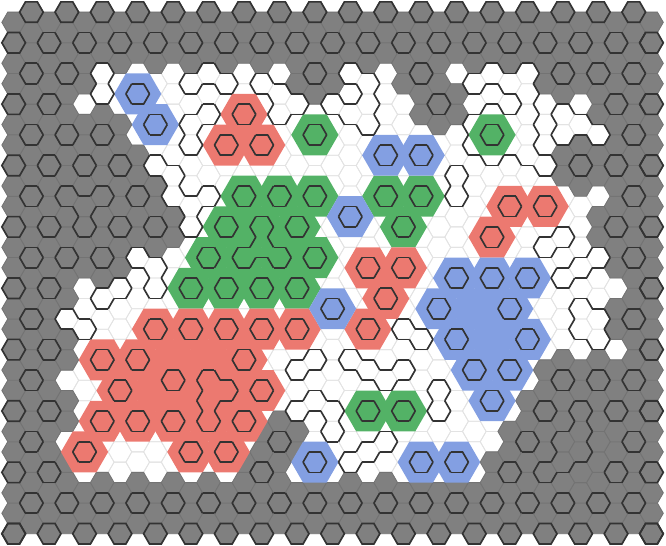}
        \caption{The clusters are found within the breakup (with 0/1/2-clusters shown in green/red/blue).}
        \label{fig:proof-illustration-2}
    \end{subfigure}
    \medbreak
    \begin{subfigure}[t]{.5\textwidth}
        \includegraphics[scale=0.74]{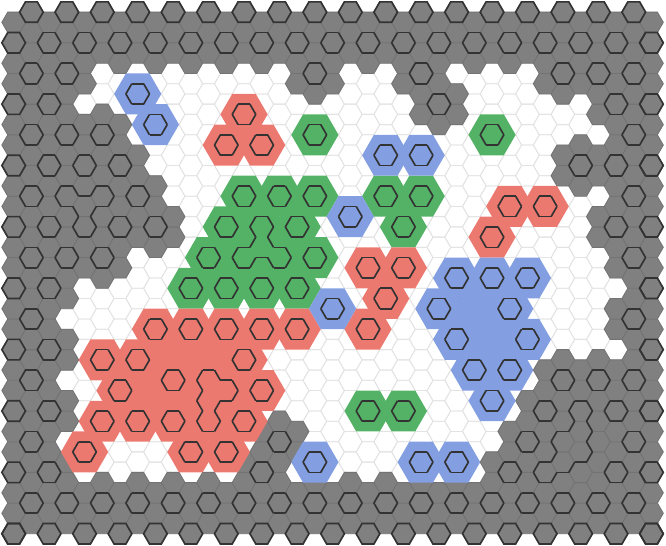}
        \caption{Bad edges are discarded.}
        \label{fig:proof-illustration-3}
    \end{subfigure}%
    \begin{subfigure}{15pt}
        \quad
    \end{subfigure}%
    \begin{subfigure}[t]{.5\textwidth}
        \includegraphics[scale=0.74]{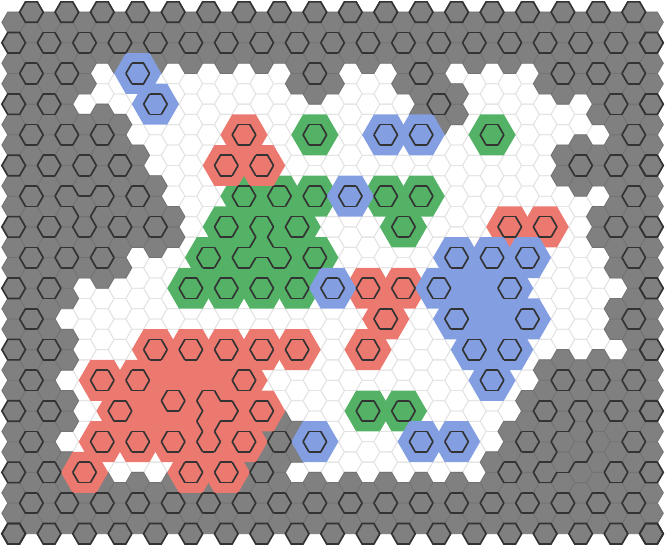}
        \caption{The clusters are shifted into the $0$-phase.}
        \label{fig:proof-illustration-4}
    \end{subfigure}
    \medbreak
    \begin{subfigure}[t]{.5\textwidth}
        \includegraphics[scale=0.74]{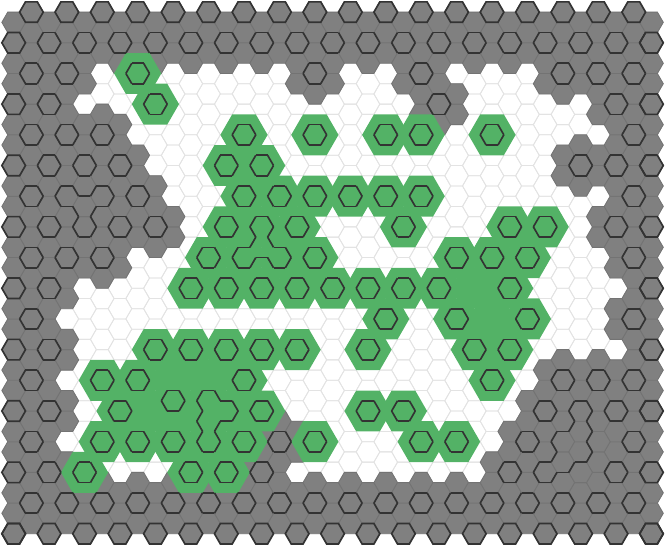}
        \caption{The empty area outside the shifted clusters is now compatible with the $0$-phase ground state.}
        \label{fig:proof-illustration-5}
    \end{subfigure}%
    \begin{subfigure}{15pt}
        \quad
    \end{subfigure}%
    \begin{subfigure}[t]{.5\textwidth}
        \includegraphics[scale=0.74]{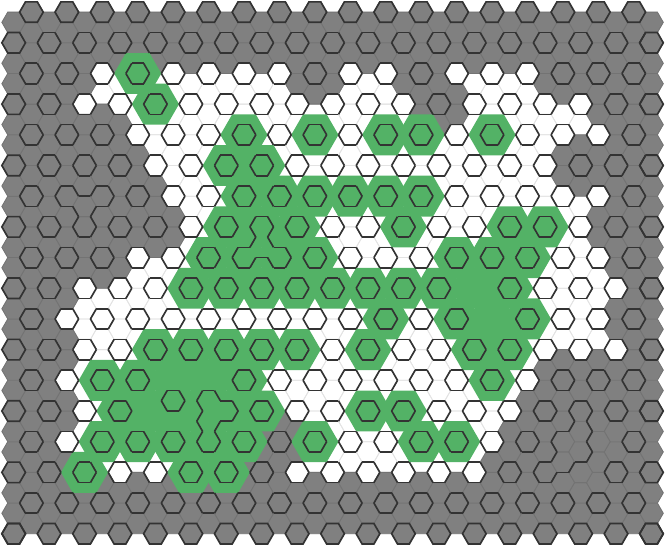}
        \caption{Trivial loops are packed in the empty area outside the shifted clusters.}
        \label{fig:proof-illustration-6}
    \end{subfigure}
    \caption{An illustration of finding the breakup and applying the repair map in it. The initial loop configuration is modified step-by-step, resulting in a loop configuration with many more loops and at least as many edges.}
    \label{fig:proof-illustration}
\end{figure}
\restoregeometry\clearpage }

\begin{proof}[{\bf Proof of Lemma~\ref{lem:prob-outer-circuit}.}]
\label{sec:proof of lemma}
Let $V$ be such that $\partial \IntVert\gamma \subset V \subset \IntVert\gamma$. We first bound the probability of the event
\[ E_V := \{ \omega \in \LC(H) ~:~ V(\omega,\gamma)=V \} .\]
To do so, we wish to apply Lemma~\ref{lem:prob-inequality-tool} to the repair map. To this end, we must estimate the gain in probability (parameter $p$ in Lemma~\ref{lem:prob-inequality-tool}) and the number of preimages of a given configuration (parameter $q$ in Lemma~\ref{lem:prob-inequality-tool}). Let $n>0$ and $x>0$. We may assume that $n \ge 1$ and $nx^6 \ge 1$ as otherwise the lemma is trivial. Then
\begin{align}
    &\Pr(\shift\omega) \ge (n \cdot \min\{x^6,1\})^{|V|/15} \cdot \Pr(\omega) &\quad\text{ for }\omega \in E_V ,\label{eq:repair_prob_gain}\\
&|E_V \cap \shiftFunc^{-1}(\omega')| \le (2 \sqrt{2})^{|V|} &\quad\text{ for }\omega' \in \LC(H) .\label{eq:repair_preimage_size}
\end{align}

The proof of~\eqref{eq:repair_prob_gain} is based on a precise understanding of the change in the number of edges $\Delta o := o(\shift\omega) - o(\omega)$ and in the number of loops $\Delta L := L(\shift\omega) - L(\omega)$. Namely, one may show (see Figure~\reffig{fig:proof-illustration}) that
\[ \Delta o = |V| - |\omega \cap \BadEdgesBefore(\omega)| \qquad\text{and}\qquad \Delta L = |V| / 6 - L(\omega \cap \BadEdgesBefore(\omega)) .\]
  Using this, one deduces that
  \[ 0 \le \Delta o \le |V| \qquad\text{and}\qquad \Delta L \ge \tfrac{|V|}{15} + \tfrac{|\Delta o|}{10} ,\]
  from which~\eqref{eq:repair_prob_gain} easily follows.

The proof of~\eqref{eq:repair_preimage_size} relies on the fact that the only loss of information incurred by the repair map is in the bad edges (see Figure~\reffig{fig:proof-illustration-3}). More precisely, the mapping $\omega \mapsto (\shift\omega, \omega \cap E(V))$ is injective on $E_V$. Thus, the size of $E_V \cap \shiftFunc^{-1}(\omega')$ is at most the number of subsets of $E(V)$. Since $|E(V)| \le 3|V|/2$, we obtain~\eqref{eq:repair_preimage_size}.

Now, using~\eqref{eq:repair_prob_gain} and~\eqref{eq:repair_preimage_size}, Lemma~\ref{lem:prob-inequality-tool} implies that
\[ \Pr(E_V) \le (2\sqrt{2})^{|V|} \cdot (n \cdot \min\{x^6,1\})^{-|V|/15} .\]
To complete the proof, we must sum over the possible choices for $V$.
For this, we use a connectivity property of $V(\omega,\gamma)$. Let $\HH^{\times}$ be the graph obtained from $\HH$ by adding an edge between each pair of opposite vertices of every hexagon, so that $\HH^{\times}$ is a $6$-regular non-planar graph.
One may show that $V(\omega,\gamma)$ is connected in $\HH^{\times}$ whenever $\partial\IntVert\gamma\subset V(\omega,\gamma)$. Thus, recalling Lemma~\ref{lem:count_connected_subgraphs}, when $n \cdot \min\{x^6,1\}$ is sufficiently large, we have
\begin{align*}
\Pr\big(\partial\IntVert\gamma\subset V(\omega,\gamma)\text{ and }|V(\omega,\gamma)|\ge k \big)
 &\le \sum_{\substack{V:~ |V| \ge k\\V\text{ connected in }\HH^\times\\\partial \IntVert\gamma \subset V \subset \IntVert\gamma}} \Pr(E_V) \\
&\le \sum_{\ell=k}^\infty C^\ell \cdot (2\sqrt{2})^\ell \cdot (n \cdot \min\{x^6,1\})^{-\ell/15} \\&\le (cn \cdot \min\{x^6,1\})^{-k/15} . \qedhere
\end{align*}
\end{proof}

%
%

\medbreak
\noindent{\bf Proofs of main theorems.}
\label{sec:exponential_decay_loop_lengths}
The proofs of the theorems for large $x$ mostly rely on the main lemma, Lemma~\ref{lem:prob-outer-circuit}.
The results for small $x$ follow via a Peierls argument, the basis of which is given by the following lemma that gives an upper bound on the probability that a given collection of loops appears in a random
loop configuration.

\begin{lemma}\label{lem:given-loops-are-unlikely}
    Let $H$ be a domain, let $n,x>0$ and let $\omega$ be sampled from the loop $O(n)$ model in domain $H$ with edge weight $x$. Then, for any $A \in \LC(H)$, we have
    \[ \Pr(A \subset \omega) \le n^{L(A)}x^{o(A)} .\]
\end{lemma}
\begin{proof}
    Consider the map
    \[ \sfT \colon \{ \omega \in \LC(H) :  A \subset \omega \} \to \LC(H) \]
    defined by
    \[ \sfT(\omega) := \omega \setminus A .\]
    Clearly, $\sfT$ is well-defined and injective.
    Moreover, since $L(\sfT(\omega)) = L(\omega) - L(A)$ and $o(\sfT(\omega)) = o(\omega) - o(A)$, we have
    \[ \Pr(\sfT(\omega)) = \Pr(\omega) \cdot n^{-L(A)} x^{-o(A)} .\]
    Hence, the statement follows from Lemma~\ref{lem:prob-inequality-tool}.
\end{proof}

Recall the notion of a loop surrounding a vertex given prior to Theorem~\ref{thm:no-large-loops}.

\begin{cor}\label{cl:no-large-loops-for-small-x}
  Let $H$ be a domain, let $n,x>0$ and let $\omega$ be sampled from the loop $O(n)$ model in domain $H$ with edge weight $x$. Then, for any vertex $u \in V(H)$ and any positive integer $k$, we have
  \[
  \Pr(\text{there exists a loop of length $k$ surrounding $u$}) \le k n (2x)^k .
  \]
  Moreover, for any $u_1,\dots,u_m \in V(H)$ and $k_1,\dots,k_m \ge 1$ with $k=k_1+\cdots+k_m$, we have
    \[
    \Pr(\forall i~\text{there exists a distinct loop of length $k_i$ passing through $u_i$}) \le (2n)^m (2x)^k ,
    \]
\end{cor}

\begin{proof}
Denote by $a_k$ the number of simple paths of length $k$ in $\HH$
starting at a given vertex. Clearly, $a_k \leq 3 \cdot 2^{k-1}$. It
is then easy to see that the number of loops of length $k$
surrounding $u$ is at most $k a_{k-1} \le k 2^k$. Thus, the result
follows by the union bound and
Lemma~\ref{lem:given-loops-are-unlikely}.

The moreover part follows similarly from Lemma~\ref{lem:given-loops-are-unlikely} by noting that there are at most $a_{k_1}\cdots a_{k_m} \le 2^{m+k}$ loop configurations $A$ consisting of exactly $k$ loops with the $i$-th loop having length $k_i$ and passing through $u_i$.
\end{proof}

The main lemma, Lemma~\ref{lem:prob-outer-circuit}, shows that for a given circuit $\gamma$ (which is contained in $\T \setminus \T^\clr$ for some $\clr$), it is unlikely that the set $V(\omega,\gamma)$ is large. The set $V(\omega,\gamma)$ specifies deviations from the ground states which are `visible' from $\gamma$, i.e., deviations which are not `hidden' inside clusters. In Theorem~\ref{thm:no-large-loops}, we claim that it is unlikely to see long loops surrounding a given vertex. Any such long loop constitutes a deviation from all ground states. Thus, the theorem would follow from the main lemma (in the main case, when $x$ is large) if the long loop was captured in $V(\omega,\gamma)$. The next lemma (whose proof we omit) bridges the gap between the main lemma and the theorem, by showing that even when a deviation is not captured by $V(\omega,\gamma)$, there is necessarily a smaller circuit $\sigma$ which captures it in $V(\omega,\sigma)$.

\begin{lemma}\label{lem:existence-of-good-outer-circuit}
Let $\omega$ be a loop configuration, let $\gamma \subset \T\setminus\T^0$ be a vacant circuit in $\omega$ and let $L$ be a non-trivial loop of $\omega$ in $\Int\gamma$. Then there exists $\clr\in\{0,1,2\}$ and a circuit $\sigma \subset \T\setminus\T^{\clr}$ such that $\Int\sigma \subset \Int\gamma$, $\sigma$ is vacant in $\omega$ and $V(L) \cup \partial \IntVert\sigma \subset V(\omega,\sigma)$.
\end{lemma}

\begin{proof}[{\bf Proof of Theorem~\ref{thm:no-large-loops}.}]
\label{sec:11}
Suppose that $n_0$ is a sufficiently large constant, let $n \ge n_0$, let $x \in (0,\infty]$, let $H$ be a domain of type 0 and let $u \in V(H)$. Let $\omega$ be sampled from the loop $O(n)$ model in domain $H$ with edge weight $x$. We shall estimate the probability that $u$ is surrounded by a non-trivial loop of length $k$. We consider two cases, depending on the relative values of $n$ and $x$.

Suppose first that $nx^6 < n^{1/50}$. Since $n \ge n_0$, we may
assume that $2x \le n^{-4/25}$ and that $kn^{-k/120} \le 1$ for all
$k>0$. By Corollary~\ref{cl:no-large-loops-for-small-x}, for every
$k \ge 7$,
\[
\begin{split}
  \Pr(\text{there exists a loop of length $k$ surrounding $u$}) & \le k n (2x)^k \le k n^{1-4k/25} \\
  & \le k n^{-k/60} \le n^{-k/120}.
\end{split}
\]

Suppose now that $nx^6 \ge n^{1/50}$. Since $n \ge n_0$, we may assume that $n \cdot \min\{x^6,1\}$ is sufficiently large for our arguments to hold.
Let $L \subset H$ be a non-trivial loop of length $k$ surrounding $u$. Note that if $L \subset \omega$ then by Lemma~\ref{lem:existence-of-good-outer-circuit}, for some $\clr\in\{0,1,2\}$, there exists a circuit $\sigma \subset \T\setminus\T^{\clr}$ such that $\Int\sigma \subset H$, $\sigma$ is vacant in $\omega$ and $V(L) \cup \partial \IntVert\sigma \subset V(\omega,\sigma)$.
Using the fact that $H$ is of type 0, the domain Markov property and Lemma~\ref{lem:prob-outer-circuit} imply that for every fixed circuit $\sigma \subset \T\setminus\T^{\clr}$ with $\Int\sigma \subset H$,
\[
\Pr\big(\sigma\text{ vacant and }V(L)\cup\partial\IntVert\sigma \subset V(\omega,\sigma)\big) \le (c n \cdot \min\{x^6,1\})^{-|V(L)\cup \partial\IntVert\sigma|/15}.
\]
Thus, denoting by $\mathcal{G}(u)$ the set of circuits $\sigma$ contained in $\T\setminus\T^{\clr}$ for some $\clr\in\{0,1,2\}$ and having $u\in\IntVert\sigma$, we obtain
\[
\begin{split}
  \Pr(L \subset \omega)
  & \le \sum_{\sigma \in \mathcal{G}(u)} (c n \cdot \min\{x^6,1\})^{-|V(L)\cup\partial\IntVert\sigma|/15}\\
  & \le \sum_{\ell = 1}^{\infty} D^{\ell} (cn \cdot \min\{x^6,1\})^{-\max\{k,\ell\}/15} \le (c n \cdot \min\{x^6,1\})^{-k/15} ,
\end{split}
\]
where we used the facts that the length of a circuit $\sigma$ such that $|\partial \IntVert\sigma|=\ell$ is at most $3\ell$, that the number of circuits $\sigma$ of length at most $3\ell$ with $u \in \IntVert{\sigma}$ is bounded by $D^\ell$ for some sufficiently large constant $D$, and in the last inequality we used the assumption that $n \cdot \min\{x^6,1\}$ is sufficiently large.
Since the number of loops of length $k$ surrounding a given vertex is smaller than $k 2^k$, our assumptions that $nx^6 \ge n^{1/50}$ and $n \ge n_0$ yield
\[
  \Pr(\text{there exists a loop of length $k$ surrounding $u$})
   \le k 2^k (c n^{1/50})^{-k/15} \le n^{-k/800}. \qedhere
\]
\end{proof}

%
%

\begin{proof}[{\bf Proof of Theorem~\ref{thm:large-n-transition}.}]
Let $n>0$, let $x \in (0,\infty]$, let $H$ be a domain of type 0 and let $u \in V(H)$.
Let $\omega$ be sampled from the loop $O(n)$ model in domain $H$ with edge weight $x$.

We first prove the upper bound on the probability that $u$ is loop-connected to distance~$k$. For this, we may assume that $x$ and $nx^6$ are sufficiently small as the bound is trivial otherwise. Denote $u_0 := u$ and observe that if $u$ is loop-connected to some vertex at distance $k$ from~$u$, then there exist integers $m \ge 1$, $\ell_1,\dots,\ell_m \ge 6$ and vertices $u_1,\dots,u_m \in V(H)$ such that $k \le \ell := \ell_1+\cdots+\ell_m$ and, for all $1 \le i \le m$, $\dist(u_i,u_{i-1}) \le \ell_i$ and $u_i$ belongs to a distinct loop of $\omega$ of length $\ell_i$. Thus, summing over the possible choices (for brevity, we omit the conditions on $\ell_i$ and $u_i$ in the sum below) and applying Corollary~\ref{cl:no-large-loops-for-small-x}, we obtain
\begin{align*}
    \Pr\Big(\substack{\text{\footnotesize{$u$ is loop-connected}}\\\text{\footnotesize{to distance $k$}}}\Big)
        &\le \sum_{\substack{\ell \ge k\\\ell/6 \ge m \ge 1}} \sum_{\substack{\ell_1,\dots,\ell_m\\u_1,\dots,u_m}} \Pr(\forall i~u_i\text{\small{ belongs to a distinct loop of length }}\ell_i) \\
        &\le \sum_{\ell \ge k} \ell \cdot 2^\ell \cdot 3^\ell \cdot (2n+1)^{\ell/6} (2x)^\ell \\
        &\le \sum_{\ell \ge k} (C(n+1)x^6)^{\ell/6}  \le C(C(n+1)x^6)^{k/6} .
\end{align*}

We now prove the lower bound on the probability that $u$ is ground-connected to the boundary of $H$, i.e., that $u$ and $v$ are ground-connected for some $v \in \partial V(H)$.
For this, we may assume that both $n$ and $nx^6$ are sufficiently large as the bound is trivial otherwise.
Assume that $u$ is not ground-connected to the boundary of $H$.
Let $A(\omega)$ be the set of vertices of $\HH$ belonging to loops in $\omega \cap \ground^0$ and let $B(\omega)$ be the unique infinite connected component of $A(\omega) \cup (V(\HH) \setminus V(H))$.
Note that $u \notin B(\omega)$ by assumption and define the \emph{breakup} $\breakup$ to be the connected component of $\HH \setminus B(\omega)$ containing $u$.
One may check that the subgraph induced by $\breakup$ is a domain of type 0, and that the enclosing circuit $\Gamma$ (i.e., the circuit satisfying $\breakup = \IntVert\Gamma$, which exists by Fact~\ref{fact:circuit-domain-bijection}) is vacant in $\omega$ and is contained in $\T\setminus \T^0$.
Furthermore, we have $\partial \IntVert\Gamma \subset V(\omega,\Gamma)$. Indeed, this follows as $\Gamma$ is vacant in $\omega$ and, by the definition of $B(\omega)$, no vertex of $\partial\IntVert{\Gamma}$ belongs to a trivial loop surrounding a hexagon in $\T^0$.
Thus, denoting by $\mathcal{G}$ the set of circuits $\gamma
\subset \T \setminus \T^0$ having $u\in\IntVert\gamma$,
Lemma~\ref{lem:prob-outer-circuit} implies that
\begin{align*}
\Pr\Big(\substack{\text{\footnotesize{$u$ is not ground-connected}}\\\text{\footnotesize{to the boundary of $H$}}}\Big)
 & \le \sum_{\gamma \in \mathcal{G}}\Pr\big(\gamma\text{ vacant and }\partial\IntVert\gamma\subset V(\omega,\gamma)\big)\\
 & \le \sum_{\gamma \in \mathcal{G}} (c n \cdot \min\{x^6,1\})^{-|\partial\IntVert\gamma|/15} \\
 & \le \sum_{k \ge 1} D^{k}(c n \cdot \min\{x^6,1\})^{-k/15}
\le C(n \cdot \min\{x^6,1\})^{-c},
\end{align*}
where in the third inequality we used the facts that the length of a circuit $\gamma$ such that $|\partial \IntVert\gamma|=k$ is at most $3k$, and that the number of circuits of length at most $3k$ surrounding~$u$ is bounded by $D^k$ for some sufficiently large constant $D$.
\end{proof}

%
%

\bibliographystyle{amsplain}
\bibliography{biblicomplete}

\end{document}